\documentclass[11pt]{article}
\usepackage[T1]{fontenc}
\usepackage{lmodern}
\usepackage[margin=1in]{geometry}
\usepackage{booktabs,tabularx,array,enumitem,microtype}
\usepackage{xcolor}

\usepackage{amsmath,amssymb,amsthm,mathtools}
\usepackage{thmtools}

\usepackage[colorlinks=true,linkcolor=blue!55!black,
            citecolor=blue!55!black,urlcolor=blue!60!black]{hyperref}

\newtheorem{theorem}{Theorem}[section]

\newtheorem{proposition}[theorem]{Proposition}
\newtheorem{lemma}[theorem]{Lemma}
\newtheorem{corollary}[theorem]{Corollary}
\theoremstyle{definition}
\newtheorem{definition}[theorem]{Definition}
\theoremstyle{remark}
\newtheorem{remark}[theorem]{Remark}

\allowdisplaybreaks
\newcommand{\E}{\mathbb E}
\newcommand{\Prb}{\mathbb P}
\newcommand{\R}{\mathbb R}
\newcommand{\vol}{\operatorname{vol}}
\newcommand{\Lip}{\operatorname{Lip}}
\newcommand{\cP}{\mathcal P}
\newcommand{\cQ}{\mathcal Q}
\newcommand{\one}{\mathbf 1}
\newcommand{\Unif}{\operatorname{Unif}}
\newcommand{\dtv}{d_{\mathrm{TV}}}
\newcommand{\mix}{\mathrm{mix}}
\newcommand{\Id}{\mathsf I}
\newcommand{\HRzero}{\mathsf P_{\mathrm{HR}}}
\newcommand{\HR}{\mathsf P_{\mathrm{HR}}}
\newcommand{\rK}{r_K}
\DeclarePairedDelimiter{\abs}{\lvert}{\rvert}

\title{Hit-and-Run Mixes as Fast as the Ball Walk}
\author{Ruizhe Zhang\thanks{\texttt{rzzhang@purdue.edu}}\\Purdue University}
\date{}

\begin{document}
\maketitle

\begin{abstract}
Let $K\subset\R^n$ be an isotropic convex body.  We prove that the  hit-and-run walk, started from any $M$-warm distribution, reaches total-variation distance $\varepsilon$ from the uniform distribution on $K$ in $O\!\left(n^2\psi_n^{-2}\log^3(M/\varepsilon)\right)$ steps, where  $\psi_n^{-1}$ is the Kannan--Lov\'asz--Simonovits (KLS) constant.  Up to logarithmic factors, this matches the best-known warm-start mixing time for the ball walk.  Chen and Eldan [Discrete Comput. Geom. 2026] obtained the same $n^2\psi_n^{-2}$ dependence for hit-and-run, but with polynomial dependence on $M/\varepsilon$.  Our result improves that polynomial dependence to a polylogarithmic one, fully resolving their open question about warm-start mixing of hit-and-run in isotropic convex bodies.
\end{abstract}

\section{Introduction}\label{sec:introduction}
Sampling nearly uniformly from a convex body is a fundamental algorithmic
primitive, with applications to volume computation, Bayesian inference, convex optimization,
and high-dimensional statistics
\cite{DFK91,albert1993bayesian,KLS97,bertsimas2004solving,LV,kalai2006simulated,Holmes2006BayesianAV,CV18,JLLV26,KookVempalaDiffusion}.  The \emph{hit-and-run walk}, introduced by
Smith~\cite{Smith}, is one of the canonical Markov chains for this task. From an interior point $x\in K$, it chooses a uniformly random direction and samples the next point uniformly from the chord of $K$ through $x$ in that direction (Fig.~\ref{fig:hit-and-run}).  This update defines a reversible Markov chain whose stationary distribution is uniform on $K$.

Hit-and-run has several features that distinguish it from the ball walk.
The ball walk proposes a point from a Euclidean ball of a prescribed radius
centered at the current state and rejects the proposal if it lies outside
\(K\).  Its analysis and implementation therefore require a suitable choice
of step size.  Hit-and-run has no such proposal-radius parameter: it samples
from the entire chord, so it can make long moves even when the current
point lies near a narrow corner.  These properties make the hit-and-run walk both algorithmically attractive and technically subtle.

The mixing of hit-and-run has been studied for several decades. Suppose that \(K\subseteq\R^n\) contains a Euclidean ball of radius \(r\) and is contained in a Euclidean ball of radius \(R\).  Lov\'asz showed that, from a \(2\)-warm start, hit-and-run reaches total-variation distance \(\varepsilon\) from the target uniform distribution in $O(n^2(R/r)^2\varepsilon^{-2}\log(2/\varepsilon))$ steps~\cite{LovaszHR}. We say an initial distribution is \emph{$M$-warm} if its density is at most $M$ times the target density. Lov\'asz and Vempala improved the bound to $O(n^2(R/r)^2\log(M/\varepsilon))$ from an \(M\)-warm start~\cite{LV}. They also proved polynomial mixing time from any single interior point, without first constructing a warm initial distribution.  Separately, their second-moment refinement measures the typical distance of a uniform point from the centroid.  For an isotropic convex body $K$, whose uniform distribution has mean zero and covariance $I_n$, this gives $O(n^3\log^3(M/\varepsilon))$ warm-start mixing time.

For comparison, the Kannan--Lov\'asz--Simonovits analysis of the ball walk gives $\widetilde{O}_{M/\varepsilon}(n^2\psi_n^{-2})$~\footnote{We use $\widetilde{O}_a(f)$ to hide logarithmic factors in $f$ and $a$.} mixing time from an $M$-warm start on an isotropic convex body~\cite{KLS97,JLLV26}.  Here, \(\psi_n\) is the smallest Cheeger constant among isotropic log-concave probability measures on \(\R^n\), and \(\psi_n^{-1}\) is commonly called the KLS constant~\cite{KLS95}.  The KLS
conjecture predicts that \(\psi_n=\Theta(1)\).  Although the conjecture
remains open, a sequence of recent breakthroughs has established nearly constant lower bounds on $\psi_n$~\cite{ChenKLS,KlartagLehecKLS,KlartagKLS,LetwinKLS}.  With the current estimates, the ball walk therefore already mixes in \(\widetilde{O}_{M/\varepsilon}(n^2)\) steps.

Whether hit-and-run could attain the same dimension dependence remained open for many years.  Chen and Eldan answered this question using localization schemes~\cite {CE}. Provided $n=\Omega(\log(M/\varepsilon))$, they proved that the hit-and-run walk in an isotropic convex body mixes in
\begin{align}\label{eq:mixing_CE}
 O\!\left(
 n^2\psi_n^{-2}
 \left(\frac{M}{\varepsilon}\right)^{11}
 \log^5\!\left(\frac{M}{\varepsilon}\right)
 \right)
\end{align}
steps from an \(M\)-warm start.  This was the first hit-and-run mixing time bound with leading dependence $n^2/\psi_n^2$, but its cost grows polynomially with both warmness and inverse accuracy.  Chen and Eldan therefore asked whether the hit-and-run walk on an isotropic convex body admits an $O_{\varepsilon}\!\left(n^2\operatorname{polylog} M\right)$ warm-start mixing bound.

\begin{figure}
    \centering
    \includegraphics[width=0.95\linewidth]{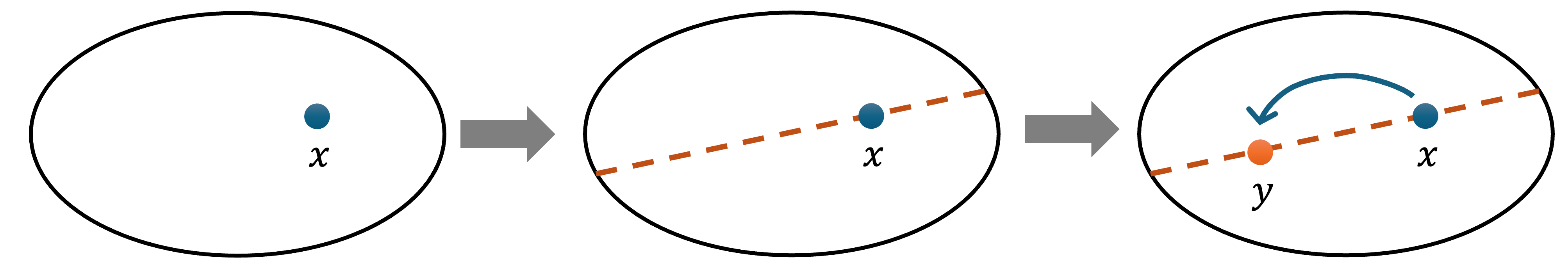}
    \caption{One-step update of the hit-and-run walk}
    \label{fig:hit-and-run}
\end{figure}

\subsection{Main result}
In this paper, we answer this open question in~\cite{CE} affirmatively.

\begin{theorem}[Main theorem, informal version of Theorem~\ref{thm:main_formal}]\label{thm:main}
Let $n\ge2$, $K\subset\R^n$ be an isotropic convex body, and $\mu=\Unif(K)$ be the uniform distribution over $K$. For every $M$-warm initial distribution and every $0<\varepsilon<1/2$, the hit-and-run walk has mixing time
\begin{align}\label{eq:main_informal_bound}
    \tau_{\rm mix}=O\!\left(n^2\psi_n^{-2}\log^3\left(\frac{M}{\varepsilon}\right)\right).
\end{align}
\end{theorem}
Using the best current bound on the KLS constant $\psi_n^{-1}=O(\log^{1/4}(n))$ due to Letwin~\cite{LetwinKLS}, \eqref{eq:main_informal_bound} becomes
\begin{align*}
    \tau_{\rm mix} = O\!\left(n^2\log^{\frac{1}{2}}(n) \log^3\left(\frac{M}{\varepsilon}\right)\right).
\end{align*}
We remark that Theorem~\ref{thm:main} applies only to warm starts.  A point mass is singular with respect to \(\mu\) and therefore is not \(M\)-warm for any finite \(M\); our result does not replace the Lov\'asz--Vempala guarantee from an arbitrary interior point~\cite{LV}.  Within the warm-start regime, however, it places hit-and-run on the same asymptotic scaling as the ball walk: up to logarithmic factors, the two bounds have the same dependence on dimension, the KLS constant, warmness, and accuracy~\cite{KLS97}. Compared with Chen--Eldan~\cite{CE}, it replaces the polynomial dependence on \(M/\varepsilon\) by a polylogarithmic one and removes the requirement \(n\gtrsim\log(M/\varepsilon)\) \footnote{We use $\gtrsim$, $\lesssim$, $\asymp$ to hide the universal constants.}.

We also extend our approach to hit-and-run walk for a \emph{log-affine distribution} over a convex body $d\nu_b(x)\propto e^{-\langle b,x\rangle}\one_K(x)\,dx$.  Suppose $\nu_b$ is $C_0$-near-isotropic, i.e., $C_0^{-1}I_n \preceq \operatorname{Cov}_{\nu_b}(X)\preceq C_0I_n$, then the warm-start mixing time of the hit-and-run walk is 
\begin{align}
    \tau_{\mix}=O\left(C_0^{2}n^2\psi_n^{-2}\log^3\left(\frac{M}{\varepsilon}\right) \right), 
\end{align}
see Theorem~\ref{thm:log-affine-main}. This extension is particularly relevant to annealing-based volume estimation, whose intermediate distributions are log-affine~\cite{LVVolume}.

\subsection{Technical overview}\label{sec:technical-overview}
We prove Theorem~\ref{thm:main} using the conductance method. In Section~\ref{sec:tech_ov_standard}, we review the standard route of using local overlap and an isoperimetric inequality to obtain an $s$-conductance bound and, in turn, a warm-start mixing-time bound.  In Section~\ref{sec:tech_ov_previous}, we explain where the previous analyses~\cite{LV,CE} lose factors. Finally, in Section~\ref{sec:tech_ov_ours}, we introduce the key ideas of our proof.
\subsubsection{The standard conductance analysis for hit-and-run}\label{sec:tech_ov_standard}
Let $X\sim\mu$ be at stationarity and let $Y$ be obtained from $X$ by one hit-and-run step.  For a measurable set $S\subseteq K$, the probability
\[
 \cQ_{\HR}(S,S^c):=\Prb\{X\in S,\ Y\notin S\}
\]
is the stationary, or ergodic, flow across the cut $(S,S^c)$.  A set whose
flow is small compared with its stationary mass is a bottleneck.

For an $M$-warm start, it is enough to rule out bottlenecks above a small mass threshold.  The $s$-conductance $\Phi_s$, defined in~\eqref{eq:s-conductance}, tests only sets of stationary mass greater than $s$.  This is natural because an $M$-warm initial distribution assigns at most $Ms$ probability to any set of $\mu$-mass at most $s$.  Applying the standard Lov\'asz--Simonovits $s$-conductance bound~\cite{LS} with $s=\varepsilon/(2M)$ gives
\[
 \tau_{\mix}
 \lesssim
 \Phi_s^{-2}\log(2M/\varepsilon).
\]
Thus Theorem~\ref{thm:main} reduces to proving a conductance lower bound
\begin{equation}\label{eq:intro-conductance}
 \Phi_s
 \gtrsim
 \frac{\psi_n}
 {n\bigl(1+\log(2/s)\bigr)}.
\end{equation}

The geometric conductance argument combines local overlap with an isoperimetric separator estimate~\cite{LS,LV,VempalaSurvey}.  More concretely, consider a cut $(S,S^c)$ with mass $p:=\mu(S)\in(s,1/2]$ and flow $Q:=\cQ_{\HR}(S,S^c)$. If $Q$ is already a fixed positive fraction of $p$, then the desired conductance lower bound is immediate.  The only nontrivial case is therefore when $Q\ll p$.  Let $\alpha_0>0$ be a small constant, and define the low-escape cores (Fig.~\ref{fig:cores}):
\begin{align}\label{eq:intro_def_S}
 S_-:=
 \left\{
 x\in S:
 \HR(x,S^c)<\frac{\alpha_0}{4}
 \right\},
 \qquad
 S_+:=
 \left\{
 y\in S^c:
 \HR(y,S)<\frac{\alpha_0}{4}
 \right\}.
\end{align}
Points in $S_-$ are unlikely to leave $S$ in one step, while points in
$S_+$ are unlikely to enter $S$.  Markov's inequality and reversibility give
\[
 \mu(S\setminus S_-)
 \le\frac{4Q}{\alpha_0},
 \qquad
 \mu(S^c\setminus S_+)
 \le\frac{4Q}{\alpha_0}.
\]
Thus, when $Q$ is small compared with $p$, both cores have mass comparable with $p$, whereas the remaining region has mass
\begin{equation}\label{eq:intro-separator-upper}
 \mu(K\backslash (S_-\cup S_+))\le\frac{8Q}{\alpha_0}.
\end{equation}

We next show that $K\backslash (S_-\cup S_+)$ must separate $S_-$ and $S_+$ by a nontrivial width in the geometry relevant to hit-and-run. Ordinary Euclidean distance is not the appropriate measure of proximity near the physical boundary of $K$.  A small Euclidean displacement deep inside the body changes the surrounding chord geometry very little, while the same displacement near $\partial K$ may substantially change the available directions and chord endpoints.  Following~\cite{LovaszHR,LV}, define the local conductance radius
\begin{equation}\label{eq:intro-local-radius}
 r_K(x)
 :=
 \sup\left\{
 r>0:
 \frac{\vol(K\cap B(x,r))}
 {\vol(B(0,r))}
 \ge\frac{63}{64}
 \right\}.
\end{equation}
Thus, $r_K(x)$ is the largest scale at which a Euclidean ball centered at $x$ lies almost entirely inside $K$.  It is large in regions with ample local room and tends to zero as $x$ approaches $\partial K$. The corresponding intrinsic metric measures displacement in units of this
local scale:
\[
 d_{r_K}(x,y)
 :=
 \inf_\gamma
 \int_0^1
 \frac{\|\gamma'(t)\|_2}{r_K(\gamma(t))}\,dt.
\]
Infinitesimally, a Euclidean displacement of length $ds$ at $x$ has intrinsic cost $ds/r_K(x)$.  Hence a fixed intrinsic distance corresponds to a shorter Euclidean displacement near the boundary and a longer one in the interior. The local-overlap analysis of
\cite{LovaszHR,LV} gives universal constants $c_0,\alpha_0>0$ such that
\begin{equation}\label{eq:intro-local-overlap}
 d_{r_K}(x,y)<\frac{c_0}{\sqrt n}
 \quad\Longrightarrow\quad
 \dtv\bigl(\HR(x,\cdot),\HR(y,\cdot)\bigr)
 \le1-\alpha_0.
\end{equation}
On the other hand, if $x\in S_-$ and $y\in S_+$, then by~\eqref{eq:intro_def_S} and using $S$ as the test set for the total variation distance, we obtain that
\[
 \dtv\bigl(\HR(x,\cdot),\HR(y,\cdot)\bigr)  >  1-\frac{\alpha_0}{2}  >
 1-\alpha_0.
\]
The contrapositive of \eqref{eq:intro-local-overlap} therefore implies
\begin{equation}\label{eq:intro-core-separation}
 d_{r_K}(S_-,S_+)
 \ge\frac{c_0}{\sqrt n}.
\end{equation}
Thus local overlap supplies the ``width'' of the separator: every path from $S_-$ to $S_+$ has intrinsic length at least $c_0/\sqrt n$.
\begin{figure}
    \centering
    \includegraphics[width=0.45\linewidth]{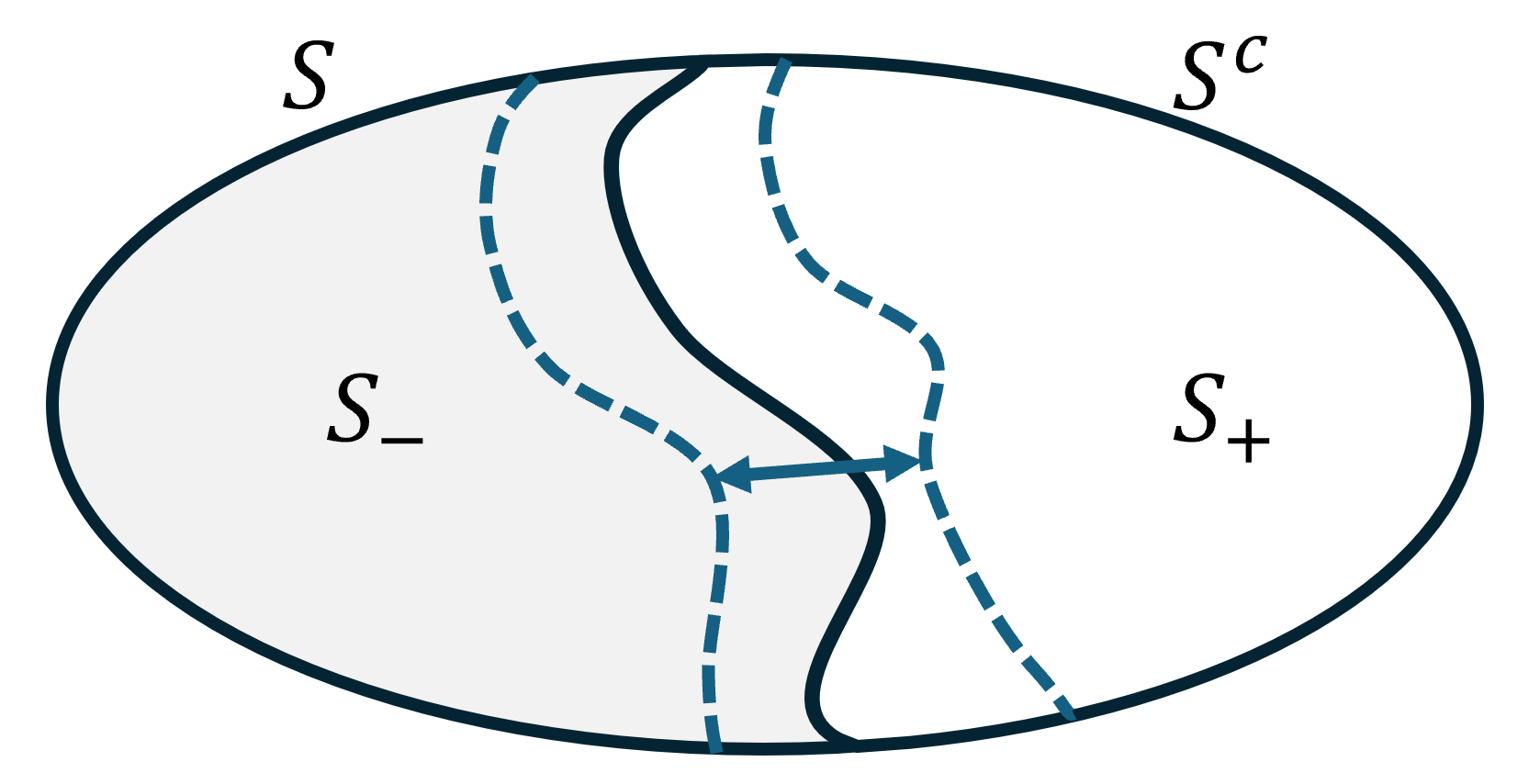}
    \caption{Low-escape cores $S_-$ and $S_+$ can be identified on both sides of the cut $(S,S^c)$. They have relatively large mass and are well-separated in the intrinsic metric.}
    \label{fig:cores}
\end{figure}

It remains to lower-bound the stationary mass of the region separating the two cores.  The appropriate boundary area is weighted by the local scale $r_K$.  To see this, suppose first that $E$ has a smooth boundary.  An intrinsic boundary layer of small width $\tau$ has Euclidean thickness approximately $\tau r_K(x)$ near a boundary point $x$ (Fig.~\ref{fig:intrinisic_boundary}).  By expressing its normalized volume to the first order in $\tau$, we can see that the quantity
\begin{equation}\label{eq:intro-smooth-weighted-perimeter}
 \frac1{\vol(K)}
 \int_{\partial E\cap\operatorname{int}(K)}
 r_K(x)\,d\mathcal H^{n-1}(x)
\end{equation}
is the stationary mass created per unit intrinsic thickening of the boundary. Conductance involves arbitrary measurable cuts, for which the surface
integral in \eqref{eq:intro-smooth-weighted-perimeter} need not be defined.
We therefore use its lower-semicontinuous BV relaxation
$\cP_{r_K}(E)$~\cite{AFP}.  For regular sets it agrees with
\eqref{eq:intro-smooth-weighted-perimeter}; for general measurable sets it
retains the coarea and slicing properties needed in the proof.
\begin{figure}
    \centering
    \includegraphics[width=0.45\linewidth]{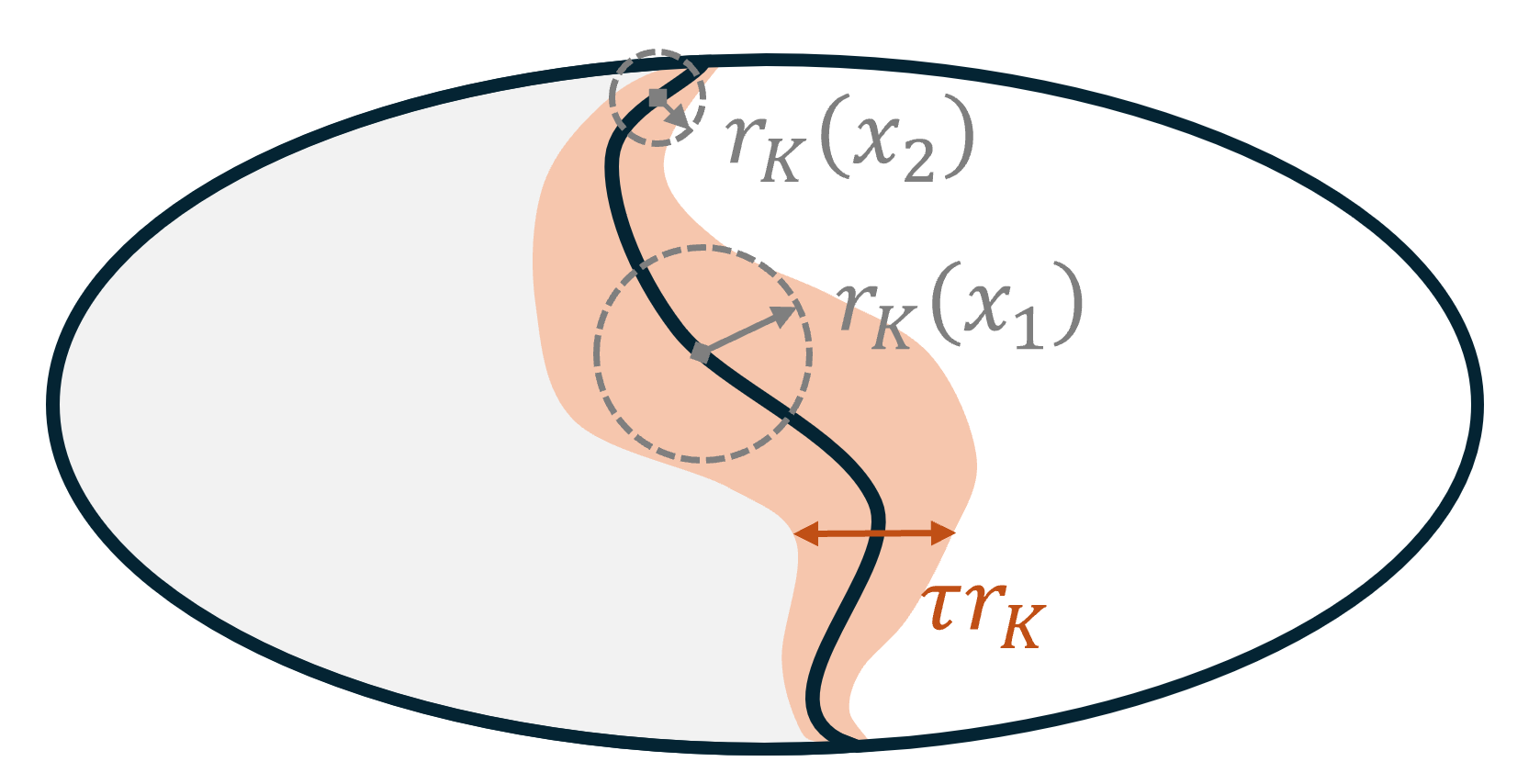}
    \caption{An intrinsic boundary layer of width $\tau$ (the orange-shaded area). Its Euclidean width $\tau \rK(x)$ is large for $x$ inside the convex body and small for $x$ close to the boundary.}
    \label{fig:intrinisic_boundary}
\end{figure}

Suppose that every measurable set $E$ whose smaller side has mass at least $s/2$ satisfies the following isoperimetric inequality:
\begin{equation}\label{eq:intro-weighted-expansion-abstract}
 \cP_{r_K}(E)
 \ge
 \gamma_s\min\{\mu(E),1-\mu(E)\},
 \qquad
 \gamma_s
 \asymp
 \frac{\psi_n}
 {\sqrt n\bigl(1+\log(2/s)\bigr)}.
\end{equation}
By~\eqref{eq:intro-core-separation}, as one expands $S_-$ in the intrinsic metric until reaching $S_+$, every intermediate set
\begin{align*}
    E_t:=\{x\in K:d_{\rK}(S_-,x)<t\},\qquad 0<t<\frac{c_0}{\sqrt{n}}
\end{align*}
has both sides of mass comparable with $p$.  Therefore, \eqref{eq:intro-weighted-expansion-abstract} gives weighted boundary at least $\Omega(\gamma_s p)$ throughout an intrinsic interval of length $\Omega(1/\sqrt n)$.  The weighted coarea inequality then yields
\[
 \mu\bigl(K\setminus(S_-\cup S_+)\bigr)
 \gtrsim
 \frac{\gamma_s p}{\sqrt n}.
\]
In words, the separator mass is bounded below by its intrinsic width times its weighted boundary area. Combining with the upper bound~\eqref{eq:intro-separator-upper} gives
\[
 \frac{Q}{p}
 \gtrsim
 \frac{\psi_n}
 {n\bigl(1+\log(2/s)\bigr)},
\]
which implies the desired $s$-conductance estimate~\eqref{eq:intro-conductance}.  

Thus the standard overlap-and-isoperimetry argument reduces the main theorem to establishing the weighted isoperimetric inequality~\eqref{eq:intro-weighted-expansion-abstract}.

\subsubsection{Where previous approaches lose factors}\label{sec:tech_ov_previous}
The local-overlap argument above already identifies the correct position-dependent scale. The loss in the Lov\'asz--Vempala analysis occurs in the separator estimate \cite{LV}.  Their localization argument reduces the problem to a single interval, and the resulting one-dimensional estimate must hold for any interval selected.  They therefore normalize the local radius by a \emph{worst-case} chord-length bound $D$, using a weight of the form
\[
 h(x)\asymp\frac{r_K(x)}{D\sqrt n}.
\]
This yields a conductance estimate of order $1/(nD)$.  A second-moment argument makes the effective value of $D$ about $\sqrt n$ in isotropic position, leading to the classical $n^{-3/2}$ dependence.  

Chen--Eldan obtain the sharper $\psi_n/n$ dependence on dimension through stochastic localization \cite{CE}.  Their argument first restricts to a core where $r_K(x)$ is bounded below by a common radius.  To make the discarded boundary layer negligible for an $s$-conductance estimate, this common radius must shrink with $s$; the resulting fixed-radius flow estimate therefore loses a polynomial factor in $s$.  A further loss comes from requiring the localization process to retain substantial mass on both sides of the cut. Tracking these effects in~\cite{CE} gives the factor $s^{11/2}$ in their conductance estimate and hence polynomial dependence on $M/\varepsilon$ in the warm-start mixing time.

These analyses point to the missing ingredient: an isoperimetric inequality that uses the typical conditional spread of the localized mass, rather than the worst-case length of one selected chord, while remaining effective for small sets. The key new idea is to retain and average over a balanced family of needles instead of reducing the proof to one witness needle.

\subsubsection{Our approach: average over balanced needles}\label{sec:tech_ov_ours}
We now make this idea precise.  A needle is a line interval $I\subseteq K$ equipped with a conditional log-concave probability measure $\mu_I$.  Let $t_I:I\to\R$ be a unit-speed coordinate along the needle, and define
\[
 m_I:=\int_I r_K\,d\mu_I,
 \qquad
 \sigma_I^2:=\operatorname{Var}_{\mu_I}(t_I),
\]
where $m_I$ is the mean local step scale on the needle, and $\sigma_I$ measures the typical spread of its conditional mass along the needle. Note that $\sigma_I$ can be much smaller than the full length of $I$ when most of the mass is concentrated in a short portion of a long chord. Our one-dimensional concave-weight inequality (Theorem~\ref{thm:one-dimensional}) shows that, for any cut $A$ with conditional mass $\mu_I(A)=p$,
\begin{equation}\label{eq:intro-needle-profile}
 \cP_{r_K,\mu_I}(A\cap I)
 \gtrsim
 \frac{m_I}{\sigma_I}
 \frac{p}{1+\log(1/p)}.
\end{equation}
After slicing and averaging \eqref{eq:intro-needle-profile}, the high-dimensional problem reduces to lower-bounding
\[
 \int\frac{m_I}{\sigma_I}\,d\pi(I),
\]
where $\pi$ is the probability measure over the needles.  

To control this average using global information, the decomposition must have two additional properties:
\begin{enumerate}
    \item Every needle must see the same cut mass,
    \[
     \mu_I(A)=\mu(A)=p,
    \]
    so that the factor $p/(1+\log(1/p))$ in~\eqref{eq:intro-needle-profile} is uniform across the decomposition.
    \item The coordinates $t_I$ must arise from one global Lipschitz function $u$; then
    \[
     \sigma_I^2=\operatorname{Var}_{\mu_I}(u)
    \]
    and the law of total variance allows the conditional spreads to be bounded by the global Poincar\'e inequality.
\end{enumerate}
Klartag's guided localization provides exactly such a decomposition~\cite{KlartagNeedles}: it produces
\[
 \mu=\int\mu_I\,d\pi(I),
 \qquad
 \mu_I(A)=p
\]
for almost every needle, together with a common $1$-Lipschitz guiding function $u$ that is a unit-speed coordinate on every needle. Intuitively, this decomposition comes from optimal transport. Consider a function $f:=\one_A-p$ with $p=\mu(A)$. The positive and negative parts of $f$ have equal mass.  We may therefore regard them as an excess of mass inside $A$ and an equal deficit inside $A^c$, and transport the former to the latter using Euclidean distance as the cost. Then, the dual potential for this transport problem is a $1$-Lipschitz function $u$.  Optimal transport takes place along line segments on which the Lipschitz inequality for $u$ is saturated, and these transport rays form the needles. Moreover, each ray transports its own excess mass to its own deficit mass; formally,
\[
 \int_I(\one_A-p)\,d\mu_I=0, 
\]
and therefore, $\mu_I(A)=p$ for almost every needle.

Then, it is straightforward to control the averages of $m_I$ and $\sigma_I$ under this decomposition and complete the proof.  Disintegration and the mean local-radius estimate in~\cite{LV} give
\[
 \int m_I\,d\pi(I)  =  \int_K r_K\,d\mu  \gtrsim  \frac1{\sqrt n}.
\]
And the common guide function and the Poincar\'e inequality imply
\[
 \int\sigma_I^2\,d\pi(I)  \le  \operatorname{Var}_\mu(u)  \le  C_{\sf PI}(\mu)\lesssim \psi_n^{-2},
\]
where we use the relation between Poincar\'e and Cheeger constants~\cite{Cheeger,buser1982note,ledoux2004spectral,KLS95} in the last step.
Finally, conditional Jensen and Berwald's inequality
\cite{Berwald,LangharstPutterman} give
\[
 \int m_I^2\,d\pi(I)
 \le
 \int_K r_K^2\,d\mu
 \lesssim
 \left(\int_K r_K\,d\mu\right)^2.
\]
Consequently, Cauchy--Schwarz yields
\[
 \int\frac{m_I}{\sigma_I}\,d\pi(I)\ge
 \frac{\left(\int m_I\,d\pi\right)^2}
      {\int m_I\sigma_I\,d\pi}\ge
 \frac{\left(\int m_I\,d\pi\right)^2}
      {\left(\int m_I^2\,d\pi\right)^{1/2}
       \left(\int\sigma_I^2\,d\pi\right)^{1/2}}
 \gtrsim
 \frac{(\int \rK\,d\mu)^2}
      {\psi_n^{-1}\int \rK\,d\mu}
 \gtrsim
 \frac{\psi_n}{\sqrt n}.
\]
Averaging the one-dimensional inequalities~\eqref{eq:intro-needle-profile} gives the desired weighted isoperimetric inequality~\eqref{eq:intro-weighted-expansion-abstract}:
\[
 \cP_{r_K}(A)
 \gtrsim\frac{p}{1+\log(1/p)}
 \int\frac{m_I}{\sigma_I}\,d\pi(I) \gtrsim 
 \frac{\psi_n}{\sqrt n}
 \frac{p}{1+\log(1/p)}.
\]

\paragraph{Extension to log-affine targets.}
The balanced-needle argument is not specific to the uniform distribution. For the log-affine target
\[
 d\nu_b(x)\propto e^{-\langle b,x\rangle}\one_K(x)\,dx,
\]
we replace the local radius $r_K(x)$ by the density-sensitive radius $\rho_b(x)$ as in~\cite{LV}.  Roughly speaking, $\rho_b(x)$ is the largest scale on which most nearby points both remain inside $K$ and have density comparable with the density at $x$.  It therefore captures the two reasons that a local move may become unstable: proximity to the boundary and rapid variation of the target density.  Using the concavity and mean-radius estimates for $\rho_b$, the same guided-localization and Poincar\'e-averaging argument gives the corresponding weighted isoperimetric inequality.

The only substantially new step is local overlap.  For the uniform target, hit-and-run samples uniformly from each chord.  For a log-affine target, the chord conditional is a truncated exponential whose normalizing constant depends on the chord, so the overlap lemma does not apply directly. We resolve this issue by constructing an auxiliary multiscale Metropolis kernel that makes local moves at the scale $\rho_b(x)$ (Fig.~\ref{fig:multiscale}).  Nearby points have proposal balls with a large common part.  A constant fraction of this common region is valid at both points and is accepted by both Metropolis updates.  Thus, their transition laws have a constant overlap with respect to total variation distance.  On every chord, hit-and-run completely resamples the conditional distribution, whereas the auxiliary kernel makes only a lazy local update.  A Dirichlet-form comparison shows that complete refresh is at least as effective, so the conductance bound for the auxiliary kernel transfers to hit-and-run~\cite{RudolfUllrichComparison}.  Appendix~\ref{app:log-affine} gives the details.

\begin{figure}
    \centering
    \includegraphics[width=0.3\linewidth]{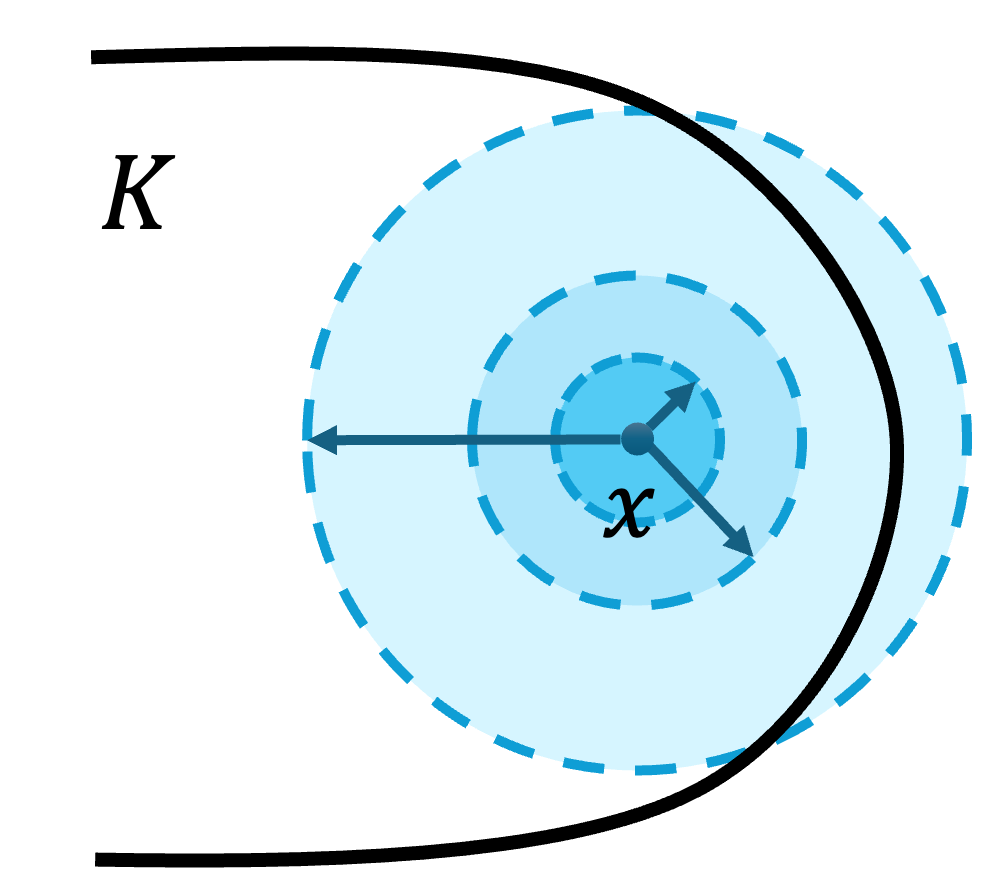}
    \caption{The auxiliary multiscale Metropolis kernel. In each step, it randomly chooses a step size of scale comparable with its local radius $\rho_b(x)$ from a dyadic set, and performs a ball-walk-like update.}
    \label{fig:multiscale}
\end{figure}


\subsection{Related work}

\paragraph{Geometric sampling.}
Classical sampling methods for the uniform distribution over a convex body given by its membership oracle (or equivalently, the zeroth-order oracle for its indicator function) include the ball walk~\cite{LS,KLS97} and hit-and-run~\cite{Smith,LovaszHR,LV}; both extend to general log-concave targets~\cite{LVlogconcave}.  See Vempala's survey~\cite{VempalaSurvey} for a broader account. A newer line begins with the restricted-Gaussian-oracle formulation of the proximal sampler~\cite{LeeShenTian21} and its heat-flow analysis~\cite{ChenChewiSalimWibisono22}.  Kook--Vempala--Zhang implemented this framework for convex bodies using membership queries, obtaining the In-and-Out walk and R\'enyi-divergence guarantees \cite{KookVempalaZhangInOut}.  Subsequent work strengthened the output guarantee to R\'enyi-infinity divergence~\cite{KookZhangRenyi} and allowed both the initial warmness and the final error to be measured in the same finite-order R\'enyi divergence~\cite{KookVempalaZerothOrder}.  Related developments produced subcubic cold-start algorithms~\cite{KookVempalaCold} and extended algorithmic diffusion to arbitrary log-concave targets through exponential lifting~\cite{KookVempalaDiffusion,KookVempalaUnified}.
For explicitly represented polytopes, another line of work uses the self-concordant barriers and related local geometry, which includes the Dikin walk~\cite{KannanNarayananDikin,kook2026d25mixingbounddikinwalks}, the Vaidya and John walks~\cite{ChenDwivediWainwrightYu}, Riemannian Hamiltonian Monte Carlo~\cite{LeeVempalaRHMC,RHMCLewis}, the strong-self-concordance framework~\cite{LaddhaLeeVempala}, and the interior-point method for sampling~\cite{kook2024gaussian}.  

\paragraph{Volume estimation.}
Volume estimation is one of the most important applications of geometric sampling.  Simulated-annealing algorithms connect a distribution with a known normalizing constant to the uniform distribution on the target body and use samples from consecutive, mutually warm phases to estimate successive ratios of normalizing constants. Dyer--Frieze--Kannan gave the first randomized polynomial-time volume algorithm in the membership-oracle model~\cite{DFK91}. Later milestones include the $O^*(n^5)$~\footnote{We use $O^*$ to hide logarithmic factors and the dependence on accuracy.} algorithm of Kannan--Lov\'asz--Simonovits~\cite{KLS97}, the $O^*(n^4)$ simulated-annealing algorithm of Lov\'asz--Vempala~\cite{LVVolume}, and the $O^*(n^3)$ Gaussian-cooling algorithm of Cousins--Vempala for well-rounded convex bodies~\cite{CV18}. For general convex bodies, Jia--Laddha--Lee--Vempala showed that the combined cost of rounding and volume estimation is $\widetilde O(n^{3.5}/\psi_n^2+n^3/\varepsilon^2)$ membership queries~\cite{JLLV26}.  More recently, Kook--Vempala developed an algorithmic-diffusion framework for sampling, rounding, and integrating general log-concave functions; for uniform distributions on convex bodies, it matches the best-known complexity~\cite{KookVempalaDiffusion}.

\paragraph{KLS, thin shell, and slicing.}
The \emph{KLS conjecture} asserts that, after isotropic normalization, the least expanding cut of a log-concave measure is comparable up to a universal factor with a half-space cut~\cite{KLS95}.  In our notation, it predicts $\psi_n\gtrsim1$. The \emph{thin-shell conjecture} asks whether every isotropic log-concave random vector $X\in\R^n$ satisfies $\operatorname{Var}(\|X\|_2^2)\lesssim n$~\cite{anttila2003central,bobkov2004central}; in particular, this says that $\|X\|_2$ fluctuates by at most a constant around $\sqrt n$. Bourgain's \emph{slicing conjecture} asks whether every unit-volume convex body has a hyperplane section of $(n-1)$-dimensional volume bounded below by a universal constant~\cite{bourgain1986high}. KLS implies thin shell, and thin shell implies slicing~\cite{EldanKlartagThinSlice}.  Conversely, Eldan showed that thin shell implies KLS up to polylogarithmic factors~\cite{EldanSL}, while no implication from slicing back to thin shell is known.  In particular, the recent resolutions of thin shell and slicing do not settle KLS.  

Building on Guan's estimate~\cite{GuanSlicing}, Klartag--Lehec resolved the slicing conjecture~\cite{KlartagLehecSlicing}; Bizeul subsequently gave an alternative proof through small-ball estimates~\cite{BizeulSlicing}. Klartag--Lehec later resolved the thin-shell conjecture using parallel coupling~\cite{KlartagLehecThinShell}.  A recent preprint of Chen--Klartag sharpens this result to the optimal inequality $\operatorname{Var}(\|X\|_2^2)\le8n$~\cite{ChenKlartagThinShell}. A concurrent preprint of Letwin proves the same sharp bound as part of an inequality for all quadratic forms and uses it to improve the KLS estimate discussed below~\cite{LetwinKLS}. 

For the KLS conjecture, the original localization argument gave $\psi_n\gtrsim n^{-1/2}$~\cite{KLS95}.  Eldan's stochastic-localization method yielded $\psi_n\gtrsim n^{-1/3}(\log n)^{-1/2}$~\cite{EldanSL}, and Lee--Vempala
improved this to $\psi_n\gtrsim n^{-1/4}$~\cite{LeeVempalaKLS}.  Chen later
obtained $\psi_n\geq n^{-o(1)}$~\cite{ChenKLS}, followed by the first
polylogarithmic bound $\psi_n\gtrsim (\log n)^{-5}$ of Klartag--Lehec~\cite{KlartagLehecKLS}.  The best published estimate is Klartag's $\psi_n\gtrsim(\log n)^{-1/2}$~\cite{KlartagKLS}; a recent preprint of Letwin improves this to $\psi_n\gtrsim(\log n)^{-1/4}$~\cite{LetwinKLS}.  

\paragraph{Localization and needle decompositions.}
Classical localization reduces high-dimensional integrals and isoperimetric questions to one-dimensional inequalities on weighted line segments, or needles~\cite{LS90,LS,KLS95}.  Klartag's optimal-transport formulation strengthens classical localization to a measurable decomposition generated by a common \(1\)-Lipschitz guiding function~\cite{KlartagNeedles}.  This additional structure has previously enabled Poincar\'e-based arguments in rigidity and stability theory~\cite{MaiRigidity,MaiOhta,BertrandFathi}.  We combine it instead with concave-weight isoperimetry to study hit-and-run.  This static decomposition differs from stochastic localization~\cite{EldanSL}, which evolves a random probability measure; see the recent survey~\cite{ShiTianZhangSL} and the references therein.

\paragraph{Roadmap.} In Section~\ref{sec:preliminaries}, we introduce the notations and the preliminaries in probability theory, convex geometry, and Markov chains.  In Section~\ref{sec:depth}, we define the local conductance radius and show its basic geometric properties. In Section~\ref{sec:one-dimensional}, we prove the one-dimensional concave-weight isoperimetric inequality, and in Section~\ref{sec:balanced-needles}, we lift it to high dimensions using guided localization and Poincar\'e averaging. In Section~\ref{sec:intrinsic_metric}, we prove the local overlap for hit-and-run in the intrinsic geometry. In Section~\ref{sec:conductance}, we combine the overlap with weighted isoperimetry to prove the conductance and mixing-time bounds. 

\section{Preliminaries}\label{sec:preliminaries}
\paragraph{Notation.} A \emph{convex body} is a compact convex subset of \(\R^n\) with nonempty interior.  Throughout the paper, $K\subseteq\R^n$ is a convex body in \emph{isotropic position}: the uniform distribution $\mu=\operatorname{Unif}(K)$ has zero mean and identity covariance, i.e.,
\begin{equation}\label{eq:isotropic}
 \E_{X\sim\mu}[X]=0,
 \qquad
 \E_{X\sim\mu}[X\otimes X]=I_n.
\end{equation}
We write \(\vol(\cdot)\) and \(dx\) for Lebesgue measure, $B(x,r)$ for a ball in $\R^n$ with center $x$ with radius $r$, $v_n=\vol(B(0,1))$ for the volume of the unit ball in $\R^n$, and \(\sigma_{n-1}\) for the rotation-invariant probability measure on \(\mathbb{S}^{n-1}\). For measures $\lambda,\nu$ on the same measurable space, $\lambda\ll\nu$ means that $\lambda$ is absolutely continuous with respect to \(\nu\); equivalently, \(\nu(A)=0\) implies \(\lambda(A)=0\) for every measurable \(A\).  When \(\lambda\ll\nu\), we denote its Radon--Nikodym derivative by \(d\lambda/d\nu\).
Vector norms are Euclidean and are written as \(\|x\|_2\).  We write $a\lesssim b$ when \(a\le Cb\) for a universal constant \(C\), and define \(a\gtrsim b\) and \(a\asymp b\) analogously.  Universal constants may change from one occurrence to the next.

\subsection{Markov chains and mixing time}
For probability measures \(\nu_1,\nu_2\), their total-variation distance is
\[
 \dtv(\nu_1,\nu_2)
 :=
 \sup_{A\text{ Borel}}
 \bigl|\nu_1(A)-\nu_2(A)\bigr|.
\]
\begin{definition}[$M$-warmness]
An initial distribution \(\mu_{\mathrm{init}}\) is
\emph{\(M\)-warm} with respect to \(\mu\) if
\[
 \mu_{\mathrm{init}}(A)\le M\mu(A)
\]
for every Borel set \(A\); or equivalently,
\(\frac{d\mu_{\mathrm{init}}}{d\mu}\le M\) almost everywhere.    
\end{definition}

Suppose that a Markov kernel \(\mathsf P\) is reversible with stationary measure \(\mu\).
For measurable sets \(A,B\subseteq K\), define the probability flow
\begin{equation}\label{eq:stationary-flow}
 \cQ_{\mathsf P}(A,B)
 :=
 \int_A\mathsf P(x,B)\,d\mu(x).
\end{equation}
Reversibility implies
\[
 \cQ_{\mathsf P}(A,B)
 =
 \cQ_{\mathsf P}(B,A).
\]
In particular, when $B=A^c$, $\cQ_{\mathsf P}(A,A^c)$ is called an \emph{ergodic flow}.

The mixing time of a Markov kernel \(\mathsf P\) from an initial distribution \(\mu_{\mathrm{init}}\) is
\begin{equation}\label{eq:mixing-time}
 \tau_{\mix}
 \bigl(
   \varepsilon,\mu_{\mathrm{init}},\mu;\mathsf P
 \bigr)
 :=
 \inf\left\{
 t\in\mathbb N:
 \dtv(\mu_{\mathrm{init}}\mathsf P^t,\mu)
 \le\varepsilon
 \right\},
\end{equation}
where \(\mathsf P^t\) denotes the \(t\)-step transition kernel:
\[
 (\mu_{\mathrm{init}}\mathsf P^t)(A)
 :=
 \int_K \mathsf P^t(x,A)\,d\mu_{\mathrm{init}}(x).
\]

Bounding the mixing time directly is difficult.  Instead, we bound the
expansion of the chain, measured by conductance, and let the
Lov\'asz--Simonovits machinery convert expansion into mixing.

\begin{definition}[Truncated conductance]
For \(0<s<1/2\), the \emph{\(s\)-conductance} of \(\mathsf P\) is
\begin{equation}\label{eq:s-conductance}
 \Phi_s(\mathsf P)
 :=
 \inf_{\substack{A\subseteq K\text{ measurable}\\s<\mu(A)\le1/2}}
 \frac{
   \cQ_{\mathsf P}(A,A^c)
 }{
   \mu(A)-s
 }.
\end{equation}
\end{definition}


The following lemma converts a lower bound on the $s$-conductance $\Phi_s({\sf P})$ into a warm-start mixing time bound. Lov\'asz--Simonovits proved it under the assumption that \(\mathsf P\) is lazy~\cite[Corollary~1.6]{LS}. As observed by Rudolf--Ullrich~\cite[Section~1]{RudolfUllrichPositivity}, the same conclusion holds for any reversible kernel whose associated Markov operator is \emph{positive}, meaning that
\begin{align*}
    \langle f,P f\rangle_{L^2(\mu)}:=\int f(x)(Pf)(x)\,d\mu(x)\ge0\qquad \text{for every}~f\in L^2(\mu).
\end{align*}

\begin{lemma}[Warm-start mixing time bound~\cite{LS,RudolfUllrichPositivity}]\label{lem:LS}
Let \(\mathsf P\) be a reversible positive Markov kernel with a
nonatomic stationary probability measure \(\mu\).  If
\(\mu_{\mathrm{init}}\) is \(M\)-warm, then for every
\(0<s<1/4\) and every \(t\in\mathbb N\),
\begin{equation}\label{eq:LS-explicit}
 \dtv(\mu_{\mathrm{init}}\mathsf P^t,\mu)
 \le
 Ms+
 M\left(
   1-\frac{\Phi_s(\mathsf P)^2}{2}
 \right)^t.
\end{equation}
In particular, if \(0<\varepsilon<1/2\), then for $s=\varepsilon/(2M)$,
one has
\begin{equation}\label{eq:LS}
 \tau_{\mix}
 \bigl(
   \varepsilon,\mu_{\mathrm{init}},\mu;\mathsf P
 \bigr)
 \lesssim
 \Phi_s(\mathsf P)^{-2}
 \log\frac{2M}{\varepsilon}.
\end{equation}
\end{lemma}

\subsection{Hit-and-run walk}\label{sec:hit-and-run}

The hit-and-run walk makes the following update.

\begin{center}
\begin{minipage}{0.86\textwidth}
\hrule
\smallskip
\textbf{One hit-and-run update from \(X_t=x\in K \).}
\begin{enumerate}[leftmargin=2em,itemsep=2pt,topsep=4pt]
\item Sample \(\theta\sim\sigma_{n-1}\).
\item Let
\[
 [a,b]=\{\lambda\in\R:x+\lambda\theta\in K\}.
\]
Sample \(\lambda\) uniformly from \([a,b]\) and set
\(X_{t+1}=x+\lambda\theta\).
\end{enumerate}
\smallskip
\hrule
\end{minipage}
\end{center}


We denote this kernel by \(\HR\).  It is reversible with respect to the uniform measure \(\mu\); see, e.g.~\cite[Section~2]{VempalaSurvey}. Moreover, $\HR$ is positive~\cite{RudolfUllrichPositivity}.  Thus \(\HR\) satisfies the conditions of Lemma~\ref{lem:LS}.

\subsection{Log-concavity}
A function $f:\R^n \rightarrow [0, \infty)$ is \emph{log-concave} if $-\log f$ is convex, with the convention $\log 0=-\infty$. A probability measure $\nu$ is log-concave if it has a log-concave density with respect to Lebesgue measure. In particular, for a convex body $K$, the uniform measure $\mu=\Unif(K)$ is log-concave.

The following lemma collects some standard estimates for one-dimensional
log-concave distributions, and the proof is deferred
to Appendix~\ref{app:one-dimensional-facts}.

\begin{restatable}[\cite{LVlogconcave,CE}]{lemma}{onedimest}\label{lem:one-dimensional-facts}
Let $X$ follow a nondegenerate log-concave distribution $\lambda$ with density $\rho$ on an interval in $\R$, distribution function $F(x)=\Prb\{X\le x\}$, and variance $\sigma^2$. Let $q(x)=\min\{F(x),1-F(x)\}$.  For every interior point $x$ of the
support,
\begin{align}
 \rho(x)&\gtrsim\frac{q(x)}{\sigma},
 \label{eq:quantile-density}\\
 \|\rho\|_\infty&\lesssim\frac1\sigma,
 \label{eq:max-density}\\
 \abs{x-\E X}&\lesssim
 \sigma\left(1+\log\frac1{q(x)}\right).
 \label{eq:quantile-location}
\end{align}
\end{restatable}

\subsection{Asymptotic convex geometry}
Our proof uses some standard tools in high-dimensional convex geometry.
\begin{definition}[Cheeger constant]
For a probability distribution \(\nu\) over $\R^n$, the \emph{Cheeger isoperimetric constant} of $\nu$ is 
\begin{equation}\label{eq:Cheeger}
 \psi_{\sf ch}(\nu):=
 \inf_{\substack{A\subseteq\R^n\text{ measurable}\\0<\nu(A)<1}}
 \frac{
   \nu^+(A)
 }{
   \min\{\nu(A),1-\nu(A)\}
 },
\end{equation}
where 
\begin{align*}
    \nu^+(A)=\liminf_{r\downarrow0}\frac{\nu\left(\{x:
 \operatorname{dist}(x,A)<r\}\right)-\nu(A)}{r}.
\end{align*}
\end{definition}
At dimension \(n\), write
\begin{equation}\label{eq:KLS-constant}
 \psi_n
 :=
 \inf\left\{
 \psi_{\sf ch}(\nu):
 \nu
 \text{ is isotropic and log-concave on }\R^n
 \right\}.
\end{equation}
It is easy to show that
\begin{equation}\label{eq:KLS-upper}
 \psi_n\leq \psi_{\sf ch}(\mathcal{N}(0,I_n))\lesssim 1.
\end{equation}
The reciprocal quantity \(\psi_{\sf KLS}(n):=\psi_n^{-1}\) is called the
KLS constant.  The KLS conjecture~\cite{KLS95} asserts that
\(\psi_{\sf KLS}(n)=O(1)\), or equivalently that \(\psi_n\gtrsim1\). The state-of-the-art result is \(\psi_{\sf KLS}(n)\lesssim(\log n)^{1/4}\),  or equivalently \(\psi_n\gtrsim(\log n)^{-1/4}\) by Letwin~\cite{LetwinKLS}. 

\begin{definition}[Poincar\'e inequality]
A probability measure $\nu$ satisfies a Poincar\'e inequality with constant \(C_{\sf PI}(\nu)>0\) if for every locally Lipschitz \(f:\R^n\to\R\),
\begin{equation}\label{eq:Poincare}
 \operatorname{Var}_\nu(f)
 \le
 C_{\sf PI}(\nu)
 \int\|\nabla f(x)\|_2^2\,d\nu(x)
\end{equation}
\end{definition}

The following proposition shows that the Poincar\'e constant and the Cheeger constant are equivalent:
\begin{proposition}[\cite{Cheeger,buser1982note,ledoux2004spectral}]
For any log-concave distribution $\nu$,
\begin{equation}\label{eq:Cheeger-Poincare}
 C_{\sf PI}(\nu)\asymp \psi_{\sf ch}(\nu)^{-2}.
\end{equation}
\end{proposition}
Applying to the isotropic log-concave measure \(\mu\), we have
\begin{equation}\label{eq:isotropic-Poincare}
 C_{\sf PI}(\mu)\asymp \psi_{\sf ch}(\mu)^{-2} \leq
 \psi_n^{-2}.
\end{equation}

Berwald's inequality provides the reverse H\"older estimate needed in our proof.  We use its log-concave-measure form; the classical uniform-measure theorem has a slightly sharper constant.

\begin{lemma}[Berwald's inequality~\cite{Berwald,LangharstPutterman}]\label{lem:berwald}
Let $\lambda$ be a log-concave probability measure, and let $w\ge0$ be a finite concave function on the convex support of $\lambda$.  Then
\begin{equation}\label{eq:Berwald-R}
 \int w^2\,d\lambda
 \le 2\left(\int w\,d\lambda\right)^2.
\end{equation}
If $\lambda=\Unif(K)$ on an $n$-dimensional convex body, the classical Berwald inequality gives the sharper factor $2(n+1)/(n+2)$ in place of $2$.
\end{lemma}

\subsection{Relaxed weighted perimeter}
\label{sec:relaxed-perimeter}

The conductance argument needs a notion of boundary size for arbitrary measurable cuts.  Let $\lambda$ be a probability measure supported on a convex body $K\subseteq\R^n$, with $\lambda\ll dx$, and let $w:K\to[0,\infty)$ be positive on $\operatorname{int}(K)$ with a Lipschitz zero extension to $\R^n$.  One should think of $w(x)$ as the local distance that the walk can move from $x$.  Because this scale vanishes at $\partial K$, the distance functions used later need not be globally Lipschitz in Euclidean distance.  We therefore use a direct relaxation that places no regularity assumption on the cut itself.

We write \(\Lip(K)\) for the real-valued Lipschitz functions on \(K\).
The following definition is adapted from the lower-semicontinuous relaxation of the weighted variation $f\mapsto\int w\|\nabla f\|_2\,d\lambda$ in BV functions theory; see, for example, \cite{AFP}.
\begin{definition}[Relaxed weighted perimeter]\label{def:relax_weight_per}
For a measurable set \(E\subseteq K\), define its \emph{relaxed weighted
perimeter with weight \(w\)} by
\begin{equation}\label{eq:relaxed-perimeter}
 \cP_{w,\lambda}(E)
 :=
 \inf
 \left\{
 \liminf_{k\to\infty}
 \int_K
 w(x)\|\nabla f_k(x)\|_2\,d\lambda (x):
 \begin{array}{l}
 f_k\in\Lip(K),\quad 0\le f_k\le1,\\
 f_k\to\one_E\text{ in }L^1(\lambda)
 \end{array}
 \right\}.
\end{equation}
When $\lambda=\mu=\Unif(K)$, we abbreviate $\cP_{w,\mu}$ to $\cP_w$.
\end{definition}

To see why $\mathcal{P}_{w}(E)$ is a notion of perimeter, note that each \(f_k\) is a bounded Lipschitz approximation to the indicator \(\one_E\): \(f_k\approx 1\) on most of \(E\), \(f_k\approx 0\) on most of \(K\setminus E\), and changes continuously across a thin transition region.  The gradient norm \(\|\nabla f_k(x)\|_2\) measures the local rate of this change.  A $0$--$1$ transition across a layer of thickness \(\delta\) has a gradient of order \(1/\delta\), while the layer has volume of order \(\delta\) times its interfacial area.  Integrating the gradient norm therefore produces a quantity of the order of that area, independent of the chosen transition width.  The factor \(w(x)\) assigns a local weight to placing the interface near \(x\). This interpretation is exact for regular sets.  Suppose that \(E\subseteq K\) has a sufficiently smooth boundary, and we take $w\equiv 1$. Then
\[
 \cP_1(E)= \mu^+(E)=\liminf_{r\downarrow0}\frac{\mu\left(\{x\in K:\operatorname{dist}(x,E)<r\}\right)-\mu(E)}{r}.
\]
The relaxation extends this boundary-area interpretation to arbitrary measurable sets.  

\begin{restatable}[Basic properties of the relaxed weighted perimeter]{proposition}{propperimeter}
\label{prop:relaxed-perimeter-toolkit}
Let $\lambda$ be a probability measure supported on a convex body $K\subseteq\R^n$, with $\lambda\ll dx$.
Let \(w:K\to[0,\infty)\) be positive on \({\operatorname{int}(K)}\), with a Lipschitz zero
extension to \(\R^n\).

\begin{enumerate}[label=\textup{(\roman*)}]

\item \emph{Complement symmetry.}  For every measurable \(E\subseteq K\),
\[
 \cP_{w,\lambda}(E)=\cP_{w,\lambda}(K\setminus E).
\]

\item \emph{One-sided coarea.}
Let $u$ be locally Lipschitz on $\operatorname{int}(K)$ and suppose that $w\|\nabla u\|_2$ is bounded Lebesgue-almost everywhere.  Then, for every $a<b$,
\begin{equation}\label{eq:relaxed-coarea}
 \int_a^b
 \cP_{w,\lambda}(\{u<t\})\,dt
 \le
 \int_{\{x:a<u(x)<b\}}
 w(x)\|\nabla u(x)\|_2\,d\lambda(x).
\end{equation}
Here and below, extend \(u\) by zero on \(\partial K\).  Any measurable
choice of boundary values gives the same relaxed weighted perimeter because
\(\lambda(\partial K)=0\).

\item \emph{Slicing.}
Suppose that
\[
 \lambda=\int\lambda_I\,d\pi(I)
\]
is a measurable disintegration over nondegenerate line intervals \(I\), and that each \(\lambda_I\) is absolutely continuous with respect to arclength.  If \(H(I)\) is a nonnegative
measurable function such that
\[
 \cP_{w,\lambda_I}(E\cap I)\ge H(I)
 \qquad\text{for $\pi$-almost every }I,
\]
then
\begin{equation}\label{eq:slicing}
 \cP_{w,\lambda}(E)
 \ge
 \int H(I)\,d\pi(I).
\end{equation}
\end{enumerate}
\end{restatable}
Appendix~\ref{app:relaxed-perimeter} gives elementary and self-contained
proofs of these properties tailored to the present weighted setting.

\section{The local conductance radius}\label{sec:depth}

The local geometry around a starting point controls whether nearby hit-and-run transition laws overlap. We quantify the room available at \(x\) by the largest radius for which \(B(x,r)\) lies almost entirely inside \(K\). The resulting local conductance radius plays two roles: it weights the boundary in our isoperimetric inequality and provides the local unit of distance in the overlap argument.
\begin{definition}[Local conductance~{\cite{LovaszHR}}]
For \(x\in K\) and \(r>0\), define the \emph{local conductance at scale
\(r\)} by
\begin{equation}\label{eq:occupancy-ratio}
 \lambda(x,r)
 :=
 \frac{\vol(K\cap B(x,r))}{\vol(B(0,r))};
\end{equation}
The associated \emph{local conductance radius} is
\begin{equation}\label{eq:local-radius}
 \rK(x)
 :=
 \sup\left(
 \{0\}\cup
 \left\{r>0:\lambda(x,r)\ge\frac{63}{64}\right\}
 \right).
\end{equation}
\end{definition}
The threshold $63/64$ ensures that a hit-and-run step has a constant probability of moving a distance comparable with $\rK(x)$; see Lemma~\ref{lem:step-quantile}.  Replacing it by another fixed value sufficiently close to one changes only universal constants.

The following lemma summarizes the properties of $\rK$ we will use later, and the proof is deferred to Appendix~\ref{app:proof_local_cond}. 

\begin{restatable}[Properties of the local conductance radius~\cite{LV}]{lemma}{lemlocalconductance}\label{lem:local-radius}
Let \(K\subseteq\R^n\) be a convex body.  Then:
\begin{enumerate}[label=\textup{(\roman*)}]
\item for every $x\in \operatorname{int}(K)$,
\begin{equation}\label{eq:radius-positive-finite}
 0<\operatorname{dist}(x,\partial K)\le \rK(x)<\infty;
\end{equation}
 \item \(\rK\) is concave on \(K\).

 \item for every \(x\in K\),
 \begin{equation}\label{eq:radius-boundary-distance}
  \rK(x)
  \lesssim
  \sqrt n\,\operatorname{dist}(x,\partial K).
 \end{equation}
 In particular, \(\rK=0\) on \(\partial K\), and its zero extension to \(\R^n\) is \(O(\sqrt n)\)-Lipschitz.
\end{enumerate}
If, in addition, $\mu=\Unif(K)$ is isotropic, then
 \begin{equation}\label{eq:mean-local-radius}
  \int_K\rK(x)\,d\mu(x)\gtrsim\frac1{\sqrt n}.
 \end{equation}
\end{restatable}


\section{One-dimensional concave-weight isoperimetry}
\label{sec:one-dimensional}
We now prove the one-dimensional isoperimetric inequality needed after localization. Let \(\lambda\) be a log-concave probability measure on an interval, with standard deviation \(\sigma\), and let \(w\ge0\) be a concave weight. For a set \(B\) of mass \(p\le1/2\), the expected weighted perimeter scale is $p\E_\lambda[w]/\sigma$. Theorem~\ref{thm:one-dimensional} shows that the possible decay of \(w\) near the endpoints costs at most the factor \(1+\log(1/p)\).

More formally, let $I:=(a,b)$ be an open interval. For a nonnegative finite-valued concave weight $w:I\to[0,\infty)$, and a probability measure $\lambda$ with density function $\rho(x)$ supported on $I$, let the one-dimensional relaxed weighted perimeter be
\begin{equation}\label{eq:def_P_w_lambda}
 \cP_{w,\lambda}(B)
 :=\inf\left\{
   \liminf_{k\to\infty}\int_a^b w(x)|f_k'(x)|\,d\lambda(x):
   \begin{array}{l}
    f_k\in\Lip(I),\quad 0\le f_k\le1,\\
    f_k\to\one_B\text{ in }L^1(\lambda)
   \end{array}
 \right\}.
\end{equation}
Each $f_k$ is differentiable almost everywhere, and its weighted total variation can be written as
\begin{align}\label{eq:one_d_perimeter_integral}
    \int_a^b w(x)\rho(x)\cdot |f_k'(x)|\,dx.
\end{align}
Suppose that we have a uniform lower bound $w(x)\rho(x)\ge c$ for all $x$ in an interval $J\subseteq I$. Then, for any $s<t$ in $J$, the fundamental theorem of calculus gives
\[
 \int_a^bw\rho|f'_k|\,dx\geq \int_s^t w\rho|f_k'|\,dx
 \ge c\int_s^t|f_k'|\,dx \geq c\left|\int_s^tf'_k\,dx\right|
 = c|f_k(t)-f_k(s)|.
\]
Thus, a lower bound on $w\rho$ turns any definite transition of $f_k$ into
a lower bound on~\eqref{eq:one_d_perimeter_integral}.  To bound $\cP_{w,\lambda}(B)$, we must make this argument for the liminf of the total variations along every admissible sequence.
The proof has two ingredients.  First, we show that $w(x)\rho(x)$ cannot be
too small in a prescribed quantile range.  We then show that every sequence
converging to $\one_B$ makes a definite transition within a suitable central
quantile interval.  

The following lemma provides the first ingredient.
\begin{lemma}[Weighted density at a quantile]
\label{lem:weighted-density-quantile}
Let $\lambda$ be a nondegenerate log-concave probability measure on $\R$. Let $I=(a,b)$ be the interior of its support, $\rho$ be the continuous positive representative of its density on $I$, and  $\sigma^2\in(0,\infty)$ be its variance. Let $F(x)=\lambda((a,x])$ be its cumulative distribution function, and $q(x)=\min\{F(x),1-F(x)\}$.

Let $w:I\to[0,\infty)$ be finite-valued, concave, and not identically zero. Let $m_I=\int_I w\,d\lambda$.  
Then, $0<m_I<\infty$, and for every $x\in I$,
\begin{equation}\label{eq:pointwise-needle-boundary}
 w(x)\rho(x)
 \gtrsim
 \frac{{m_I}}{\sigma}
 \frac{q(x)}{1+\log(1/q(x))}.
\end{equation}
\end{lemma}

\begin{proof}
A nondegenerate one-dimensional log-concave law is atomless, and its canonical density
$\rho$ is continuous and strictly positive on $I=(a,b)$.  Thus, $F$ is continuous and strictly increasing from $0$ to $1$ on this interval. Since $w$ is nonnegative, concave, and not identically zero, it is strictly positive throughout $I$. 

Fix $x\in I$. Without loss of generality, we may assume $F(x)=q(x)\le1/2$. We first control the distances from $x$ to the support endpoints $a$ and $b$.  

For the left-endpoint, we claim that
\begin{equation}\label{eq:left-endpoint}
 \frac1{x-a}\le2\frac{\rho(x)}{q(x)}.
\end{equation}
If $a=-\infty$, we define $1/(x-a):=0$, and~\eqref{eq:left-endpoint} follows immediately.  Suppose that $a> -\infty$.  
If $\rho(t)\le\rho(x)$ for every $t\in (a,x)$, then
\[
 q(x)=\int_a^x\rho(t)\,dt\le(x-a)\rho(x)\qquad \Longrightarrow \qquad \frac{1}{x-a}\leq \frac{\rho(x)}{q(x)},
\]
which is stronger than~\eqref{eq:left-endpoint}.  Otherwise, fix an arbitrary $z\in (a,x)$ with
$\rho(z)>\rho(x)$ and set
\[
 \alpha:=\frac{\log\rho(z)-\log\rho(x)}{x-z}>0.
\]
For $t\in[z,x]$, concavity of $\log\rho$ implies that
\[
 \log\rho(t)
 \ge
 \frac{x-t}{x-z}\log\rho(z)
 +\frac{t-z}{x-z}\log\rho(x)
 =\log\rho(x)+\alpha(x-t).
\]
Consequently,
\begin{equation}\label{eq:rho-left}
 \int_z^x\rho(t)\,dt
 \ge
 \frac{\rho(x)}\alpha\bigl(e^{\alpha(x-z)}-1\bigr).
\end{equation}
For $t>x$, monotonicity of secant slopes of the concave function
$\log\rho$ gives
\[
 \frac{\log\rho(t)-\log\rho(x)}{t-x}
 \le
 \frac{\log\rho(x)-\log\rho(z)}{x-z}
 =-\alpha.
\]
Hence, $\rho(t)\le\rho(x)e^{-\alpha(t-x)}$, and therefore
\begin{equation}\label{eq:rho-right}
 1-q(x)=\int_x^b\rho(t)\,dt
 \le
 \frac{\rho(x)}\alpha.
\end{equation}
Since $1-q(x)\ge1/2$, \eqref{eq:rho-right} implies
$\alpha\le2\rho(x)$.  On the other hand,
\begin{align*}
    \int_z^x\rho(t)\,dt\leq \int_a^x\rho(t)\,dt=q(x)\leq \frac{1}{2},
\end{align*}
so \eqref{eq:rho-left} yields
\[
 e^{\alpha(x-z)}-1
 \le
 \frac{\alpha}{2\rho(x)}
 \le1.
\]
Thus, $\rho(z)=\rho(x)e^{\alpha(x-z)}\le2\rho(x)$ holds for any $z\in (a,x)$ with $\rho(z)>\rho(x)$. Hence, we have
\[
 \sup_{a<z<x}\rho(z)\le2\rho(x).
\]
It follows that
\[
 q(x)=\int_a^x\rho(t)\,dt\le2(x-a)\rho(x)\qquad \Longrightarrow \qquad \frac{1}{x-a}\leq 2\frac{\rho(x)}{q(x)},
\]
which proves \eqref{eq:left-endpoint}.

For the right-endpoint, we claim that 
\begin{equation}\label{eq:right-endpoint}
 \frac1{b-x}\lesssim\frac1\sigma.
\end{equation}
If $b=\infty$, define $1/(b-x):=0$.  If $b<\infty$, then the
maximum-density estimate \eqref{eq:max-density} gives
\[
 \frac12\le1-q(x)=\int_x^b\rho(t)\,dt
 \le(b-x)\|\rho\|_\infty
 \lesssim\frac{b-x}{\sigma}.
\]

Concavity now compares the mean of $w$ with its value at $x$.  If $y\ge x$ and
$a> -\infty$, choose $z\in(a,x)$. 
Concavity and nonnegativity of $w(z)$ imply
\[
 w(x)
 \ge
 \frac{y-x}{y-z}w(z)+\frac{x-z}{y-z}w(y)
 \ge
 \frac{x-z}{y-z}w(y)\qquad \Longrightarrow\qquad w(y)\leq w(x)\left(1+\frac{y-x}{x-z}\right).
\]
Letting $z\downarrow a$ gives
\[
 w(y)\le w(x)\left(1+\frac{y-x}{x-a}\right).
\]
If $a=-\infty$, the same inequality follows by letting $z\to-\infty$, which gives simply $w(y)\le w(x)$.
The reflected argument gives, for $y\le x$,
\[
 w(y)\le w(x)\left(1+\frac{x-y}{b-x}\right),
\]
where the fractional term is zero when $b=\infty$.  Combining the two
cases,
\begin{equation}\label{eq:wy-wx-unified}
 w(y)
 \le
 w(x)\left(
 1+\frac{(y-x)_+}{x-a}+\frac{(x-y)_+}{b-x}
 \right).
\end{equation}

The bound \eqref{eq:wy-wx-unified} also proves that
$m_I<\infty$: the right-hand side grows at most
linearly in $|y|$, while finite variance of $\lambda$ implies that, $\E_{X\sim \lambda}[|X|]<\infty$.
Moreover, $m_I>0$ because $w>0$ on $I$ and $\lambda(I)=1$.

Integrating \eqref{eq:wy-wx-unified} with respect to $X\sim\lambda$ and
dividing by $w(x)>0$ yields
\begin{equation}\label{eq:mean-vs-point}
 \frac{m_I}{w(x)}
 \le
 1+\frac{\E[(X-x)_+]}{x-a}
   +\frac{\E[(x-X)_+]}{b-x}.
\end{equation}
The triangle inequality, Cauchy--Schwarz, and
\eqref{eq:quantile-location} give
\begin{equation}\label{eq:first-moment-at-quantile}
\begin{aligned}
 \E[|X-x|]\le &~ \E[|X-\E [X]|]+|x-\E [X]|\\
 \leq &~ \sigma + |x-\E [X]|\\
 \lesssim &~  \sigma\left(1+\log\frac{1}{q(x)}\right).
\end{aligned}
\end{equation}
Thus, we obtain
\begin{align*}
 \frac{m_I}{w(x)} \leq &~ 1+\E\left[|X-x|\right]\left(\frac{1}{x-a}+\frac{1}{b-x}\right)\\
 \lesssim &~ 1 + \sigma\left(1+\log\frac{1}{q(x)}\right) \left(\frac{1}{x-a}+\frac{1}{b-x}\right)\\
 \lesssim &~ 1 + \sigma\left(1+\log\frac{1}{q(x)}\right)\left(\frac{2\rho(x)}{q(x)}+\frac{1}{\sigma}\right) \\
 \lesssim &~ \left(1+\log\frac{1}{q(x)}\right)
 \frac{\sigma\rho(x)}{q(x)},
\end{align*}
where the third step follows from~\eqref{eq:left-endpoint} and~\eqref{eq:right-endpoint}, and the last step follows from~\eqref{eq:quantile-density} that $\frac{\sigma \rho(x)}{q(x)}\gtrsim 1$. Therefore,
\begin{equation*}
 w(x)\rho(x)
 \gtrsim
 \frac{m_I}{\sigma}
 \frac{q(x)}{1+\log(1/q(x))},
\end{equation*}
completing the proof of the lemma.
\end{proof}

The lemma controls the local price $w(x)\rho(x)$.  The theorem below completes the argument by showing that every Lipschitz approximation to $\one_B$ accumulates a definite amount of variation in a central interval where $q(x)\ge p/4$.

\begin{theorem}[One-dimensional concave-weight isoperimetry]\label{thm:one-dimensional}
Let $\lambda$ be a nondegenerate log-concave probability measure on $\R$. Let $I=(a,b)$ be the interior of its support, $\rho$ be the continuous positive representative of its density on $I$, and  $\sigma^2\in(0,\infty)$ be its variance. 

Let $w:I\to[0,\infty)$ be finite-valued, concave, and not identically zero. Let $m_I=\int_I w\,d\lambda$. For every measurable $B\subseteq I$ with
$p:=\lambda(B)\in (0, 1/2]$, its relaxed weighted perimeter satisfies
\begin{equation}\label{eq:one-dimensional}
 \cP_{w,\lambda}(B)
 \gtrsim
 \frac{m_I}{\sigma}
 \frac{p}{1+\log(1/p)}.
\end{equation}
\end{theorem}

\begin{proof}
Let $F$ be the cumulative distribution function of $\lambda$. Define $x_-,x_+\in I$ to be the $p/4$- and $(1-p/4)$-quantiles of $\lambda$; that is, $F(x_-)=p/4$ and $F(x_+)=1-p/4$. Let  $J_p:=[x_-,x_+]$. The quantiles are finite and lie strictly inside $I$.  Since $\lambda$
is atomless, we have
\begin{equation}\label{eq:central-interval-mass}
 \lambda(J_p)=1-\frac p2,
 \qquad
 \lambda(J_p^c)=\frac p2.
\end{equation}
Hence,
\begin{equation}\label{eq:central-two-sided-mass}
\begin{aligned}
 \lambda(B\cap J_p)\geq &~ \lambda(B)-\lambda(J_p^c)\ge\frac p2,\\
 \lambda(J_p\setminus B) \geq &~ \lambda(J_p)-\lambda(B) 
 \ge1-\frac{3p}{2}\ge\frac14.
\end{aligned}
\end{equation}

Let $(f_k)$ be an arbitrary admissible sequence for $\mathcal{P}_{w,\lambda}(B)$ in
\eqref{eq:def_P_w_lambda}, and set
\[
 \eta_k:=\|f_k-\one_B\|_{L^1(\lambda)}\longrightarrow0.
\]
The average of $f_k$ over $B\cap J_p$ satisfies
\[
 \frac1{\lambda(B\cap J_p)}\int_{B\cap J_p}f_k\,d\lambda
 =
 1-
 \frac1{\lambda(B\cap J_p)}\int_{B\cap J_p}(1-f_k)\,d\lambda
 \ge1-\frac{2\eta_k}{p},
\]
where the last step follows from $0\leq f_k\leq 1$ and~\eqref{eq:central-two-sided-mass}. 
Similarly,
\[
 \frac1{\lambda(J_p\setminus B)}
 \int_{J_p\setminus B}f_k\,d\lambda
 \le4\eta_k.
\]
Because $f_k$ is continuous on the compact interval $J_p$, its maximum
and minimum there are attained.  An average lies between the minimum and
maximum, so
\[
 \max_{J_p}f_k\ge1-\frac{2\eta_k}{p},
 \qquad
 \min_{J_p}f_k\le4\eta_k.
\]
Consequently,
\begin{equation}\label{eq:central-oscillation}
 \operatorname{osc}_{J_p}(f_k)
 :=\max_{J_p}f_k-\min_{J_p}f_k
 \ge1-\left(\frac2p+4\right)\eta_k.
\end{equation}
A Lipschitz function is absolutely continuous on the compact interval
$J_p$.  Therefore the fundamental theorem of calculus gives
\begin{equation}\label{eq:one-transition}
 \int_{J_p}|f_k'(x)|\,dx
 \ge
 \operatorname{osc}_{J_p}(f_k)
 \ge1-\left(\frac2p+4\right)\eta_k.
\end{equation}

For every $x\in J_p$, one has $q(x)\ge p/4$. Note that the function
$s\mapsto s/(1+\log(1/s))$ is increasing on $(0,1)$. Thus, \eqref{eq:pointwise-needle-boundary} in Lemma~\ref{lem:weighted-density-quantile} gives the uniform estimate
\begin{equation}\label{eq:uniform-central-weight}
 w(x)\rho(x)\gtrsim
 \frac{m_I}{\sigma}
 \frac{q(x)}{1+\log(1/q(x))}
 \gtrsim
 \frac{m_I}{\sigma}
 \frac{p}{1+\log(1/p)},
 \qquad x\in J_p.
\end{equation}
Using~\eqref{eq:one-transition} and~\eqref{eq:uniform-central-weight}, we conclude that
\begin{align*}
 \int_a^b w(x)|f_k'(x)|\,d\lambda(x)
 &\ge
 \int_{J_p}w(x)\rho(x)|f_k'(x)|\,dx\\
 &\gtrsim
 \frac{m_I}{\sigma}
 \frac{p}{1+\log(1/p)}
 \int_{J_p}|f_k'(x)|\,dx\\
 &\gtrsim
 \frac{m_I}{\sigma}
 \frac{p}{1+\log(1/p)}
 \left(1-\left(\frac2p+4\right)\eta_k\right).
\end{align*}
 Since $p>0$ is fixed and
$\eta_k\to0$,
\[
 \liminf_{k\to\infty}
 \int_a^b w(x)|f_k'(x)|\,d\lambda(x)
 \gtrsim
 \frac{m_I}{\sigma}
 \frac{p}{1+\log(1/p)}.
\]
This lower bound holds for every admissible sequence $(f_k)$.  Taking the infimum over all admissible sequences 
in \eqref{eq:def_P_w_lambda} proves that
\begin{align*}
    \mathcal{P}_{w,\lambda}(B)\gtrsim\frac{m_I}{\sigma}\frac{p}{1+\log(1/p)},
\end{align*}
completing the proof of the theorem.
\end{proof}

\section{Guided localization and local-radius weighted isoperimetry}
\label{sec:balanced-needles}
Section~\ref{sec:one-dimensional} gives the required inequality on a single interval. To lift it to \(K\), we disintegrate \(\mu\) into conditional measures \(\mu_I\) and average the one-dimensional bounds. The obstacle is the cut mass: a generic disintegration may have \(\mu_I(A)\in\{0,1\}\) on almost every interval, making \(A\cap I\) trivial and its conditional perimeter zero. Guided localization removes this obstacle. Applied to \(\one_A-\mu(A)\), it produces a disintegration satisfying $\mu_I(A)=\mu(A)$ for almost every \(I\), together with a common \(1\)-Lipschitz guiding function that is a unit-speed coordinate on each needle. Exact balance makes the one-dimensional estimate uniform across the intervals, while the common guide allows their conditional variances to be controlled after averaging.

The following lemma shows a more general result for log-concave measures, which is a Euclidean specialization of Klartag's theorem.  Appendix~\ref{app:klartag-balanced-needles} derives it from the general Riemannian statement.

\begin{restatable}[Localization with a guiding function~\cite{KlartagNeedles}]{lemma}{lemklartag}\label{lem:balanced-needles}
Let $\lambda$ be a full-dimensional log-concave probability measure with compact convex support $K\subseteq\R^n$, and let $A\subseteq K$ be measurable with $0<\lambda(A)<1$.  There exist a $1$-Lipschitz function $u:{\operatorname{int}(K)}\to\R$, a probability measure $\pi$ on a family of nondegenerate line intervals $I\subset{\operatorname{int}(K)}$, and conditional probability measures $\lambda_I$ such that
\begin{equation}\label{eq:balanced-needle-disintegration}
 \lambda=\int\lambda_I\,d\pi(I).
\end{equation}
For $\pi$-almost every $I$, there are an interval $J_I\subseteq\R$ and a unit-speed parametrization $\gamma_I:J_I\to I$ such that:
\begin{enumerate}[label=\textup{(\roman*)},leftmargin=2.2em]
\item after translating the parameter on $J_I$ if necessary,
\begin{equation}\label{eq:guiding-coordinate-on-needle}
 u(\gamma_I(t))=t\qquad(t\in J_I);
\end{equation}
\item the push-forward $(\gamma_I^{-1})_\#\lambda_I$ has a strictly positive log-concave density on $J_I$;
\item the set $A$ has the same conditional mass on every needle:
\begin{equation}\label{eq:balanced-needle-mass}
 \lambda_I(A)=\lambda(A).
\end{equation}
\end{enumerate}
\end{restatable}

The balanced decomposition lets us apply Theorem~\ref{thm:one-dimensional} with the same value $p=\mu(A)$ on every interval:
\begin{align*}
    \cP_{w,\mu_I}(A\cap I)\gtrsim\frac{m_I}{\sigma_I}\,\frac{p}{1+\log(1/p)},
\end{align*}
where $m_I$ is the mean weight and $\sigma_I^2$ is the variance of any unit-speed arclength
coordinate on $I$. After slicing and averaging over the intervals, the remaining task is to lower-bound \(\int \frac{m_I}{\sigma_I}\,d\pi(I)\). It can be controlled by a second-moment bound on $m_I$, together with a key observation that the averaged variance can be bounded by the global Poincar\'e constant:
\begin{align*}
    \int\sigma_I^2\,d\pi(I) \le \operatorname{Var}_\mu(u) \le C_{\sf PI}(\mu).
\end{align*}
The following theorem makes this argument formal and more general. 


\begin{theorem}[Isoperimetry with a concave weight]
\label{thm:balanced-weighted}
Let $\lambda$ be a full-dimensional log-concave probability measure with compact convex support $K\subseteq\R^n$. Let  $w:K\to[0,\infty)$ be concave and positive on ${\operatorname{int}(K)}$, and assume that its zero extension to $\R^n$ is Lipschitz.
Then, for every measurable $A$ with $p=\lambda(A)\in (0,1/2]$, 
\begin{equation}\label{eq:balanced-weighted}
 \cP_{w,\lambda}(A)
 \gtrsim
 \frac{\int_Kw\,d\lambda}{\sqrt{C_{\sf PI}(\lambda)}}
 \frac{p}{1+\log(1/p)}.
\end{equation}
\end{theorem}

\begin{proof}
Apply Lemma~\ref{lem:balanced-needles} to $A$.  For each nontrivial needle, set
\[
 m_I:=\int_Iw\,d\lambda_I,
 \qquad
 \sigma_I^2:=\operatorname{Var}_{\lambda_I}(u).
\]
Both quantities are finite and positive.  The Lipschitz zero extension makes
$w$ bounded on the compact body $K$, while $w>0$ on ${\operatorname{int}(K)}$ and the
conditional density is positive on $I$.  Moreover,
$u(\gamma_I(t))=t$ on a bounded nondegenerate interval with positive
density, so its conditional variance lies in $(0,\infty)$.  The restriction
of $w$ to $I$ is concave, while the balancing
property gives $\lambda_I(A)=p$.  Hence,
Theorem~\ref{thm:one-dimensional} gives
\begin{equation}\label{eq:needle-profile}
 \cP_{w,\lambda_I}(A\cap I)
 \gtrsim
 \frac{m_I}{\sigma_I}
 \frac{p}{1+\log(1/p)}.
\end{equation}
The conditional moments $m_I$ and $\sigma_I^2$ are measurable functions of
$I$.  Taking the right-hand side of \eqref{eq:needle-profile} with a fixed
universal constant as the function $H$ in Proposition~\ref{prop:relaxed-perimeter-toolkit}~\textup{(iii)} yields
\begin{equation}\label{eq:integrated-needle-profile}
 \cP_{w,\lambda}(A)
 \gtrsim
 \frac{p}{1+\log(1/p)}
 \int\frac{m_I}{\sigma_I}\,d\pi(I).
\end{equation}
If \(\int m_I/\sigma_I\,d\pi(I)=\infty\), the desired lower bound is
immediate from \eqref{eq:integrated-needle-profile}.  Otherwise, by Cauchy--Schwarz, we have
\begin{align}\label{eq:mI_sigmaI_ratio}
    \left(\int m_I\,d\pi(I)\right)^2=\left(\int\sqrt{\frac{m_I}{\sigma_I}}\sqrt{m_I\sigma_I}\,d\pi(I)\right)^2\le\left(\int\frac{m_I}{\sigma_I}\,d\pi(I)\right)\left(\int m_I\sigma_I\,d\pi(I)\right).
\end{align}

First, disintegration gives
\begin{equation}\label{eq:mean-needle-weight}
 \int m_I\,d\pi(I)=\int \left(\int_I w(x)\,d\lambda_I(x)\right)\,d\pi(I)=\int w\,d\lambda.
\end{equation}

Conditional Jensen and Berwald's inequality (Lemma~\ref{lem:berwald}) give
\begin{equation}\label{eq:second-needle-weight}
 \int m_I^2\,d\pi(I)
 \le\int w^2\,d\lambda
 \lesssim \left(\int w\,d\lambda\right)^2.
\end{equation}
The guiding function is Euclidean \(1\)-Lipschitz on the dense set
\({\operatorname{int}(K)}\), so it extends uniquely to a \(1\)-Lipschitz function on \(K\);
we use the same symbol for this extension. Thus, we obtain that 
\begin{equation}\label{eq:average-needle-variance}
 \int\sigma_I^2\,d\pi(I)
 \le\operatorname{Var}_\lambda(u)
 \le C_{\sf PI}(\lambda)\int_K\|\nabla u\|_2^2\,d\lambda
 \le C_{\sf PI}(\lambda),
\end{equation}
where the first step follows from the law of total variance, the second step follows from the Poincar\'e inequality, and the third step follows from  Rademacher's theorem and $\|\nabla u\|_2\le1$ almost everywhere. Cauchy--Schwarz applied to \eqref{eq:second-needle-weight} and
\eqref{eq:average-needle-variance} gives
\begin{equation}\label{eq:m-sigma-upper}
 \int m_I\sigma_I\,d\pi(I) \leq \left(\int m_I^2\,d\pi(I)\right)^{1/2}\left(\int \sigma_I^2 \,d\pi(I)\right)^{1/2}
 \lesssim \left(\int w\,d\lambda\right)
       \sqrt{C_{\sf PI}(\lambda)}.
\end{equation}

Therefore, combining \eqref{eq:mI_sigmaI_ratio},~\eqref{eq:mean-needle-weight}, and~\eqref{eq:m-sigma-upper} yields
\begin{equation}\label{eq:ratio-average}
 \int\frac{m_I}{\sigma_I}\,d\pi(I)
 \gtrsim
 \frac1{\sqrt{C_{\sf PI}(\lambda)}}\int w\,d\lambda.
\end{equation}
Thus, by~\eqref{eq:integrated-needle-profile}, we get that
\begin{align*}
    \cP_w(A) \gtrsim\frac1{\sqrt{C_{\sf PI}(\lambda)}}\left(\int w\,d\lambda\right)
 \frac{p}{1+\log(1/p)},
\end{align*}
which completes the proof of the theorem.
\end{proof}

We now apply the theorem to the local-radius weight $w=\rK$ and the uniform measure $\lambda = \mu$.

\begin{corollary}[Local-radius weighted isoperimetric inequality]\label{cor:full-depth}
For every measurable $A\subset K$, let
$p=\min\{\mu(A),1-\mu(A)\}$.  If $0<p\le1/2$, then
\begin{equation}\label{eq:full-depth-profile}
 \cP_{\rK}(A)
 \gtrsim
 \frac{\psi_n}{\sqrt n}
 \frac{p}{1+\log(1/p)}.
\end{equation}
\end{corollary}

\begin{proof}
Wlog, we may assume $p=\mu(A)$.  Apply Theorem~\ref{thm:balanced-weighted} with $w=\rK$.  By \eqref{eq:radius-positive-finite} and Lemma~\ref{lem:local-radius}, we know that $\rK$ is positive on ${\operatorname{int}(K)}$, concave, Lipschitz, and
\begin{align*}
    \int\rK\,d\mu\gtrsim \frac{1}{\sqrt{n}}.
\end{align*}
Finally,
\eqref{eq:isotropic-Poincare} gives
$C_{\sf PI}(\mu)^{-1/2}\gtrsim\psi_n$.  Substitution in
\eqref{eq:balanced-weighted} proves
\eqref{eq:full-depth-profile}.
\end{proof}

\section{Intrinsic geometry and local overlap for the hit-and-run walk}
\label{sec:intrinsic_metric}
Corollary~\ref{cor:full-depth} says that every cut has large relaxed perimeter when
its boundary near $x$ is weighted by the local scale $\rK(x)$.  To turn this
geometric statement into a conductance bound, we need a notion of proximity
adapted to the same varying scale.  This leads naturally to a path metric in
which an infinitesimal Euclidean displacement at $x$ is measured in units of
$\rK(x)$.

We will show that two points sufficiently close in this metric have
hit-and-run transition distributions with a constant amount of overlap.
This local overlap is the link between the relaxed weighted perimeter bound
and the ergodic flow studied in the next section.  We begin by defining
the metric.

\begin{definition}[Local-radius metric]\label{def:intrinsic-metric}
For $x,y\in\operatorname{int}(K)$, define
\begin{equation}\label{eq:drK}
 d_{\rK}(x,y)
 =
 \inf_{\gamma}
 \int_0^1
 \frac{\|\gamma'(t)\|_2}{\rK(\gamma(t))}\,dt,
\end{equation}
where the infimum is over absolutely continuous curves in
$\operatorname{int}(K)$ joining $x$ to $y$.
\end{definition}

$d_{\rK}(\cdot,\cdot)$ is indeed a metric.  To see that it is finite, fix
$x,y\in\operatorname{int}(K)$.  The segment $[x,y]$ is a compact subset of
the open convex set $\operatorname{int}(K)$, so continuity and positivity of
$\rK$ give $\min_{z\in[x,y]}\rK(z)>0$. Using this segment in \eqref{eq:drK} yields
\[
 d_{\rK}(x,y)\le\frac{\|x-y\|_2}{\min_{z\in[x,y]}\rK(z)}.
\]
Conversely, every admissible curve $\gamma$ satisfies
\[
 \int_0^1\frac{\|\gamma'(t)\|_2}{\rK(\gamma(t))}\,dt
 \ge
 \frac1{\max_{z\in K}\rK(z)}\int_0^1\|\gamma'(t)\|_2\,dt
 \ge
 \frac{\|x-y\|_2}{\max_{z\in K}\rK(z)}.
\]
Consequently,
\begin{equation}\label{eq:rK_bounds}
 \frac{\|x-y\|_2}{\max_{z\in K}\rK(z)}
 \le d_{\rK}(x,y)
 \le\frac{\|x-y\|_2}{\min_{z\in[x,y]}\rK(z)}.
\end{equation}
These bounds show that the metric separates points and is locally
bi-Lipschitz equivalent to Euclidean distance.  Symmetry and the triangle
inequality follow by reversing and concatenating curves.  For
$A,B\subseteq\operatorname{int}(K)$, write
\[
 d_{\rK}(A,B):=\inf_{x\in A,\,y\in B}d_{\rK}(x,y).
\]

We next relate $\rK(x)$ to the distance moved by one hit-and-run
step.  The one-eighth step quantile was introduced by Lov\'asz~\cite[Section~3]{LovaszHR}, while its comparison with the state-dependent radius $\rK$ appears in \cite[Lemma~3.2]{LV}.  We include a short polar-coordinate proof that is uniform for every $n\ge2$.

\begin{lemma}[Local radius and hit-and-run step length]
\label{lem:step-quantile}
Assume $n\ge2$.  For $x\in\operatorname{int}(K)$, define the $1/8$
step quantile
\begin{align}\label{eq:def_F_HR}
 F_{\mathrm{HR}}(x)
 :=
 \inf\left\{
 t\ge0:
 \HRzero\bigl(x,B(x,t)\bigr)\ge\frac18
 \right\}.
\end{align}
Then
\begin{equation}\label{eq:step-quantile}
 F_{\mathrm{HR}}(x)\ge\frac{\rK(x)}{32}.
\end{equation}
\end{lemma}

\begin{proof}
For $\theta\in\mathbb S^{n-1}$, let
\[
 R_x(\theta):=\max\{a\ge0:x+a\theta\in K\}
\]
be the distance from $x$ to the end of the chord in direction $\theta$.
Fix $0<r<\rK(x)$.  Polar integration about $x$ gives
\begin{equation}\label{eq:polar-occupancy}
 \lambda(x,r)=\frac{\vol(K\cap B(x,r))}{\vol(B(0,r))}
 =
 \int_{\mathbb S^{n-1}}
 \min\left\{
 1,\left(\frac{R_x(\theta)}r\right)^n
 \right\}
 \,d\sigma_{n-1}(\theta).
\end{equation}
The right-hand side is nonincreasing in $r$, so the definition of $\rK$
implies $\lambda(x,r)\ge63/64$. Thus,
\begin{align*}
    \frac{1}{64}\geq &~ 1-\lambda(x,r)=1-\int_{\mathbb S^{n-1}}\min\left\{1,\left(\frac{R_x(\theta)}r\right)^n\right\}\,d\sigma_{n-1}(\theta)\\
    \geq &~ \int_{\{\theta:R_x(\theta)<r/2\}} \underbrace{1- \min\left\{1,\left(\frac{R_x(\theta)}r\right)^n\right\}}_{\geq 1-2^{-n}}\,d\sigma_{n-1}(\theta)\\
    \geq &~ \left(1-2^{-n}\right)\sigma_{n-1}\left(\{\theta:R_x(\theta)<r/2\}\right).
\end{align*}
Equivalently,
\[
 \sigma_{n-1}(\{\theta:R_x(\theta)<r/2\})
 \le
 \frac{1}{64(1-2^{-n})}
 \le\frac1{48},
\]
where the last inequality uses $n\ge2$.

It follows by a union bound that, with probability at least $23/24$ over
the direction of hit-and-run, the chord extends at least $r/2$ from $x$ in
both directions.  Such a chord has length at least $r$, whereas its portion
inside $B(x,r/32)$ has length $r/16$.  Conditional on this direction,
uniform resampling from the chord therefore lands in $B(x,r/32)$ with
probability at most $1/16$.  Hence, by union bound,
\[
 \HRzero\bigl(x,B(x,r/32)\bigr)
 \le
 \frac1{24}+\frac{23}{24}\cdot\frac1{16}
 =
 \frac{13}{128}
 <
 \frac18.
\]

To pass to the endpoint, fix $t<\rK(x)/32$ and choose
$r\in(32t,\rK(x))$.  Then
$B(x,t)\subset B(x,r/32)$, so the probability of a step into $B(x,t)$ is
strictly less than $1/8$.  No such $t$ belongs to the set in the definition
of $F_{\mathrm{HR}}(x)$, which proves \eqref{eq:step-quantile}.  
\end{proof}

\begin{remark}
\cite{LovaszHR} defines the step quantile by the equality
\(
 \HRzero(x,B(x,F_{\mathrm{HR}}(x)))=1/8
\), rather than by the infimum in~\eqref{eq:def_F_HR}.  The two definitions agree here.  Indeed, conditional on a direction, the signed displacement is uniform on a nondegenerate interval, so its absolute value has no atoms.  Averaging over the direction preserves this property.  Consequently,
\(t\mapsto\HRzero(x,B(x,t))\) is continuous, and the infimum in our definition attains the value $1/8$.  Thus, these two definitions are equivalent.
\end{remark}

For distinct $x,y\in\operatorname{int}(K)$, let $p,q$ be the endpoints of
their common chord, labeled in the order $p,x,y,q$, and define the \emph{cross-ratio distance}~\cite{LovaszHR}:
\[
 d_K^{\mathrm{cr}}(x,y)
 :=
 \frac{\|x-y\|_2\,\|p-q\|_2}
      {\|x-p\|_2\,\|y-q\|_2}.
\]

\begin{lemma}[Lov\'asz's overlap lemma~{\cite[Lemma~8]{LovaszHR}}; see also {\cite[Lemma 4.1]{LV}}]\label{lem:lovasz_overlap}
Suppose that
\begin{align}\label{eq:lovasz_overlap_conditions}
 d_K^{\mathrm{cr}}(x,y)<\frac18\qquad \text{and}
 \qquad
 \|x-y\|_2
 <
 \frac2{\sqrt n}\max\{F_{\mathrm{HR}}(x),F_{\mathrm{HR}}(y)\}.
\end{align}
Then
\[
 \dtv\bigl(\HRzero(x,\cdot),\HRzero(y,\cdot)\bigr)
 <1-\frac1{500}.
\]
\end{lemma}
The next lemma verifies both geometric hypotheses in~\eqref{eq:lovasz_overlap_conditions} from a single bound in the
local-radius metric.

\begin{lemma}[Uniform local overlap for hit-and-run]
\label{lem:overlap}
Assume $n\ge2$.  There are universal constants $c_0,\alpha_0>0$ such that,
for all $x,y\in\operatorname{int}(K)$,
\begin{equation}\label{eq:local-overlap}
 d_{\rK}(x,y)<\frac{c_0}{\sqrt n}
 \qquad\Longrightarrow\qquad
 \dtv\bigl(\HR(x,\cdot),\HR(y,\cdot)\bigr)
 \le1-\alpha_0.
\end{equation}
One may take $\alpha_0=1/1000$.
\end{lemma}

\begin{proof}
The conclusion is immediate if $x=y$, so assume that $x\ne y$.  By
Lemma~\ref{lem:local-radius}~\textup{(iii)}, there is a universal
$C_{\mathrm{loc}}\ge1$ such that
\begin{equation}\label{eq:local-radius-two-bounds}
 \rK(z)
 \le
 C_{\mathrm{loc}}\sqrt n\,\operatorname{dist}(z,\partial K),
 \qquad
 |\rK(z)-\rK(z')|
 \le
 C_{\mathrm{loc}}\sqrt n\,\|z-z'\|_2.
\end{equation}
Set $c_0=(64C_{\mathrm{loc}})^{-1}$ and suppose that
$d_{\rK}(x,y)<c_0/\sqrt n$.  Choose an absolutely continuous curve
$\gamma$ from $x$ to $y$ whose intrinsic length is less than
$c_0/\sqrt n$.  Since $\rK\circ\gamma$ is positive and absolutely
continuous, the chain rule gives
\[
 \left|\frac{d}{dt}\log\rK(\gamma(t))\right|=\frac{\left|\frac{d}{dt}\rK(\gamma(t))\right|}{\rK(\gamma(t))}
 \le
 C_{\mathrm{loc}}\sqrt n\,
 \frac{\|\gamma'(t)\|_2}{\rK(\gamma(t))}
 \qquad\text{for almost every }t.
\]
Integrating along any initial portion of the curve gives
\[
 \left|\log\frac{\rK(\gamma(t))}{\rK(x)}\right|\leq C_{\mathrm{loc}}\sqrt{n}\int_0^t \frac{\|\gamma'(s)\|_2}{\rK(\gamma(s))}\,ds\leq C_{\mathrm{loc}}\sqrt{n} \cdot \frac{c_0}{\sqrt{n}}
 <
 C_{\mathrm{loc}}c_0
 =
 \frac1{64}.
\]
Thus, $\rK(\gamma(t))$ remains between $\rK(x)/2$ and $2\rK(x)$ along
the entire curve.  By~\eqref{eq:rK_bounds}, we obtain that
\begin{equation}\label{eq:intrinsic-to-euclidean}
 \frac12\rK(x)\le\rK(y)\le2\rK(x),
 \qquad
 \|x-y\|_2\leq \max_{z\in \gamma} \rK(z)\cdot \int_0^1\frac{\|\gamma'(t)\|_2}{\rK(\gamma(t))}dt
 <
 \frac{2c_0\rK(x)}{\sqrt n}.
\end{equation}

Let $p,q$ be the endpoints of the chord cut out by the line through $x$ and
$y$, labeled so that the points occur in the order $p,x,y,q$. By~\eqref{eq:local-radius-two-bounds} and $\rK(p)=0$, we have
\begin{align}\label{eq:bound_x_to_p}
    \|x-p\|_2\ge \frac{|\rK(x)-\rK(p)|}{C_{\rm loc}\sqrt{n}}=\frac{\rK(x)}{C_{\mathrm{loc}}\sqrt n}.
\end{align}
By~\eqref{eq:local-radius-two-bounds}, $\rK(q)=0$, and~\eqref{eq:intrinsic-to-euclidean}, we have
\begin{align}\label{eq:bound_q_to_y}
    \|q-y\|_2 \geq \frac{|\rK(y)-\rK(q)|}{C_{\rm loc}\sqrt{n}}=\frac{\rK(y)}{C_{\mathrm{loc}}\sqrt n}\geq \frac{\rK(x)}{2C_{\mathrm{loc}}\sqrt n}.
\end{align}
Because $\|p-q\|_2=\|x-p\|_2+\|x-y\|_2+\|y-q\|_2$, the cross-ratio distance satisfies
\begin{align*}
 d_K^{\mathrm{cr}}(x,y)=&~
 \frac{\|x-y\|_2\|p-q\|_2}{\|x-p\|_2\|y-q\|_2}\\
 = &~
 \frac{\|x-y\|_2}{\|x-p\|_2}+\frac{\|x-y\|_2}{\|y-q\|_2}+\frac{\|x-y\|_2^2}{\|x-p\|_2\|y-q\|_2}\\
 < &~ 2C_{\rm loc} c_0 + 4C_{\rm loc} c_0 + 8C_{\mathrm{loc}}^2c_0^2 \\
 = &~  \frac{49}{512} < \frac{1}{8},
\end{align*}
where the third step follows from~\eqref{eq:intrinsic-to-euclidean}--\eqref{eq:bound_q_to_y}, and the last step follows from $c_0=1/(64C_{\rm loc})$.
Moreover, $c_0\le1/64$, so Lemma~\ref{lem:step-quantile} and
\eqref{eq:intrinsic-to-euclidean} give
\[
 \|x-y\|_2
 <
 \frac{2c_0\rK(x)}{\sqrt n}
 \le
 \frac{\rK(x)}{32\sqrt n}
 \le
 \frac{F_{\mathrm{HR}}(x)}{\sqrt n}
 <
 \frac2{\sqrt n}
 \max\{F_{\mathrm{HR}}(x),F_{\mathrm{HR}}(y)\}.
\]

Therefore, the two hypotheses in~\eqref{eq:lovasz_overlap_conditions} hold. By Lemma~\ref{lem:lovasz_overlap}, we have
\[
 \dtv\bigl(\HRzero(x,\cdot),\HRzero(y,\cdot)\bigr)
 <
 1-\frac1{500}.
\]
Finally, $\HR(z,\cdot)=\frac12\delta_z+\frac12\HRzero(z,\cdot)$.
Convexity of total variation yields
\begin{align*}
 \dtv\bigl(\HR(x,\cdot),\HR(y,\cdot)\bigr)
 &\le
 \frac12\dtv(\delta_x,\delta_y)
 +
 \frac12
 \dtv\bigl(\HRzero(x,\cdot),\HRzero(y,\cdot)\bigr)\\
 &<
 1-\frac1{1000},
\end{align*}
which completes the proof of the lemma.
\end{proof}

\section{From local-radius isoperimetry to hit-and-run mixing}
\label{sec:conductance}
The preceding sections provide weighted expansion and local overlap in the same intrinsic geometry. The standard low-escape-set argument combines them: small ergodic flow produces large low-escape cores, local overlap forces those cores to be separated in \(d_{\rK}\), and the one-sided weighted coarea inequality converts that separation into a flow lower bound. 

We prove the following general proposition, which will be useful for analyzing both uniform and log-affine sampling.

\begin{proposition}[Intrinsic conductance transfer]
\label{prop:conductance-transfer}
Let $\lambda$ be a probability measure supported on a convex body $K$, with $\lambda\ll dx$, and let $w:K\to[0,\infty)$ be positive on $\operatorname{int}(K)$ with a Lipschitz zero extension.  Let $\mathsf P$ be a reversible Markov kernel with stationary law $\lambda$. For $x,y\in\operatorname{int}(K)$, define a general version of the intrinsic metric
\begin{align}\label{eq:def_intrinsic_metric_general}
 d_w(x,y):=\inf_\gamma\int_0^1\frac{\|\gamma'(t)\|_2}{w(\gamma(t))}\,dt.
\end{align}
Fix $0<s<1/4$. Suppose there are $0<r_0,\alpha_0,\gamma_s\le1$ such that
\begin{enumerate}[label=\textup{(\roman*)}]
\item $d_w(x,y)<r_0$ implies
\[
 \dtv\bigl(\mathsf P(x,\cdot),\mathsf P(y,\cdot)\bigr)
 \le1-\alpha_0;
\]
\item every measurable $E\subseteq K$ with
$\min\{\lambda(E),1-\lambda(E)\}\ge s/2$ satisfies
\[
 \cP_{w,\lambda}(E)
 \ge\gamma_s\min\{\lambda(E),1-\lambda(E)\}.
\]
\end{enumerate}
Then, the $s$-conductance of ${\sf P}$ can be lower bounded by
\begin{equation}\label{eq:conductance-transfer}
 \Phi_s(\mathsf P)
 \gtrsim
 \alpha_0r_0\gamma_s.
\end{equation}
\end{proposition}

\begin{proof}
Fix a measurable set $S\subset K$ with
$s<p:=\lambda(S)\le1/2$.  By reversibility, its ergodic flow can be written
in either direction:
\[
 Q
 :=
 \cQ_{\sf P}(S,S^c)
 =
 \int_S{\sf P}(x,S^c)\,d\lambda(x)
 =
 \int_{S^c}{\sf P}(y,S)\,d\lambda(y).
\]
If $Q\ge\alpha_0p/16$, then
\[
 \frac{Q}{p-s}
 \ge\frac{\alpha_0}{16}
 \gtrsim\alpha_0r_0\gamma_s,
\]
so assume that $Q<\alpha_0p/16$.  Define the low-escape cores
\[
 S_-:=
 \left\{
 x\in S\cap\operatorname{int}(K):
 {\sf P}(x,S^c)<\frac{\alpha_0}{4}
 \right\},
 \qquad
 S_+:=
 \left\{
 y\in S^c\cap\operatorname{int}(K):
 {\sf P}(y,S)<\frac{\alpha_0}{4}
 \right\}.
\]
Therefore, the definition of $Q$ gives
\begin{equation}\label{eq:bad-mass}
 \lambda(S\setminus S_-)
 \le\frac{4Q}{\alpha_0},
 \qquad
 \lambda(S^c\setminus S_+)
 \le\frac{4Q}{\alpha_0}.
\end{equation}
Indeed, the first inequality follows from
\begin{align*}
    Q=\int_S{\sf P}(x,S^c)\,d\lambda(x)\geq \int_{S\setminus S_{-}}{\sf P}(x,S^c)\,d\lambda(x)\geq \frac{\alpha_0}{4}\lambda(S\setminus S_-),
\end{align*}
and the second inequality uses the reverse-flow expression for $Q$.

For $x\in S_-$ and $y\in S_+$,
\[
 {\sf P}(x,S)>1-\frac{\alpha_0}{4},
 \qquad
 {\sf P}(y,S)<\frac{\alpha_0}{4}.
\]
Taking $S$ as a test set in the definition of total variation yields
\[
 \dtv\left({\sf P}(x,\cdot),{\sf P}(y,\cdot)\right)
 >
 1-\frac{\alpha_0}{4}-\frac{\alpha_0}{4}
 >
 1-\alpha_0.
\]
The contrapositive of \textup{(i)} therefore gives
\begin{equation}\label{eq:intrinsic-separation}
 d_{w}(S_-,S_+)
 \ge
 r_0.
\end{equation}
Moreover, \eqref{eq:bad-mass} and $Q<\alpha_0p/16$ imply
\begin{equation}\label{eq:surviving-masses}
 \lambda(S_-)\ge\frac{3}{4}p,
 \qquad
 \lambda(S_+)\ge\frac{3}{4}p.
\end{equation}

Let
\[
 u(x):=d_{w}(S_-,x),
 \qquad
 E_t:=\{x\in\operatorname{int}(K):u(x)<t\}.
\]
Distance to a set is $1$-Lipschitz in any metric space, so it holds that
\[
 |u(x)-u(y)|\le d_{w}(x,y).
\]
For fixed $x\in\operatorname{int}(K)$ and sufficiently small $h$, the segment $\gamma(t)=x+th$ for $t\in [0,1]$ lies in the interior. Therefore,
\begin{align*}
    |u(x)-u(y)|\le d_{w}(x,x+h)\le\int_0^1\frac{\|\gamma(t)'\|_2}{w(x+th)}\,dt=\|h\|_2 \int_0^1\frac{1}{w(x+th)}\,dt.
\end{align*}
Since $w$ is continuous and positive in the interior, this first shows that $u$ is locally Euclidean-Lipschitz. At every differentiability point $x$, taking $h=sv$ and letting $s\to0$ gives
\[
 |\langle\nabla u(x),v\rangle|
 \le \frac{\|v\|_2}{w(x)}
 \qquad(v\in\R^n).
\]
Taking the supremum over unit vectors $v$ yields
\begin{align}\label{eq:intrinsic-gradient}
    w(x)\|\nabla u(x)\|_2\le1.
\end{align}

For $0<t<r_0$, the separation
\eqref{eq:intrinsic-separation} gives
$S_-\subseteq E_t$ and $S_+\subseteq E_t^c$.  Hence, \eqref{eq:surviving-masses} implies
\[
 \min\{\lambda(E_t),1-\lambda(E_t)\}
 \ge\frac{3}{4}p
 \ge\frac{s}{2},
\]
and assumption~\textup{(ii)} gives
\begin{equation}\label{eq:P_rK_lower_bound}
 \cP_{w,\lambda}(E_t)\ge\frac{3\gamma_s}{4}p
 \qquad(0<t<r_0).
\end{equation}
Every point $x$ with $0<u(x)<r_0$ lies outside both cores.  Using
\eqref{eq:bad-mass}, the one-sided weighted coarea inequality (Proposition~\ref{prop:relaxed-perimeter-toolkit}~\textup{(ii)}), and
\eqref{eq:intrinsic-gradient}, we obtain
\[
 r_0\frac{3\gamma_s}{4}p
 \le
 \int_0^{r_0}\cP_{w,\lambda}(E_t)\,dt
 \le
 \int_{\{0<u<r_0\}}
 w(x)\|\nabla u(x)\|_2\,d\lambda(x)
 \le
 \lambda\bigl(K\setminus(S_-\cup S_+)\bigr)
 \le
 \frac{8Q}{\alpha_0}.
\]
Hence,
\[
 \frac{Q}{p-s}
 \ge
 \frac{3\alpha_0r_0\gamma_s }{32}
 \frac{p}{p-s}
 \gtrsim
 \alpha_0r_0\gamma_s.
\]

Taking the infimum over $S$ gives
\begin{align*}
    \Phi_s({\sf P})=\inf_{s<\lambda(S)\le1/2}
 \frac{
   \cQ_{{\sf P}}(S,S^c)
 }{
   \lambda(S)-s
 }\gtrsim\alpha_0r_0\gamma_s,
\end{align*}
completing the proof of the proposition.
\end{proof}

Combining with the local-radius weighted isoperimetry
of Corollary~\ref{cor:full-depth} gives the conductance lower bound.

\begin{corollary}[Conductance of the hit-and-run walk]
\label{cor:hit-and-run-conductance}
For every $0<s<1/4$,
\begin{equation}\label{eq:hit-and-run-conductance}
 \Phi_s(\HR)
 \gtrsim
 \frac{\psi_n}{n\bigl(1+\log(2/s)\bigr)}.
\end{equation}
\end{corollary}

\begin{proof}
Let $E\subset K$ be any measurable set with 
\begin{align*}
    p=\min\{\mu(E),1-\mu(E)\}\geq \frac{s}{2}.
\end{align*}
By Corollary~\ref{cor:full-depth},
\[
 \cP_{\rK}(E)
 \gtrsim
 \frac{\psi_n}{\sqrt n}
 \frac{p}{1+\log(1/p)}\ge \frac{\psi_n}{\sqrt n(1+\log(2/s))}p.
\]
By~\eqref{eq:KLS-upper}, $\psi_n\lesssim 1$. Choose a universal $c>0$ small enough that
\[
\cP_{\rK}(E)\geq \gamma_s p,\qquad
 \gamma_s
 :=c\frac{\psi_n}{\sqrt n\bigl(1+\log(2/s)\bigr)}\leq 1.
\]
Apply Proposition~\ref{prop:conductance-transfer} with $\lambda=\mu$, $w=\rK$, $\mathsf P=\HR$, $r_0=c_0/\sqrt n$, and the universal overlap constant $\alpha_0$ from Lemma~\ref{lem:overlap}. This gives \eqref{eq:hit-and-run-conductance}.
\end{proof}

\subsection{Proof of the main theorem}

\begin{theorem}[Main theorem]\label{thm:main_formal}
Let $n\ge2$, $K\subset\R^n$ be an isotropic convex body, and $\mu=\Unif(K)$ be the uniform distribution over $K$.
Let $\HR$ be the hit-and-run kernel on $K$.  For
every $M$-warm initial distribution $\mu_{\mathrm{init}}$ and every
$0<\varepsilon<1/2$,
\begin{equation}\label{eq:main}
 \tau_{\mix}(\varepsilon,\mu_{\mathrm{init}},\mu;\HR)
 \lesssim
 \frac{n^2}{\psi_n^2}
 \left(1+\log\frac{4M}{\varepsilon}\right)^2
 \log\frac{2M}{\varepsilon}.
\end{equation}
Consequently,
\begin{equation}\label{eq:main-simple}
 \tau_{\mix}(\varepsilon,\mu_{\mathrm{init}},\mu;\HR)
 \lesssim \frac{n^2}{\psi_n^2}\log^3\frac{8M}{\varepsilon}.
\end{equation}
\end{theorem}

\begin{proof}
Set
\[
 s=\frac{\varepsilon}{2M}.
\]
By Lemma~\ref{lem:LS},
\[
 \tau_{\mix}\left(\varepsilon,\mu_{\mathrm{init}},\mu;\HR\right)
 \lesssim
 \Phi_s\left(\HR\right)^{-2}\log\frac{2M}{\varepsilon}.
\]
Substituting \eqref{eq:hit-and-run-conductance} gives
\[
 \tau_{\mix}\left(\varepsilon,\mu_{\mathrm{init}},\mu;\HR\right)
 \lesssim
 \frac{n^2}{\psi_n^2}
 \left(1+\log\frac{4M}{\varepsilon}\right)^2
 \log\frac{2M}{\varepsilon},
\]
which is \eqref{eq:main}.  Because $M\ge1$ and
$0<\varepsilon<1/2$, each logarithmic factor is at most a universal
multiple of $\log(8M/\varepsilon)$.  This proves
\eqref{eq:main-simple}.
\end{proof}

\section*{Acknowledgments}
The authors would like to thank Sam Power for helpful discussions about the positivity of the hit-and-run kernel and the comparison argument. After the first version of this paper appeared on arXiv, Kook and Vempala~\cite{kook2026spectral} independently proved the spectral-gap bound for the hit-and-run walk on a convex body, which implies an $O(n^2C_{\sf PI}(\mu)\log(M/\varepsilon))$ mixing time and thus improves the logarithmic dependence obtained here.  Their proof is complementary to ours: it uses a dual characterization of the spectral gap together with the Babu\v{s}ka--Aziz and improved Poincar\'e inequalities.

\section*{AI Disclosure}
The authors prompted GPT-5.6 Sol Pro to explore whether localization methods could improve the mixing-time bound in~\cite{CE}. After iterative conversations, the model identified the current route of using Klartag's guided localization theorem~\cite{KlartagNeedles} to retain and average over a balanced family of needles. An initial approach incurred an additional $\operatorname{polylog}(n)$ factor.  More rounds of guided conversations led to a proof of the current Theorem~\ref{thm:main} using a multiscale Metropolis kernel and a Dirichlet-form comparison argument. The authors subsequently used GPT-5.6 Sol Ultra and Claude Fable 5 to understand the proofs and check the correctness. During this process, GPT-5.6 Sol Ultra also helped simplify the uniform-target proof by replacing the comparison-kernel argument with a direct application of the hit-and-run overlap lemmas in~\cite{LovaszHR,LV}; see Section~\ref{sec:intrinsic_metric}.  The multiscale comparison-kernel argument is retained for the log-affine extension in Appendix~\ref{app:log-affine}, where the chord conditionals are nonuniform. All content produced by these AI systems was verified by the authors before being incorporated into the manuscript.  The authors wrote the final proofs and the paper (with AI for polishing) and take full responsibility for the paper's correctness and presentation.

\newpage
\appendix

\section{Comparison with the ball walk}
\label{sec:ball-speedy}

In this appendix, we discuss a simple approach to proving the mixing time of the hit-and-run walk by applying the Markov chain comparison with the ball walk~\cite{RudolfUllrichComparison}, and explain why it will not give our result (Theorem~\ref{thm:main_formal}). 

We first define the \emph{Dirichlet form} associated with a Markov kernel ${\sf P}$ that is reversible with respect to $\mu$:
\begin{align*}
    \mathcal E_{\mathsf P}(f,f) :=  \langle f,(\Id-\mathsf P)f\rangle_{L^2(\mu)}=\frac{1}{2}\iint(f(x)-f(y))^2{\sf P}(x,dy) \,d\mu(x).
\end{align*}
If we take $f=\one_A$ for any subset $A\subseteq K$, then
\begin{align}\label{eq:Dirichlet_to_flow}
    \mathcal{E}_{\sf P}(\one_A,\one_A)=\int_A\int_{A^c}{\sf P}(x,dy)\,d\mu(x)={\cal Q}_{\sf P}(A,A^c).
\end{align}

Then, we give a more formal definition of the ball walk. Let $\delta>0$ be the step-size of the ball walk, and denote its kernel as ${\sf P}_{{\rm ball},\delta}$. The one-step update is as follows.
\begin{center}
\begin{minipage}{0.86\textwidth}
\hrule
\smallskip
\textbf{One ball-walk update from $X_t=x\in K$.}
\begin{enumerate}[leftmargin=2em,itemsep=2pt,topsep=4pt]
\item Sample $Y$ uniformly from $B(x,\delta)$.
\item If $Y\in K$, set $X_{t+1}=Y$; otherwise, set $X_{t+1}=x$.
\end{enumerate}
\smallskip
\hrule
\end{minipage}
\end{center}
${\sf P}_{{\rm ball},\delta}$ is reversible with respect to the uniform measure $\mu$. 

Rudolf--Ullrich proved the following comparison between the hit-and-run walk and the lazy version of the ball walk.

\begin{theorem}[{\cite[Theorem 5]{RudolfUllrichComparison}}]
For every $\delta>0$, let ${\sf P}_{{\rm ball},\delta}^{\rm L}:=(\Id+{\sf P}_{{\rm ball},\delta})/2$ be the lazy ball walk. Then, for every $f\in L^2(\mu)$,
\begin{equation}\label{eq:RU-comparison}
 \langle f, {\sf P}_{{\rm ball},\delta}^{\rm L} f\rangle_{L^2(\mu)}\geq \langle f, \HR f\rangle_{L^2(\mu)}.
\end{equation}
\end{theorem}

Thus, for any $f\in L^2(\mu)$,
\begin{align*}
    {\cal E}_{{\sf P}_{{\rm ball},\delta}}(f,f)=&~\langle f, (\Id-{\sf P}_{{\rm ball},\delta})f\rangle_{L^2(\mu)}  =  2\langle f, (\Id -{\sf P}_{{\rm ball},\delta}^{\rm L})f\rangle_{L^2(\mu)} \\
    \leq &~ 2\langle f, (\Id -\HR)f\rangle_{L^2(\mu)}=2{\cal E}_{\HR}(f,f).
\end{align*}
Using~\eqref{eq:Dirichlet_to_flow}, we get that
\begin{align*}
    {\cal Q}_{{\sf P}_{{\rm ball},\delta}}(\one_A,\one_{A^c})\leq 2{\cal Q}_{\HR}(\one_A,\one_{A^c})\qquad \text{for any}~A\subseteq K.
\end{align*}
Consequently, for any $\delta>0$ and $s\in [0,1/2)$,
\begin{align}
    \Phi_s({\sf P}_{{\rm ball},\delta})\leq 2\Phi_s(\HR).
\end{align}
Therefore, if we know a lower bound on the $s$-conductance of the ball walk, then we immediately have one for the hit-and-run walk.

There are two different analyses for the ball walk's mixing time, as explained in~\cite{KookVempalaZhangInOut}. 
\paragraph{I. Direct analysis:}
\begin{theorem}[Ball walk conductance~{\cite[Section 5]{VempalaSurvey}}]
Let $K$ be an isotropic convex body. For any $s\in (0,1/2)$, take the step-size $\delta\asymp s/\sqrt{n}$. Then 
\begin{align*}
    \Phi_s({\sf P}_{{\rm ball},\delta})\gtrsim \frac{s\psi_n}{n}.
\end{align*}
\end{theorem}
Hence, for an $M$-warm start $\mu_{\rm init}$,
\begin{align*}
    \tau_{\mix}\left(\varepsilon,\mu_{\mathrm{init}},\mu;\HR\right) \asymp \tau_{\mix}\left(\varepsilon,\mu_{\mathrm{init}},\mu;{\sf P}_{{\rm ball},\delta}\right)\lesssim \frac{n^2}{\psi_n^2} \left(\frac{M}{\varepsilon}\right)^2\log\frac{2M}{\varepsilon}.
\end{align*}
This mixing time bound is already better than~\eqref{eq:mixing_CE}, although it still has an undesirable $(M/\varepsilon)^2$ factor.

\paragraph{II: Indirect analysis through the \emph{Speedy walk}:}
The Speedy walk can be recognized as the ball walk with only \emph{proper steps}. It makes the following one-step update.
\begin{center}
\begin{minipage}{0.86\textwidth}
\hrule
\smallskip
\textbf{One Speedy-walk update from $X_t=x\in K$.}
\begin{enumerate}[leftmargin=2em,itemsep=2pt,topsep=4pt]
\item Sample $Y$ uniformly from $B(x,\delta)\cap K$
\item Set $X_{t+1}=Y$.
\end{enumerate}
\smallskip
\hrule
\end{minipage}
\end{center}
Let ${\sf P}_{\rm Spd}$ denote the Markov kernel of the Speedy walk. The argument in~\cite{KLS97} implies the following lower bound on the conductance.
\begin{theorem}[Speedy walk conductance~\cite{KLS97}]\label{thm:Speedy}
Let $K$ be an isotropic convex body. Let the step-size $\delta\asymp 1/\sqrt{n}$. Then
\begin{align*}
    \Phi({\sf P}_{\rm Spd})\gtrsim \frac{\psi_n}{n}.
\end{align*}
\end{theorem}

Theorem~\ref{thm:Speedy} gives the well-known $\widetilde{O}_{M/\varepsilon}(n^2\psi_n^{-2})$ warm-start mixing time for the ball walk. However, the Speedy walk converges to the \emph{Speedy distribution}, which is biased from $\mu$. Indeed, \cite{KLS97} uses another rejection sampling to obtain an approximately uniform sample; see~\cite[Appendix C]{KookVempalaZhangInOut} for a more detailed explanation. Thus, ${\sf P}_{\rm Spd}$ is not self-adjoint with respect to $\langle \cdot,\cdot\rangle_{L^2(\mu)}$, and the comparison argument cannot transfer it to the hit-and-run walk mixing time.

\section{One-dimensional log-concave estimates}
\label{app:one-dimensional-facts}
\onedimest*
\begin{proof}[Proof of Lemma~\ref{lem:one-dimensional-facts}]
Let $X$ have a nondegenerate log-concave density $\rho_X$ on $\R$, mean $m$, variance $\sigma^2$, cumulative distribution function $F$, and $q(x)=\min\{F(x),1-F(x)\}$. Then, the law of $Y=(X-m)/\sigma$ is an \emph{isotropic} log-concave distribution with density $\rho_Y(y)=\sigma\rho_X(m+\sigma y)$.

The standardized density estimate \cite[Lemma~5.5(a)]{LVlogconcave} gives $\rho_Y(y)\leq 1$ for any $y\in \R$, which implies that
\begin{equation}\label{eq:standard-density-facts}
 \|\rho_X\|_\infty
 =\frac1\sigma\|\rho_Y\|_\infty
 \lesssim\frac1\sigma.
\end{equation}
This proves~\eqref{eq:max-density}.

For the quantile estimate, let $F_Y$ be the cumulative distribution function of $Y$ and $q_Y(y)=\min\{F_Y(y),1-F_Y(y)\}$. \cite[Lemma~28]{CE} implies that, for $0<\delta\le1/e$,
\begin{equation}\label{eq:CE-quantile-density}
 q_Y(y)\ge\delta \qquad\Longrightarrow\qquad
 \rho_Y(y)\ge\frac{\delta}{8e}.
\end{equation}
Note that $q(x)=q_Y((x-m)/\sigma)$.  If $q(x)\le 1/e$, we can take $\delta=q(x)$ in~\eqref{eq:CE-quantile-density} and obtain $\rho(x)\geq \frac{q(x)}{8e\sigma}$.  If $q(x)>1/e$, take $\delta=1/e$ in~\eqref{eq:CE-quantile-density} and get that
\[
 \rho(x)\ge\frac{1}{8e^2\sigma}
 \gtrsim\frac{q(x)}{\sigma},
\]
because $q(x)\le1/2$.  This proves \eqref{eq:quantile-density} over
the full range $0<q(x)\le1/2$.

Finally, \cite[Lemma~5.17]{LVlogconcave} gives, for $R>1$,
\[
 \Prb\{|Y|>R\}<e^{-R+1}.
\]
Suppose first that $x\ge m+2\sigma$ and set $R=(x-m)/\sigma$.  Then
\[
 1-F(x)
 =\Prb\{Y>R\}
 \le\Prb\{|Y|>R\}
 <e^{-R+1}<\frac12,
\]
so $q(x)=1-F(x)$ and 
\begin{align*}
    x-m=\sigma R\leq \sigma\left(1+\log\frac1{q(x)}\right).
\end{align*}
If $x\le m-2\sigma$, the same argument with the left tail gives $q(x)=F(x)$ and 
\begin{align*}
    m-x\leq \sigma\left(1+\log\frac1{q(x)}\right).
\end{align*}
In the remaining case $|x-m|<2\sigma$, the additive constant on the right is already enough.
Thus, in all cases,
\[
 \abs{x-m}
 \lesssim\sigma\left(1+\log\frac1{q(x)}\right),
\]
which proves \eqref{eq:quantile-location}.
\end{proof}

\section{Relaxed weighted perimeter}
\label{app:relaxed-perimeter}
In this appendix, we prove Proposition~\ref{prop:relaxed-perimeter-toolkit}. Throughout this appendix, we assume that $\lambda$ is a probability measure supported on a convex body $K\subseteq\R^n$ and $\lambda\ll dx$. We begin with a technical lemma that handles the boundary of $K$.

\begin{lemma}[Boundary cutoff]\label{lem:boundary-cutoff}
Suppose that $w:K\to[0,\infty)$ is positive on $\operatorname{int}(K)$ and that its zero extension to $\R^n$ is Lipschitz.
Let $v:{\operatorname{int}(K)}\to[0,1]$ be locally Lipschitz, and suppose that
\[
 w\|\nabla v\|_2\le C
 \qquad\text{Lebesgue-almost everywhere}
\]
for some finite $C$.  Then there are functions $v_j\in\Lip(K)$, taking
values in $[0,1]$, such that $v_j\to v$ in $L^1(\lambda)$ and
\[
 \limsup_{j\to\infty}\int_Kw\|\nabla v_j\|_2\,d\lambda
 \le\int_{\operatorname{int}(K)} w\|\nabla v\|_2\,d\lambda.
\]
\end{lemma}

\begin{proof}
Let $L_w$ be the Lipschitz constant of the zero extension of $w$.  For
$\rho>0$, define
\[
 H_\rho(x)=\min\left\{1,\frac{w(x)}{\rho}\right\},
 \qquad
 v_\rho(x)=\begin{cases}
     H_\rho(x) v(x) & x\in \operatorname{int}(K),\\
     0 & x\in \partial K.
 \end{cases}
\]

We first have that $v_\rho$ is locally Lipschitz on ${\operatorname{int}(K)}$.  Indeed, 
\begin{align*}
    \|\nabla v_\rho\|_2 \leq \|H_\rho\nabla v\|_2 + \|v\nabla H_\rho\|_2= \begin{cases}
        \|\nabla v\|_2\leq C/\rho & \qquad w\geq \rho\\
        w\|\nabla v\|_2/\rho + v\|\nabla w\|_2/\rho\leq (C+L_w)/\rho  & \qquad 0<w<\rho
    \end{cases}. 
\end{align*}
Thus the gradient of $v_\rho$ is bounded almost everywhere on ${\operatorname{int}(K)}$ by \(L_\rho:=(C+L_w)/\rho\).  
 
We now turn this almost-everywhere bound into a pointwise Lipschitz bound.  Fix distinct $x,y\in{\operatorname{int}(K)}$ and write $h=y-x$.  Since the segment $[x,y]$ is compactly contained in $\operatorname{int}(K)$, all sufficiently small parallel translates $[x+z,y+z]$, with $z\in h^\perp$, remain in ${\operatorname{int}(K)}$.  By Fubini's theorem, there are offsets $z_m\in h^\perp$ tending to zero such that, for each $m$, the function
\[
  s\longmapsto v_\rho(x+z_m+sh),\qquad 0\le s\le1,
\]
is absolutely continuous and has derivative of absolute value at most
\begin{align*}
\left|\frac{d}{ds}v_\rho(x+z_m+sh)\right|= \left|\langle\nabla v_\rho(x+z_m+sh),h\rangle\right|\leq L_\rho\|h\|_2 \qquad \text{for almost every}~s\in [0,1].
\end{align*}
The fundamental theorem of calculus now yields
\[
  |v_\rho(y+z_m)-v_\rho(x+z_m)|
  \le L_\rho\|x-y\|_2.
\]
Letting $m\to\infty$ and using the continuity of $v_\rho$ proves the same inequality for $x$ and $y$. In dimension one, the same argument applies directly to the unshifted interval.  Hence $v_\rho$ is Lipschitz on ${\operatorname{int}(K)}$.  Furthermore, since $v\in [0,1]$ and the zero extension of $w$ is $L_w$-Lipschitz, we have
\[
 |v_\rho(x)|\le\frac{w(x)}\rho
 \le\frac{L_w}{\rho}\operatorname{dist}(x,\partial K).
\]
 As $x$ approaches any point of $\partial K$, the right-hand side tends to
 zero.  Defining $v_\rho=0$ on the boundary therefore gives a continuous
 extension.  Letting an interior point tend to a boundary endpoint in the
 preceding segment inequality shows that the same Lipschitz bound holds on
 all of $K$.  Thus $v_\rho\in\Lip(K)$.

As $\rho\downarrow0$, dominated convergence gives
$v_\rho\to v$ in $L^1(\lambda)$.  We also have
\begin{equation}\label{eq:boundary-cutoff-energy}
\begin{aligned}
 \int_K w\|\nabla v_\rho\|_2\,d\lambda
 \le &~ 
 \int_{\operatorname{int}(K)} H_\rho w\|\nabla v\|_2\,d\lambda
 +\frac1\rho
   \int_{\{0<w<\rho\}}w\|\nabla w\|_2\,d\lambda\\
 \le &~
 \int_{\operatorname{int}(K)} H_\rho w\|\nabla v\|_2\,d\lambda
 +L_w\,\lambda(\{0<w<\rho\}).
\end{aligned}
\end{equation}
The sets $\{0<w<\rho\}$ decrease to the empty set as $\rho\downarrow0$, because $w>0$ on ${\operatorname{int}(K)}$.  Since $\lambda$ is finite, the last term tends to zero.  Moreover, $H_\rho\to1$ almost everywhere and $H_\rho w\|\nabla v\|_2\le w\|\nabla v\|_2$, so dominated convergence handles the first term.  Consequently,
\begin{equation}\label{eq:boundary-cutoff-limit}
 \limsup_{\rho\downarrow0}
 \int_K w\|\nabla v_\rho\|_2\,d\lambda
 \le
 \int_{\operatorname{int}(K)} w\|\nabla v\|_2\,d\lambda.
\end{equation}

Choose, for example, $\rho_j=1/j$ and define $v_j=v_{\rho_j}$.  The construction above shows that $v_j\in\Lip(K)$ and $0\le v_j\le1$.  The $L^1$ convergence established above gives $v_j\to v$, and \eqref{eq:boundary-cutoff-limit} gives the desired inequality in the lemma.
\end{proof}

Next, we record two general properties of the relaxed weighted perimeter that the main proof uses repeatedly.
\begin{lemma}[Lower semicontinuity and Lipschitz coarea]
\label{lem:relaxed-lsc-coarea}
Let \(w:K\to[0,\infty)\) be as in
Proposition~\ref{prop:relaxed-perimeter-toolkit}.
\begin{enumerate}[label=\textup{(\alph*)}]
\item If $\mathbf1_{E_j}\rightarrow\mathbf1_E$ in $L^1(\lambda)$, then
\begin{equation}\label{eq:relaxed-lsc}
 \mathcal P_{w,\lambda}(E) \le \liminf_{j\to\infty}\mathcal P_{w,\lambda}(E_j).
\end{equation}

\item If \(f\in\operatorname{Lip}(K)\) and \(0\le f\le1\), then
\begin{equation}\label{eq:global-lipschitz-coarea}
 \int_0^1\mathcal P_{w,\lambda}(\{f>t\})\,dt
 \le
 \int_K w(x)\|\nabla f(x)\|_2\,d\lambda(x).
\end{equation}

\end{enumerate}
\end{lemma}

\begin{proof}
For part \textup{(a)}, there is nothing to prove if the right-hand side
of \eqref{eq:relaxed-lsc} is $\infty$.  Otherwise, choose a subsequence \(j_m\) such that
\[
 \mathcal P_{w,\lambda}(E_{j_m})\le \liminf_{j\to\infty}\mathcal P_{w,\lambda}(E_j)+\frac1m.
\]
For each \(m\), the definition of $\mathcal P_{w,\lambda}$ in~\eqref{eq:relaxed-perimeter} provides \(g_m\in\operatorname{Lip}(K)\) with \(0\le g_m\le1\) such that
\[
 \|g_m-\mathbf1_{E_{j_m}}\|_{L^1(\lambda)}\le\frac1m,\qquad \text{and}\qquad 
 \int_Kw\|\nabla g_m\|_2\,d\lambda
 \le
 \mathcal P_{w,\lambda}(E_{j_m})+\frac1m.
\]
Since \(\mathbf1_{E_{j_m}}\to\mathbf1_E\) in \(L^1(\lambda)\), we have
\(g_m\to\mathbf1_E\) in \(L^1(\lambda)\).  Hence, $(g_m)$ is admissible for $\mathcal P_{w,\lambda}(E)$ and
\[
 \mathcal P_{w,\lambda}(E)
 \le
 \liminf_{m\to\infty}
 \int_Kw\|\nabla g_m\|_2\,d\lambda
 \le \liminf_{j\to\infty}\mathcal P_{w,\lambda}(E_j).
\]

For part \textup{(b)}, define for \(t\in(0,1)\) and \(h>0\), 
\[
 R_{t,h}(s)
 :=
 \begin{cases}
 0, & s\le t,\\[1mm]
 (s-t)/h, & t<s<t+h,\\[1mm]
 1, & s\ge t+h.
 \end{cases}
\]
Then \(R_{t,h}\circ f\in\operatorname{Lip}(K)\), takes values in
\([0,1]\), and converges pointwise to \(\mathbf1_{\{f>t\}}\) as
\(h\downarrow0\).  Dominated convergence therefore gives convergence in
\(L^1(\lambda)\).  The Lipschitz chain rule gives
\[
 \int_K
 w\|\nabla(R_{t,h}\circ f)\|_2\,d\lambda
 \le
 \frac1h
 \int_{\{t<f<t+h\}}
 w\|\nabla f\|_2\,d\lambda.
\]
Consequently,
\begin{equation}\label{eq:ramp-upper-bound}
 \mathcal P_{w,\lambda}(\{f>t\})
 \le
 \liminf_{h\downarrow0}
 \frac1h
 \int_{\{t<f<t+h\}}
 w\|\nabla f\|_2\,d\lambda.
\end{equation}
Note that the map
\[
 t\longmapsto\mathbf1_{\{f>t\}}
\]
is Borel as an \(L^1(\lambda)\)-valued map, and part \textup{(a)} implies
that \(\mathcal P_{w,\lambda}\) is lower semicontinuous on indicator functions.
Thus
\[
 t\longmapsto\mathcal P_{w,\lambda}(\{f>t\})
\]
is measurable.  Fix any sequence \(h_j\downarrow0\).  Fatou's lemma followed by Tonelli's theorem gives
\begin{align*}
 \int_0^1\mathcal P_{w,\lambda}(\{f>t\})\,dt \le &~ 
 \liminf_{j\to\infty}
 \int_0^1
 \frac1{h_j}
 \int_{\{t<f<t+h_j\}}
 w\|\nabla f\|_2\,d\lambda\,dt\\
 = &~ 
 \liminf_{j\to\infty}
 \int_K
 w(x)\|\nabla f(x)\|_2
 \frac{
   |(f(x)-h_j,f(x))\cap(0,1)|
 }{h_j}
 \,d\lambda(x)\\
 &\le
 \int_Kw(x)\|\nabla f(x)\|_2\,d\lambda(x),
\end{align*}
where the first step uses Fatou's lemma, the second step uses Tonelli's theorem, and the third step follows from $|(f(x)-h_j,f(x))\cap (0,1)|\leq h_j$.  This proves \eqref{eq:global-lipschitz-coarea}.
\end{proof}

With the boundary cutoff and these two properties in hand, we can prove
the proposition.

\propperimeter*

\begin{proof}

\textbf{Proof of \textup{(i)}.} Let $(f_k)$ be any sequence admissible for $\cP_{w,\lambda}(E)$.  Then $1-f_k\in\Lip(K)$, takes values in $[0,1]$, and
\[
 \|(1-f_k)-\one_{K\setminus E}\|_{L^1(\lambda)}
 =\|f_k-\one_E\|_{L^1(\lambda)}
 \longrightarrow0.
\]
Thus, $(1-f_k)$ is admissible for $\cP_{w,\lambda}(K\setminus E)$, and
\begin{align*}
    \cP_{w,\lambda}(K\setminus E)\leq \liminf_{k\to\infty} \int_K w\|\nabla (1-f_k)\|_2\,d\lambda = \liminf_{k\to\infty} \int_K w\|\nabla f_k\|_2\,d\lambda.
\end{align*}
Taking the infimum over all admissible sequences $(f_k)$ for $E$ gives $\cP_{w,\lambda}(K\setminus E)\le\cP_{w,\lambda}(E)$.  Replacing $E$ by $K\setminus E$ gives the reverse inequality, and hence, $\cP_{w,\lambda}(E)=\cP_{w,\lambda}(K\setminus E)$.

\textbf{Proof of \textup{(ii)}.}
We first assume that \(a,b\in\mathbb R\).  Define the decreasing
truncation
\[
 \vartheta_{a,b}(s)
 :=
 \begin{cases}
 1, & s\le a,\\[1mm]
 (b-s)/(b-a), & a<s<b,\\[1mm]
 0, & s\ge b,
 \end{cases}
\]
and set $v:=\vartheta_{a,b}\circ u$ on $\operatorname{int}(K)$. Then \(v:\operatorname{int}(K)\to[0,1]\) is locally Lipschitz.  By the Lipschitz chain
rule,
\begin{equation}\label{eq:truncated-u-gradient}
 w(x)\|\nabla v(x)\|_2
 \le
 \frac{1}{b-a}
 \mathbf1_{\{a<u<b\}}(x)\,
 w(x)\|\nabla u(x)\|_2
\end{equation}
for almost every \(x\in\operatorname{int}(K)\).  At points in the level sets
\(\{u=a\}\) or \(\{u=b\}\), the same inequality holds because a locally
Lipschitz function has zero gradient almost everywhere on each of its
level sets. Then, we apply the boundary-cutoff lemma (Lemma~\ref{lem:boundary-cutoff}) to \(v\).  It gives functions
\(v_j\in\operatorname{Lip}(K)\) with \(0\le v_j\le1\) such that
\begin{equation}\label{eq:truncation-recovery-energy}
 v_j\longrightarrow v
 \qquad\text{in }L^1(\lambda),\qquad \text{and}\qquad 
 \limsup_{j\to\infty}
 \int_Kw\|\nabla v_j\|_2\,d\lambda
 \le
 \int_{\operatorname{int}(K)} w\|\nabla v\|_2\,d\lambda.
\end{equation}
After passing to a subsequence, we may assume that \(v_j\to v\)
\(\lambda\)-almost everywhere.  For almost every \(t\in(0,1)\), the level
set \(\{v=t\}\) has \(\lambda\)-measure zero.  For every such \(t\),
\[
 \mathbf1_{\{v_j>t\}}
 \longrightarrow
 \mathbf1_{\{v>t\}}
 \qquad\text{in }L^1(\lambda).
\]
The lower semicontinuity
\eqref{eq:relaxed-lsc} therefore gives
\[
 \mathcal P_{w,\lambda}(\{v>t\})
 \le
 \liminf_{j\to\infty}
 \mathcal P_{w,\lambda}(\{v_j>t\})
\]
for almost every \(t\).  Fatou's lemma and
\eqref{eq:global-lipschitz-coarea} imply
\begin{align}
 \int_0^1\mathcal P_{w,\lambda}(\{v>t\})\,dt
 \le &~ 
 \liminf_{j\to\infty}
 \int_0^1\mathcal P_{w,\lambda}(\{v_j>t\})\,dt \notag\\
 \le &~ 
 \liminf_{j\to\infty}
 \int_Kw\|\nabla v_j\|_2\,d\lambda \notag\\
 \le &~ 
 \int_{\operatorname{int}(K)} w\|\nabla v\|_2\,d\lambda.
 \label{eq:coarea-after-cutoff}
\end{align}
For \(0<t<1\), the definition of \(\vartheta_{a,b}\) gives
\[
 \{v>t\}
 =
 \{u<b-(b-a)t\}.
\]
The change of variables $s=b-(b-a)t$ then yields
\[
 \int_0^1\mathcal P_{w,\lambda}(\{v>t\})\,dt
 =
 \frac1{b-a}
 \int_a^b\mathcal P_{w,\lambda}(\{u<s\})\,ds.
\]
On the other hand, \eqref{eq:truncated-u-gradient} gives
\[
 \int_{\operatorname{int}(K)} w\|\nabla v\|_2\,d\lambda
 \le
 \frac1{b-a}
 \int_{\{a<u<b\}}
 w\|\nabla u\|_2\,d\lambda.
\]
Substituting these two identities into
\eqref{eq:coarea-after-cutoff} and multiplying by \(b-a\) proves
\[
 \int_a^b\mathcal P_{w,\lambda}(\{u<s\})\,ds
 \le
 \int_{\{a<u<b\}}
 w\|\nabla u\|_2\,d\lambda.
\]
If \(a=-\infty\) or \(b=+\infty\), apply the finite-endpoint result on an
increasing sequence of finite intervals and pass to the limit by
monotone convergence.

\textbf{Proof of~\textup{(iii)}.}
There is nothing to prove if \(\mathcal P_{w,\lambda}(E)=\infty\).  Assume
\(\mathcal P_{w,\lambda}(E)<\infty\).  By the definition of $\mathcal P_{w,\lambda}$ in~\eqref{eq:relaxed-perimeter}, there exist functions $f_k\in\operatorname{Lip}(K)$ with $0\leq f_k\leq 1$ such that
\begin{equation}\label{eq:slicing-summable-error}
 \|f_k-\mathbf1_E\|_{L^1(\lambda)}\le2^{-k}
\end{equation}
and
\begin{equation}\label{eq:slicing-recovery}
 \int_Kw(x)\|\nabla f_k(x)\|_2\,d\lambda(x)
 \longrightarrow
 \mathcal P_{w,\lambda}(E).
\end{equation}
Disintegration of $\lambda$ gives
\[
 \int
 \left(
  \int_I|f_k-\mathbf1_E|\,d\lambda_I
 \right)d\pi(I)
 =
 \|f_k-\mathbf1_E\|_{L^1(\lambda)}.
\]
Summing over \(k\) and using
\eqref{eq:slicing-summable-error}, Tonelli's theorem yields
\[
 \sum_{k=1}^\infty
 \int_I|f_k-\mathbf1_E|\,d\lambda_I
 <\infty
\]
for \(\pi\)-almost every \(I\).  Hence, on almost every needle,
\begin{equation}\label{eq:fiberwise-convergence}
 f_k|_I
 \longrightarrow
 \mathbf1_{E\cap I}
 \qquad\text{in }L^1(\lambda_I).
\end{equation}

For each \(k\), Rademacher's theorem gives a Borel set
\(D_k\subseteq\operatorname{int}(K)\) of full Lebesgue measure on which \(f_k\) is
differentiable.  Since $\lambda\ll dx$ and the sequence \((f_k)\) is countable,
the intersection $D=\bigcap_{k=1}^\infty D_k$ satisfies \(\lambda(D)=1\).  By disintegration,
\[
 0=\lambda(D^c)
   =\int\lambda_I(D^c)\,d\pi(I).
\]
The integrand is nonnegative, so $\lambda_I(D)=1$
for \(\pi\)-almost every \(I\).  Thus, on one common
\(\pi\)-full family of needles, every \(f_k\) is ambiently
differentiable at \(\lambda_I\)-almost every point.

Fix such a needle, and choose a unit-speed parametrization
\(\gamma_I:J_I\to I\).  At every
\(t\) for which \(f_k\) is differentiable at \(\gamma_I(t)\), the chain
rule gives
\[
 |(f_k\circ\gamma_I)'(t)|=\left|\left\langle \nabla f_k(\gamma_I(t)),\gamma_I'(t) \right\rangle\right|\leq \|\nabla f_k(\gamma_I(t))\|_2,
\]
where the inequality follows from $\|\gamma_I'(t)\|_2=1$.  Because
\(\lambda_I\) is absolutely continuous with respect to arclength, this
inequality holds \(\lambda_I\)-almost everywhere.  Consequently,
\begin{equation}\label{eq:fiber-energy-domination}
 \int_I w\,\left|(f_k|_I)'\right|\,d\lambda_I
 \le
 \int_Iw\|\nabla f_k\|_2\,d\lambda_I.
\end{equation}

Choose a Borel representative of \(\nabla f_k\), setting it equal to
zero where \(f_k\) is not differentiable. Then, $w(x)\|\nabla f_k(x)\|_2$ is a nonnegative Borel function.  By the measurable disintegration property,
\[
 I\longmapsto \int_Iw\|\nabla f_k\|_2\,d\lambda_I
\]
is a measurable function of \(I\), and
\[
 \int \left( \int_Iw\|\nabla f_k\|_2\,d\lambda_I \right)\,d\pi(I)
 =
 \int_Kw\|\nabla f_k\|_2\,d\lambda.
\]

By \eqref{eq:fiberwise-convergence}, the restrictions \(f_k|_I\) form
an admissible sequence for $\mathcal{P}_{w,\lambda_I}$.  Thus, by~\eqref{eq:fiber-energy-domination},
\[
 \mathcal P_{w,\lambda_I}(E\cap I)
 \le
 \liminf_{k\to\infty}
 \int_I w\,\left|(f_k|_I)'\right|\,d\lambda_I
 \le
 \liminf_{k\to\infty}
 \int_Iw\|\nabla f_k\|_2\,d\lambda_I.
\]
Using the assumed lower bound $\cP_{w,\lambda_I}(E\cap I)\ge H(I)$ for almost every $I$, we obtain that
\[
 H(I)
 \le
 \liminf_{k\to\infty}
 \int_Iw\|\nabla f_k\|_2\,d\lambda_I
\]
for almost every \(I\). Fatou's lemma now gives
\begin{align*}
 \int H(I)\,d\pi(I)
 &\le
 \int
 \liminf_{k\to\infty}
 \left(
  \int_Iw\|\nabla f_k\|_2\,d\lambda_I
 \right)d\pi(I)\\
 &\le
 \liminf_{k\to\infty}
 \int
 \left(
  \int_Iw\|\nabla f_k\|_2\,d\lambda_I
 \right)d\pi(I)\\
 &=
 \liminf_{k\to\infty}
 \int_Kw\|\nabla f_k\|_2\,d\lambda\\
 &=
 \mathcal P_{w,\lambda}(E),
\end{align*}
where the last equality follows from
\eqref{eq:slicing-recovery}.  This proves
\eqref{eq:slicing}.
\end{proof}

\section{Properties of the local conductance radius}
\label{app:proof_local_cond}

In this appendix, we prove Lemma~\ref{lem:local-radius} following the idea of~\cite{LV}.

\lemlocalconductance*

\begin{proof}
We prove the four properties in order.

\textbf{Proof of~\textup{(i)}.} We first claim that the admissible radii for $\rK(x)$ in~\eqref{eq:local-radius} form an interval.  Indeed, if $U\sim\Unif(B(0,1))$ and $0<r_1\le r_2$, convexity gives
\[
 x+r_2U\in K
 \qquad\Longrightarrow\qquad
 x+r_1U \in K.
\]
Thus, the map $r\mapsto\lambda(x,r)$ is nonincreasing, and every
$0<r<\rK(x)$ satisfies $\lambda(x,r)\ge63/64$.
For \(x\in{\operatorname{int}(K)}\), the ball of radius
\(\operatorname{dist}(x,\partial K)>0\) centered at \(x\) lies in \(K\), whereas the $\lambda(x,r)$ tends to zero as \(r\to\infty\). Hence, \eqref{eq:radius-positive-finite} is proved.

\textbf{Proof of \textup{(ii)}.} We follow the Brunn--Minkowski argument of~\cite[Lemma~3.1]{LV}.  Fix $x_0,x_1\in K$ and choose
admissible radii $r_i<\rK(x_i)$; when $\rK(x_i)=0$, set $r_i=0$ and
interpret $K\cap B(x_i,r_i)$ as the singleton $\{x_i\}$.  For
$0\le\theta\le1$, write
\[
 x_\theta=(1-\theta)x_0+\theta x_1,
 \qquad
 r_\theta=(1-\theta)r_0+\theta r_1.
\]
Convexity of both $K$ and Euclidean balls gives
\[
 (1-\theta)\bigl(K\cap B(x_0,r_0)\bigr)
 +\theta\bigl(K\cap B(x_1,r_1)\bigr)
 \subseteq K\cap B(x_\theta,r_\theta).
\]
The Brunn--Minkowski inequality and admissibility of $r_0,r_1$ therefore
imply
\[
 \begin{aligned}
 \vol\bigl(K\cap B(x_\theta,r_\theta)\bigr)^{1/n}
 &\ge (1-\theta)\vol\bigl(K\cap B(x_0,r_0)\bigr)^{1/n}
      +\theta\vol\bigl(K\cap B(x_1,r_1)\bigr)^{1/n}\\
 &\ge \left(\frac{63}{64}v_n\right)^{1/n}r_\theta.
 \end{aligned}
\]
Thus $r_\theta$ is admissible at $x_\theta$.  Letting the two chosen radii
increase to $\rK(x_0)$ and $\rK(x_1)$ proves
\[
 \rK(x_\theta)\ge(1-\theta)\rK(x_0)+\theta\rK(x_1),
\]
which is part~\textup{(ii)}.

\textbf{Proof of \textup{(iii)}.} Fix
\(x\in\operatorname{int}(K)\), and let
\[
 d=\operatorname{dist}(x,\partial K).
\]
Choose a closest boundary point and a supporting half-space
\(G\supseteq K\) at that point.  The closed ball $B(x,d)$ lies in $K$ and
touches the boundary at the chosen point.  A supporting hyperplane of $K$
there also supports this ball, so it is tangent to the ball and lies at
distance exactly $d$ from $x$.  Hence, for a suitable unit vector \(e\) and
\(U\sim\Unif(B(0,1))\),
\[
 \frac{\vol(B(x,r)\cap G)}{\vol(B(x,r))}
 =
 \Prb\left\{\langle U,e\rangle\ge-\frac dr\right\}.
\]
Taking $(n-1)$-dimensional cross-sections of the unit ball perpendicular to
$e$ shows that \(\langle U,e\rangle\) has the symmetric density
\[
 f_n(s)
 =
 \frac{v_{n-1}}{v_n}
 (1-s^2)^{(n-1)/2}\one_{\{|s|\le1\}},
\]
and
\[
 \|f_n\|_\infty
 =
 \frac{v_{n-1}}{v_n}
 =\frac1{\sqrt\pi}
   \frac{\Gamma(n/2+1)}{\Gamma((n+1)/2)}
 \lesssim\sqrt n.
\]
Therefore
\[
 \frac{\vol(B(x,r)\cap G)}{\vol(B(x,r))}
 =
 \frac12+\int_0^{d/r}f_n(s)\,ds
 \le
 \frac12+O\left(\frac{\sqrt n\,d}{r}\right).
\]
For a sufficiently large universal constant, this quantity is smaller than
\(63/64\) whenever \(r\gtrsim\sqrt n\,d\).  Since \(K\subseteq G\), such
an \(r\) is not admissible in \eqref{eq:local-radius}, which proves
\eqref{eq:radius-boundary-distance} for interior points.

If \(x\in\partial K\), a supporting hyperplane through \(x\) bisects every
ball centered at \(x\).  Convexity gives \(K\subseteq G\) for one of the
corresponding half-spaces, and hence
\[
 \lambda(x,r)\le\frac12
 \qquad(r>0).
\]
Thus \(\rK(x)=0\).

Finally, consider a chord \(K\cap\ell=[z_-,z_+]\), parametrized by
\[
 \gamma(s)=z_-+se,
 \qquad
 0\le s\le L:=\|z_+-z_-\|_2.
\]
Set \(f(s)=\rK(\gamma(s))\).  Part \textup{(ii)} makes \(f\) concave,
while the boundary statement gives \(f(0)=f(L)=0\).  Moreover,
\eqref{eq:radius-boundary-distance} implies
\[
 0\le f(s)\lesssim\sqrt n\min\{s,L-s\}.
\]
For \(0<s<t<L\), monotonicity of the secant slopes of a concave function
gives
\[
 -\frac{f(t)}{L-t}
 \le
 \frac{f(t)-f(s)}{t-s}
 \le
 \frac{f(s)}s.
\]
The preceding envelope bounds both extreme slopes in absolute value by
\(O(\sqrt n)\).  Hence
\[
 |f(t)-f(s)|\lesssim\sqrt n\,|t-s|.
\]
Every pair of points in \(K\) lies on a common chord, so
\[
 |\rK(x)-\rK(y)|
 \lesssim
 \sqrt n\,\|x-y\|_2,
 \qquad x,y\in K.
\]
If \(x\in K\) and \(y\notin K\), let
\(z\in[x,y]\cap\partial K\) be the first boundary point encountered from
\(x\).  Since \(\rK(z)=0\),
\[
 |\rK(x)|
 =|\rK(x)-\rK(z)|
 \lesssim\sqrt n\,\|x-z\|_2
 \lesssim\sqrt n\,\|x-y\|_2.
\]
If both points lie outside $K$, the zero extension vanishes at both; the
case with the first point outside follows by symmetry.  This proves the
asserted Lipschitz bound in all cases.

\textbf{Proof of \textup{(iv)}.} We first apply \cite[Theorem~4.1]{KLS95} to obtain that $c_0B(0,1)\subseteq K$ for a universal \(c_0>0\).  Then, we apply
\cite[Lemma~3.4]{LV} to the rescaled body
\(K'=c_0^{-1}K\), which contains the unit ball.  Their lemma gives
\[
 \frac1{\vol(K')}\int_{K'} r_{K'}(z)\,dz
 \gtrsim \frac1{\sqrt n}.
\]
The occupancy ratio is invariant under simultaneous scaling of the body,
center, and radius, so
\(r_{K'}(x/c_0)=r_K(x)/c_0\).  Changing variables and scaling back proves
\eqref{eq:mean-local-radius}.
\end{proof}

\section{Balanced needles from Klartag's localization theorem}
\label{app:klartag-balanced-needles}

In this appendix, we prove Lemma~\ref{lem:balanced-needles}.  We need the following Euclidean consequence of Klartag's guided localization theorem. It simultaneously disintegrates the ambient measure, balances the localized function on each fiber, and uses the guiding function as arclength along every nontrivial fiber.

\begin{theorem}[Klartag's localization theorem, Euclidean specialization
\cite{KlartagNeedles}]
\label{thm:klartag-euclidean-localization}
Let $n\ge2$, let $K\subset\R^n$ be a convex body, and let
$\mu=\Unif(K)$.  If $f\in L^1(\mu)$ and $\int f\,d\mu=0$, then there exist
a $1$-Lipschitz function $u:{\operatorname{int}(K)}\to\R$, a measurable partition
$\mathcal Q$ of ${\operatorname{int}(K)}$, a measure $\nu$ on $\mathcal Q$, and measures
$\{\widetilde\mu_I\}_{I\in\mathcal Q}$ with the following properties.
For every measurable $B\subseteq{\operatorname{int}(K)}$,
\begin{equation}\label{eq:klartag-specialized-disintegration}
 \mu(B)=\int_{\mathcal Q}\widetilde\mu_I(B)\,d\nu(I).
\end{equation}
The map $I\mapsto\widetilde\mu_I(B)$ may be chosen measurable.  Moreover, for
$\nu$-almost every $I$, the measure $\widetilde\mu_I$ is supported on $I$
and
\begin{equation}\label{eq:klartag-specialized-balance}
 \int_I f\,d\widetilde\mu_I=0.
\end{equation}
Almost every piece $I$ is either a singleton or a nondegenerate line
interval.  In the latter case, there are an open interval
$J_I\subset\R$, a unit-speed parametrization $\gamma_I:J_I\to I$, and a
smooth function $\Psi_I:J_I\to\R$ such that
\begin{equation}\label{eq:klartag-specialized-coordinate-density}
 u(\gamma_I(t))=t,
 \qquad
 (\gamma_I^{-1})_\#\widetilde\mu_I
   =e^{-\Psi_I(t)}\,dt,
 \qquad
 \Psi_I''(t)\ge \frac{(\Psi_I'(t))^2}{n-1}.
\end{equation}
In particular, the density in \eqref{eq:klartag-specialized-coordinate-density}
is log-concave.
\end{theorem}

\begin{proof}
Regard ${\operatorname{int}(K)}$ as a Riemannian manifold with the Euclidean metric and with
measure $d\mu=\vol(K)^{-1}dx$.  Any two points of ${\operatorname{int}(K)}$ are joined by
the straight segment between them, so this manifold is geodesically convex.
Its Ricci curvature is zero and its density is constant.  It therefore
satisfies Klartag's curvature--dimension condition $\mathrm{CD}(0,n)$.
The remaining integrability hypothesis in \cite[Theorem~1.5]{KlartagNeedles} asks for an
$x_0\in{\operatorname{int}(K)}$ such that
\[
 \int_{\operatorname{int}(K)} |f(x)|\,\|x-x_0\|_2\,d\mu(x)<\infty;
\]
this follows from $f\in L^1(\mu)$ because $K$ is bounded.

Part~\textup{(A)} of \cite[Theorem~1.5]{KlartagNeedles} provides a
$1$-Lipschitz dual maximizer $u$.  Part~\textup{(C)} says that this same
$u$ admits a partition satisfying all the conclusions of \cite[Theorem~1.2]{KlartagNeedles}, with $u$ as the coordinate along every nontrivial piece.  Consequently,
\cite[Theorem~1.2~\textup{(i)}]{KlartagNeedles} gives
\eqref{eq:klartag-specialized-disintegration}, including measurability;
part~\textup{(ii)} gives the needle structure; and part~\textup{(iii)} is
exactly \eqref{eq:klartag-specialized-balance}.

Here is the mechanism behind the last conclusion.  A transport ray $I$ of $u$
is a maximal set on which
\[
 |u(x)-u(y)|=\|x-y\|_2\qquad(x,y\in I).
\]
The relative interiors of the nondegenerate rays form a measurable
partition $\mathcal{Q}$ of the set $\operatorname{Strain}[u]$.  Let $q(x)$ denote the
piece containing $x$, and disintegrate $\mu$ over these pieces as in
\cite[Theorem~1.4]{KlartagNeedles}:
\begin{align*}
    \int_{\operatorname{Strain}[u]} h(x)\,d\mu(x)=\int_{\cal Q} \left(\int_{I} h(x)\,d\widetilde{\mu}_I(x)\right)\,d\nu(I)
\end{align*}
for every nonnegative or integrable measurable function $h$. Set
\[
 g(I)=\int_I f\,d\widetilde\mu_I.
\]
This is an integrable function of $I$, because
\[
 \int_{\cal Q} |g(I)|\,d\nu(I)
 \le \int_{\cal Q} \int_I |f|\,d\widetilde\mu_I\,d\nu(I)
 =\int_{\operatorname{Strain}[u]}|f|\,d\mu<\infty.
\]
If $\mathcal S$ is any measurable family of pieces, note that $q^{-1}({\cal S})=\bigcup_{I\in {\cal S}}I$. Then, disintegration gives
\begin{equation}\label{eq:klartag-subfamily-balance}
\begin{aligned}
 \int_{\mathcal S}g(I)\,d\nu(I) =&~  \int_{\cal Q}\int_{I}\one_{\cal S}(I)f(x)\,d\widetilde{\mu}_I(x)\,d\nu(I)\\
 =&~ \int_{\cal Q}\int_{I}\one_{q^{-1}(\cal S)}(x)f(x)\,d\widetilde{\mu}_I(x)\,d\nu(I)\\
 = &~ \int_{\operatorname{Strain}[u]}\one_{q^{-1}(\cal S)}f\,d\mu
 = \int_{q^{-1}(\mathcal S)}f\,d\mu.
\end{aligned}
\end{equation}
The substantive transport fact in Klartag's proof is that every measurable
saturated set $q^{-1}(\mathcal S)\subseteq\operatorname{Strain}[u]$ has
zero $f$-integral
\cite[Lemma~4.6 and the proof of Theorem~1.5~\textup{(B)},
pp.~58--59]{KlartagNeedles}.  Thus the right-hand side of
\eqref{eq:klartag-subfamily-balance} is zero for every $\mathcal S$.
Taking $\mathcal{S}_+=\{I:g(I)>0\}$ gives $\nu({\cal S}_+)=0$, and taking
$\mathcal{S}_-=\{I:g(I)<0\}$ and nonnegative function $-g$ gives $\nu({\cal S}_-)=0$. Consequently, $g=0$ almost everywhere.  This is the
conditional balance \eqref{eq:klartag-specialized-balance}.  

\cite[Theorem~1.5~\textup{(B)}]{KlartagNeedles} also shows that $f=0$ almost everywhere outside $\operatorname{Strain}[u]$.  On this complement, push forward the restricted
measure by $x\mapsto\{x\}$ and give the singleton fiber $\{x\}$ the measure
$\delta_x$.  Adjoining these fibers into $\mathcal{Q}$ preserves conditional balance and gives the full partition of $\operatorname{int}(K)$.

It remains only to translate the geometric terminology.  \cite[Theorem~1.4~\textup{(iii)}]{KlartagNeedles} allows the geodesic in the definition of each
needle to be chosen so that $u(\gamma_I(t))=t$.  In Euclidean space this
geodesic is a unit-speed parametrization of a line interval.  With this same
parametrization, the definition of a $\mathrm{CD}(0,n)$ needle says that the
pullback of $\widetilde\mu_I$ has density $e^{-\Psi_I}$ and
\[
 \Psi_I''\ge\frac{(\Psi_I')^2}{n-1}\ge0.
\]
Hence, $\Psi_I$ is convex and $e^{-\Psi_I}$ is log-concave.  This proves
\eqref{eq:klartag-specialized-coordinate-density} and completes the proof of the theorem.
\end{proof}

We can now use Theorem~\ref{thm:klartag-euclidean-localization} to prove the lemma. 

\lemklartag*
\begin{proof}
The case $n=1$ is immediate.  Assume $n\ge2$, let $\mu=\Unif(K)$, and write
\[
 g:=\frac{d\lambda}{d\mu},
 \qquad
 p:=\lambda(A).
\]
The density $g$ has a positive log-concave representative on $\operatorname{int}(K)$.  Since $\lambda(\partial K)=0$, we may replace $A$ by $A\cap\operatorname{int}(K)$.  Apply Theorem~\ref{thm:klartag-euclidean-localization} to
\[
 f:=g(\one_A-p),
 \qquad
 \int f\,d\mu=0.
\]
Let $u$, $\mathcal Q$, $\xi$, and $\widetilde\mu_I$ be the resulting objects.  Conditional balance gives
\begin{equation}\label{eq:balanced-unnormalized-mass}
 0=\int_Ig(\one_A-p)\,d\widetilde\mu_I
 \qquad\text{for $\xi$-almost every }I.
\end{equation}
Set
\[
 z_I:=\int_Ig\,d\widetilde\mu_I.
\]
Disintegration gives $\int z_I\,d\xi(I)=1$, so $z_I<\infty$ almost everywhere.  On the measurable family $\mathcal I=\{I:0<z_I<\infty\}$, define
\[
 d\pi(I):=z_I\,d\xi(I),
 \qquad
 d\lambda_I:=\frac{g}{z_I}\,d\widetilde\mu_I.
\]
Then $\pi$ is a probability measure and, for every measurable $B\subseteq K$,
\[
 \int_{\mathcal I}\lambda_I(B)\,d\pi(I)
 =\int_Bg\,d\mu
 =\lambda(B),
\]
which proves \eqref{eq:balanced-needle-disintegration}.  Equation~\eqref{eq:balanced-unnormalized-mass} gives
\[
 \lambda_I(A)=p=\lambda(A)
\]
for $\pi$-almost every $I$.  A singleton fiber with $z_I>0$ cannot satisfy this identity because $0<p<1$, so all positive-mass fibers are nondegenerate.

Along the unit-speed parametrization supplied by Theorem~\ref{thm:klartag-euclidean-localization}, the measure $\widetilde\mu_I$ has a log-concave density and $g|_I$ is log-concave.  Their product is log-concave.  Reweighting does not alter the supporting interval or the guiding coordinate.  This proves the lemma.
\end{proof}

\section{Extension to log-affine target distributions}
\label{app:log-affine}
In this appendix, we prove the warm-start mixing time of sampling from a log-affine distribution over a convex body using the hit-and-run walk.  

Let $K\subset\R^n$ be a convex body,  $b\in\R^n$, and define the log-affine target density:
\begin{equation}\label{eq:log-affine-target}
 f_b(x):=e^{-\langle b,x\rangle}\one_K(x),
 \qquad
 Z_b:=\int_K e^{-\langle b,x\rangle}\,dx,
 \qquad
 d\nu_b(x):=Z_b^{-1}f_b(x)\,dx.
\end{equation}
Define
\begin{align*}
    \sigma_{\min}:=\sqrt{\lambda_{\min}(\operatorname{Cov}_{\nu_b}(X))},\qquad \sigma_{\max}:=\sqrt{\lambda_{\max}(\operatorname{Cov}_{\nu_b}(X))}.
\end{align*}
We say that $\nu_b$ is \emph{$C_0$-near-isotropic} if $\E_{\nu_b}[X]=0$ and
\begin{equation}\label{eq:log-affine-near-isotropic}
 C_0^{-1}I_n
 \preceq
 \operatorname{Cov}_{\nu_b}(X)
 \preceq
 C_0 I_n.
\end{equation}
That is, $C_0^{-1}\leq \sigma_{\min}^2$ and $C_0\geq \sigma_{\max}^2$.
The isotropic case corresponds to $\sigma_{\min}=\sigma_{\max}=1$.

The hit-and-run kernel with target $\nu_b$, denoted by ${\sf P}_b$, chooses $\theta\sim\sigma_{n-1}$ and resamples from the conditional law of $\nu_b$ on the chord $K\cap(x+\R\theta)$. By~\cite{RudolfUllrichPositivity}, ${\sf P}_b$ is reversible and positive.    

The main result of this appendix is the following theorem.
\begin{theorem}[Warm-start mixing for log-affine targets]
\label{thm:log-affine-main}
Let $n\ge2$, and suppose that $\nu_b$ is defined as in~\eqref{eq:log-affine-target}.  Then, for every $0<s<1/4$, the $s$-conductance of ${\sf P}_b$ with respect to $\nu_b$ is lower bounded by:
\begin{equation}\label{eq:log-affine-conductance}
 \Phi_s({\sf P}_b)
 \gtrsim
 \frac{\psi_n}{n\bigl(1+\log(2/s)\bigr)}.
\end{equation}
Consequently, if $\nu_{\rm init}$ is $M$-warm with respect to $\nu_b$, then
for every $0<\varepsilon<1/2$,
\begin{equation}\label{eq:log-affine-mixing-simple}
 \tau_{\mix}(\varepsilon,\nu_{\rm init},\nu_b;{\sf P}_b)
 \lesssim \frac{\sigma_{\max}^2}{\sigma_{\min}^2}
 \frac{n^2}{\psi_n^2}
 \log^3\!\left(\frac{8M}{\varepsilon}\right)\leq C_0^2 \frac{n^2}{\psi_n^2}
 \log^3\!\left(\frac{8M}{\varepsilon}\right).
\end{equation}
\end{theorem}

The rest of this appendix proves the theorem.  Unlike the argument for the uniform sampling, the proof uses a positive multiscale Metropolis kernel whose active scale is comparable with a density-sensitive local radius (see Appendix~\ref{app:multiscale_kernel}). 

\subsection{The density-sensitive local radius}

Following~\cite[Section~6.2]{LV}, for $x\in K$ and $r>0$, define
\begin{equation}\label{eq:log-affine-local-occupancy}
 \lambda_b(x,r)
 := 
 \frac{\vol\bigl(
 B(x,r)\cap \{y\in\R^n:f_b(y)\ge\beta f_b(x)\}
 \bigr)}{v_n r^n},\qquad \beta :=\frac{3}{4}
\end{equation}
and
\begin{equation}\label{eq:log-affine-local-radius}
 \rho_b(x)
 :=
 \sup\left(
 \{0\}\cup\left\{r>0:\lambda_b(x,r)\ge\frac{63}{64}\right\}
 \right).
\end{equation}
Thus, $\rho_b$ records both the distance to the support boundary and the scale on which the density can decrease by a constant factor.  When $b=0$, one has $\rho_0=\rK$.

\begin{lemma}[Geometry of the density-sensitive radius]
\label{lem:log-affine-radius}
Let $K\subseteq\R^n$ be a convex body and let $b\in\R^n$. 
\begin{enumerate}[label=\textup{(\roman*)}]
\item $\rho_b$ is positive on $\operatorname{int}(K)$, finite and concave on
$K$, and vanishes on $\partial K$;
\item the zero extension of $\rho_b$ to $\R^n$ is $O(\sqrt n)$-Lipschitz;
\item for every $x\in K$,
\begin{equation}\label{eq:log-affine-score-radius}
 \rho_b(x)\|b\|_2\lesssim \sqrt n;
\end{equation}
\item 
\begin{equation}\label{eq:log-affine-mean-radius}
 \int_K\rho_b\,d\nu_b
 \gtrsim \frac{\sigma_{\min}}{\sqrt n}.
\end{equation}
\end{enumerate}
\end{lemma}

\begin{proof}
For \textup{(i)}, let
\begin{equation}\label{eq:log-affine-moving-level-set}
 L_x:=\{y\in\R^n:f_b(y)\ge\beta f_b(x)\}
 =
 K\cap
 \{y:\langle b,y\rangle\le\langle b,x\rangle+\log(1/\beta)\}.
\end{equation}
Thus $L_x$ is convex, contains $x$, and is contained in $K$.  In particular, $\lambda_b(x,r)\le\lambda_K(x,r)$ and hence
\begin{equation}\label{eq:rho-dominated-by-uniform-radius}
 \rho_b(x)\le r_K(x).
\end{equation}
Positivity in the interior follows because, for sufficiently small $\delta>0$, $B(x,\delta)\subseteq L_x$; finiteness follows from boundedness of $K$. The concavity follows from a similar Brunn--Minkowski argument as in the proof for $r_K$.

For \textup{(ii)}, by~\eqref{eq:rho-dominated-by-uniform-radius} and~\eqref{eq:radius-boundary-distance}, we have
\[
 \rho_b(x)
 \le C\sqrt n\,\operatorname{dist}(x,\partial K).
\]
Consequently, $\rho_b=0$ on $\partial K$, and the chordwise secant-slope argument from the proof of Lemma~\ref{lem:local-radius} applies to the concave function $\rho_b$. This proves the Lipschitzness of $\rho_b$.

For~\textup{(iii)}, let $e=b/\|b\|_2$ and
$U\sim\operatorname{Unif}(B(0,1))$.  If $r<\rho_b(x)$, then
\[
 \frac{63}{64}
 \le\lambda_b(x,r) = \Prb\left\{
  x+rU\in L_x\right\}
 \le
 \Prb\left\{
  \langle e,U\rangle
  \le\frac{\log(1/\beta)}{r\|b\|_2}
 \right\}.
\]
The random variable $\langle e,U\rangle$ has density
\[
 g_n(t)
 =
 \frac{v_{n-1}}{v_n}(1-t^2)^{(n-1)/2}
 \one_{\{|t|\le1\}},
\]
which is symmetric and satisfies
\[
 \|g_n\|_\infty
 =
 \frac{v_{n-1}}{v_n}
 =
 \frac{\Gamma(n/2+1)}
 {\sqrt\pi\,\Gamma((n+1)/2)}
 =O(\sqrt{n}).
\]
Hence, for $s\ge0$,
\[
 \Prb\{\langle e,U\rangle\le s\}
 \le\frac12+O(\sqrt{n}\,s), 
\]
which implies that
\begin{align*}
    \frac{63}{64} \leq \frac{1}{2} +O\left(\sqrt{n}\frac{\log(1/\beta)}{r\|b\|_2}\right).
\end{align*}
We obtain $r\|b\|_2\le O(\sqrt n)$. Letting $r\uparrow\rho_b(x)$ proves
\eqref{eq:log-affine-score-radius}.

For \textup{(iv)}, let
\[
 m_b=\E_{\nu_b}X,
 \qquad
 \Sigma_b=\operatorname{Cov}_{\nu_b}(X)
\]
and let
\[
 T(x)=\Sigma_b^{-1/2}(x-m_b).
\]
Then the pushforward measure $\widetilde\nu_b=T_\#\nu_b$ is isotropic. Moreover, $\widetilde\nu_b$ is log-affine on
$\widetilde K=T(K)$, with tilt vector $\widetilde b=\Sigma_b^{1/2}b$. By~\cite[Lemma~5.13]{LVlogconcave}, a superlevel set of mass $1/8$ of $\widetilde \nu_b$ contains a Euclidean ball of radius $1/(8e)$.  Pulling this ball back shows that the corresponding superlevel set of $f_b$ contains a Euclidean ball of radius $\Theta(\sigma_{\min})$.  Therefore, using~\cite[Lemma~6.4]{LV} gives
\[
 \int_K\rho_b\,d\nu_b
 \gtrsim \,\frac{\sigma_{\min}}{\sqrt n}.
\]
\end{proof}

\begin{lemma}[Small-radius mass]
\label{lem:log-affine-small-radius}
For every $u>0$,
\begin{equation}\label{eq:log-affine-small-radius}
 \nu_b\{x:\rho_b(x)<u\}
 \le
 \min\left\{
  1,
  O\left(\frac{\sqrt n}{\sigma_{\min}}u\right)
 \right\}.
\end{equation}
\end{lemma}

\begin{proof}
Set
\[
 m_b:=\int_K\rho_b\,d\nu_b.
\]
Lemma~\ref{lem:berwald} gives
\[
 \int_K\rho_b^2\,d\nu_b\le2m_b^2.
\]
Hence, Paley--Zygmund gives
\[
 \nu_b\{\rho_b\ge m_b/2\}\ge\frac18.
\]
The function $H(t):=\nu_b\{\rho_b\ge t\}$ is log-concave because $\rho_b$ is concave and $\nu_b$ is log-concave.   Moreover, $H(0)=1$. Therefore, for
$0<u<m_b/2$,
\[
 H(u)
 \ge H(m_b/2)^{2u/m_b}
 \ge 8^{-2u/m_b}.
\]
It follows that
\[
 \nu_b\{\rho_b<u\}
 \le
 1-8^{-2u/m_b}
 \lesssim \frac{u}{m_b}.
\]
The same bound is trivial for $u\ge m_b/2$ after increasing the constant. Finally, Lemma~\ref{lem:log-affine-radius} gives $m_b\gtrsim \sigma_{\min}/\sqrt{n}$, which completes the proof of the lemma.
\end{proof}

\begin{corollary}
\label{cor:log-affine-radius-profile}
For every measurable $A\subseteq K$ with $p:=\min\{\nu_b(A),1-\nu_b(A)\}\in (0,1/2]$,
\begin{equation}\label{eq:log-affine-radius-profile-general}
 \cP_{\rho_b,\nu_b}(A)
 \gtrsim \frac{\psi_n}{\sqrt n}
  \frac{\sigma_{\min}}
       {\sigma_{\max}}
 \frac{p}{1+\log(1/p)}.
\end{equation}
\end{corollary}

\begin{proof}
By Lemma~\ref{lem:log-affine-radius},
\[
 \int_K\rho_b\,d\nu_b
 \gtrsim \frac{\sigma_{\min}}{\sqrt n}.
\]
Since
\[
 C_{\sf PI}(\nu_b)^{-1/2}
 \gtrsim
 \psi_{\sf ch}(\nu_b)
 \gtrsim
 \frac{\psi_n}
 {\sigma_{\max}},
\]
applying Theorem~\ref{thm:balanced-weighted} with
$\lambda=\nu_b$ and $w=\rho_b$ completes the proof.
\end{proof}

\subsection{Intrinsic separation}
Let the intrinsic metric associated with $\rho_b$ be defined as in~\eqref{eq:def_intrinsic_metric_general}:
\begin{equation}\label{eq:log-affine-intrinsic-metric}
 d_{\rho_b}(x,y)
 :=
 \inf_{\gamma}
 \int_0^1
 \frac{\|\gamma'(t)\|_2}{\rho_b(\gamma(t))}\,dt,
\end{equation}
where the infimum is over absolutely continuous curves
$\gamma:[0,1]\to\operatorname{int}(K)$ joining $x$ and $y$.
For measurable $A,B\subseteq K$, define
\[
 d_{\rho_b}(A,B)
 :=
 \inf\left\{
 d_{\rho_b}(x,y):
 x\in A\cap\operatorname{int}(K),\
 y\in B\cap\operatorname{int}(K)
 \right\}.
\]

\begin{lemma}[Local geometry of the intrinsic metric]
\label{lem:log-affine-intrinsic-local-geometry}
Let $L=O(\sqrt{n})$ be the Lipschitz parameter for $\rho_b$.  For every $x,y\in\operatorname{int}(K)$,
\begin{equation}\label{eq:log-affine-intrinsic-log-distortion}
 \left|\log\frac{\rho_b(y)}{\rho_b(x)}\right|
 \le L d_{\rho_b}(x,y)
\end{equation}
and
\begin{equation}\label{eq:log-affine-intrinsic-euclidean-distortion}
 \|x-y\|_2
 \le
 \rho_b(x)e^{L d_{\rho_b}(x,y)}d_{\rho_b}(x,y).
\end{equation}
Consequently, for a universal $c_*>0$, if
\begin{equation}\label{eq:log-affine-intrinsic-nearby}
 d_{\rho_b}(x,y)<\frac{c_*}{\sqrt n}
\end{equation}
then
\begin{equation}\label{eq:log-affine-intrinsic-local-geometry}
 \frac12\rho_b(x)\le\rho_b(y)\le2\rho_b(x),
 \qquad
 \|x-y\|_2\lesssim \frac{\rho_b(x)}{\sqrt n}.
\end{equation}
\end{lemma}

\begin{proof}
Let $\gamma$ be an absolutely continuous path of intrinsic length $\ell$.  Since $\rho_b$ is positive in the interior and $L$-Lipschitz,
\[
 \left|\frac{d}{dt}\log\rho_b(\gamma(t))\right|
 \le L\frac{\|\gamma'(t)\|_2}{\rho_b(\gamma(t))}
\]
for almost every $t$.  Integration along every initial segment gives
\[
 e^{-L\ell}\rho_b(x)
 \le\rho_b(\gamma(t))
 \le e^{L\ell}\rho_b(x).
\]
This proves \eqref{eq:log-affine-intrinsic-log-distortion} after taking an infimizing sequence.  It also gives
\[
 \|x-y\|_2
 \le\int_0^1\|\gamma'(t)\|_2\,dt
 \le\rho_b(x)e^{L\ell}\ell,
\]
which proves \eqref{eq:log-affine-intrinsic-euclidean-distortion}.  Finally, since $L=O(\sqrt{n})$ by  Lemma~\ref{lem:log-affine-radius}, we immediately get that~\eqref{eq:log-affine-intrinsic-nearby} implies~\eqref{eq:log-affine-intrinsic-local-geometry}.
\end{proof}

\begin{theorem}[Intrinsic separator estimate]
\label{thm:log-affine-separator}
Let $A,B\subseteq K$ be measurable subsets with $m:=\min\{\nu_b(A),\nu_b(B)\}\in (0,1/2)$. If
\begin{equation}\label{eq:log-affine-set-separation}
 d_{\rho_b}(A,B)
 \ge
 \frac{c_0}{\sqrt n}
\end{equation}
for a sufficiently small universal constant $c_0>0$, then
\begin{equation}\label{eq:log-affine-separator}
 \nu_b\bigl(K\setminus(A\cup B)\bigr)
 \gtrsim
 \frac{\psi_n}{n}
 \frac{\sigma_{\min}}{\sigma_{\max}}
 \frac{m}{1+\log(1/m)}.
\end{equation}
\end{theorem}
The theorem can be proved using Corollary~\ref{cor:log-affine-radius-profile} and a very similar argument as in Proposition~\ref{prop:conductance-transfer}. 

\subsection{An auxiliary multiscale Metropolis kernel}\label{app:multiscale_kernel}

For the uniform target, Lov\'asz's overlap lemma gives local overlap directly
for hit-and-run.  For a log-affine target, the chord law is a truncated
exponential with a chord-dependent normalizing constant, so this argument no
longer applies directly.

We therefore use the comparison strategy of
Appendix~\ref{sec:ball-speedy}, but with an auxiliary kernel specifically
designed to have stationary law $\nu_b$.  The kernel makes local Metropolis
moves at scales comparable with $\rho_b$: nearby points share a common
accepted proposal region, and its Dirichlet form is dominated chordwise by
that of hit-and-run.  Conductance of the auxiliary kernel thus transfers to
${\sf P}_b$. To accommodate the variation of $\rho_b$ without losing the number of scales, we use dyadic radii and activate each radius only on the corresponding $\rho_b$-band.

For $j\in\mathbb Z$, let
\begin{equation}\label{eq:log-affine-bands}
 \delta_j:=2^j,
 \qquad
 D_j:=\{x:4\delta_j\le\rho_b(x)\le64\delta_j\},
 \qquad
 q_j(v):=\frac{\one_{\{\|v\|_2\le\delta_j\}}}{v_n\delta_j^n}.
\end{equation}
Choose a sufficiently small universal $\alpha>0$ and define the off-diagonal density
\begin{equation}\label{eq:log-affine-local-kernel-density}
 k(x,y)
 :=
 \alpha\sum_{j\in\mathbb Z}
 \one_{D_j}(x)\one_{D_j}(y)q_j(y-x)
 \min\left\{1,\frac{f_b(y)}{f_b(x)}\right\}.
\end{equation}
For $x\in K$, let
\[
 a(x):=\int_K k(x,y)\,dy.
\]
Since at most five bands contain $x$ and each $q_j$ is a probability
density,
\[
 a(x)\le5\alpha.
\]
Choose $\alpha\le1/10$, and define
\begin{equation}\label{eq:log-affine-local-kernel}
 {\sf P}_{\rm loc}(x,dy)
 :=
 k(x,y)\,dy
 +
 \bigl(1-a(x)\bigr)\delta_x(dy).
\end{equation}
Thus ${\sf P}_{\rm loc}$ is a Markov kernel and
\[
 {\sf P}_{\rm loc}(x,\{x\})
 =
 1-a(x)\ge\frac12.
\]

\begin{lemma}[Reversibility, positivity, and comparison]
\label{lem:log-affine-local-kernel}
The kernel ${\sf P}_{\rm loc}$ is reversible with stationary law $\nu_b$, has holding probability at least $1/2$, and is positive.  Moreover,
\begin{equation}\label{eq:log-affine-dirichlet-domination}
 \mathcal E_{{\sf P}_b}(g,g)
 \ge
 \mathcal E_{{\sf P}_{\rm loc}}(g,g)
 \qquad(g\in L^2(\nu_b)).
\end{equation}
\end{lemma}

\begin{proof}
The off-diagonal stationary measure is symmetric, since $q_j(y-x)=q_j(x-y)$ and
\[
 f_b(x)\min\left\{1,\frac{f_b(y)}{f_b(x)}\right\}
 =
 \min\{f_b(x),f_b(y)\}.
\]
The diagonal part of
\eqref{eq:log-affine-local-kernel} is supported on
$\{(x,x):x\in K\}$ and is therefore symmetric as well.  Hence
${\sf P}_{\rm loc}$ is reversible with respect to $\nu_b$.

For the positivity, define
\begin{align*}
    {\sf Q}(x,dy):=2k(x,y)dy+(1-2a(x))\delta_x(dy).
\end{align*}
Since $a(x)\leq 1/2$, ${\sf Q}$ is a Markov kernel. And it is reversible by the same symmetry argument. Thus, its spectrum lies in $[-1,1]$; therefore ${\sf P}_{\rm loc}=(\Id+{\sf Q})/2$ is positive.

For the comparison, let $\sigma_{n-1}$ denote normalized surface measure on $\mathbb S^{n-1}$ and
\[
 \eta_j(t)
 :=
 \frac{n|t|^{n-1}}{2\delta_j^n}
 \one_{\{|t|\le\delta_j\}}.
\]
The signed polar-coordinate identity
\begin{equation}\label{eq:signed-polar-integral}
 \int_{\R^n}\varphi(v)q_j(v)\,dv
 =
 \int_{\mathbb S^{n-1}}\int_{\R}
 \varphi(t\theta)\eta_j(t)\,dt\,
 d\sigma_{n-1}(\theta)
\end{equation}
holds for every nonnegative measurable $\varphi$.

For a fixed direction $\theta\in\mathbb S^{n-1}$, let
${\sf P}_\theta$ be the following lazy Metropolis kernel.  From
$x\in K$, put
\[
 N(x):=\#\{j:x\in D_j\}\le5.
\]
Choose an index $J$ according to
\[
 \Prb\{J=j\}
 =
 \alpha\one_{D_j}(x),
 \qquad
 \Prb\{J=\bot\}
 =
 1-\alpha N(x).
\]
If $J=\bot$, stay at $x$.  Otherwise, draw $T$ with density
$\eta_J$, set $y=x+T\theta$, and move to $y$ with probability
\[
 \one_{D_J}(y)
 \min\left\{1,\frac{f_b(y)}{f_b(x)}\right\};
\]
if this move is rejected, stay at $x$.

The signed polar identity implies
\begin{equation}\label{eq:local-kernel-direction-average}
 {\sf P}_{\rm loc}
 =
 \int_{\mathbb S^{n-1}}
 {\sf P}_\theta\,d\sigma_{n-1}(\theta).
\end{equation}
For fixed $\theta$, the kernel ${\sf P}_\theta$ is reversible with
respect to $\nu_b$.  Indeed, a proposed move $x\mapsto y=x+t\theta$
has the same radial density as the reverse move
$y\mapsto x=y-t\theta$, and
\[
 f_b(x)\min\left\{1,\frac{f_b(y)}{f_b(x)}\right\}
 =
 \min\{f_b(x),f_b(y)\}.
\]
Thus the stationary flow of every off-diagonal move equals that of its reverse. Moreover, ${\sf P}_\theta$ has holding probability at least $1-5\alpha \geq 1/2$.
Hence, ${\sf P}_\theta$ is positive.

The remaining proof follows from the positive-kernel comparison principle underlying~\cite[Lemma~8]{RudolfUllrichComparison}. Let $\Pi_{\theta^\perp}$ be orthogonal projection onto $\theta^\perp$, and define
\[
 {\sf U}_\theta g
 :=
 \E_{\nu_b}[g(X)\mid\Pi_{\theta^\perp} X].
\]
Thus ${\sf U}_\theta$ is the orthogonal projection onto functions that
are constant on each line parallel to $\theta$.  Since
${\sf P}_\theta$ moves only along such lines, it fixes every function
in the range of ${\sf U}_\theta$.  Hence,
\[
 {\sf P}_\theta{\sf U}_\theta={\sf U}_\theta.
\]
By self-adjointness, ${\sf U}_\theta{\sf P}_\theta={\sf U}_\theta$,  and consequently,
\begin{equation}\label{eq:directional-operator-order}
 {\sf P}_\theta-{\sf U}_\theta
 =
 (\Id-{\sf U}_\theta)
 {\sf P}_\theta
 (\Id-{\sf U}_\theta)
 \succeq0.
\end{equation}
It follows that
\[
 \mathcal E_{{\sf P}_\theta}(g,g)
 \le
 \mathcal E_{{\sf U}_\theta}(g,g).
\]

Finally, hit-and-run is the average of complete directional refreshes:
\[
 {\sf P}_b
 =
 \int_{\mathbb S^{n-1}}
 {\sf U}_\theta\,d\sigma_{n-1}(\theta).
\]
Averaging the preceding Dirichlet-form inequality over $\theta$ and
using \eqref{eq:local-kernel-direction-average} gives
\[
 \mathcal E_{{\sf P}_{\rm loc}}(g,g)
 \le
 \mathcal E_{{\sf P}_b}(g,g).
\]
\end{proof}

\begin{lemma}[Local overlap of the comparison kernel]
\label{lem:log-affine-local-overlap}
There are universal constants $c,\alpha_0>0$ such that
\begin{equation}\label{eq:log-affine-local-overlap}
 d_{\rho_b}(x,y)<\frac{c}{\sqrt n}
 \quad\Longrightarrow\quad
 \dtv\bigl({\sf P}_{\rm loc}(x,\cdot),{\sf P}_{\rm loc}(y,\cdot)\bigr)
 \le1-\alpha_0.
\end{equation}
\end{lemma}

\begin{proof}
Lemma~\ref{lem:log-affine-intrinsic-local-geometry} gives
$\rho_b(y)\in[\rho_b(x)/2,2\rho_b(x)]$ and
$\|x-y\|_2\le c\rho_b(x)/\sqrt n$.  Let
\[
 \rho_-:=\min\{\rho_b(x),\rho_b(y)\},
 \qquad
 \rho_+:=\max\{\rho_b(x),\rho_b(y)\}.
\]
Since $\rho_+/\rho_-\le2$, the interval
$[\rho_+/32,\rho_-/8]$ contains at least one dyadic number.  Choose $j$ so that
\begin{equation}\label{eq:log-affine-overlap-scale}
 \frac{\rho_+}{32}\le\delta_j\le\frac{\rho_-}{8}.
\end{equation}

Take any $x,y\in D_j$.  For $z\in \{x,y\}$, let 
\begin{align*}
    L_z:=\{w:f_b(w)\ge\beta f_b(z)\}.
\end{align*}
The commonly accepted region of the two kernels is
\[
 G:=B(x,\delta_j)\cap B(y,\delta_j)\cap D_j\cap L_x\cap L_y.
\]
Our goal is to show that both ${\sf P}_{\rm loc}(x,\cdot)$ and ${\sf P}_{\rm loc}(y,\cdot)$ have large mass on $G$.

We first bound the volume of the overlap $B(x,\delta_j)\cap B(y,\delta_j)$. Let $e:=\frac{y-x}{\|y-x\|_2}$.  We integrate along lines parallel to $e$, indexed by $z\in e^\perp$.  For $\|z\|_2\le\delta_j$, the sections of $B(x,\delta_j)$ and $B(y,\delta_j)$ on the line $z+\R e$ are intervals of the same length, translated from one another by $d$.  Hence, the length lost from their intersection is at most $d$.  Fubini's theorem
therefore gives
\[
 v_n\delta_j^n
 -\vol\bigl(B(x,\delta_j)\cap B(y,\delta_j)\bigr)\le
 d\,v_{n-1}\delta_j^{n-1}.
\]
Since $v_{n-1}/v_n\leq O(\sqrt n)$,
\[
 \frac{\vol(B(x,\delta_j)\cap B(y,\delta_j))}
 {v_n\delta_j^n}
 \ge
 1-O(\sqrt n)\frac{\|x-y\|_2}{\delta_j}\geq 1-O\left(\frac{c\rho_b(x)}{\delta_j}\right)\geq 1-O(c).
\]
For sufficiently small $c$, the normalized overlap is at least $9/10$.

We next control the losses caused by the band and density conditions. Fix $z\in\{x,y\}$ and let $V\sim\Unif(B(0,\delta_j))$.  Since
\[
 2\delta_j
 \le\frac{\rho_-}{4}
 <\rho_b(z),
\]
monotonicity of
$r\mapsto\lambda_b(z,r)$ and the definition of $\rho_b(z)$ give
\[
 \lambda_b(z,\delta_j)\ge \frac{63}{64},
 \qquad
 \lambda_b(z,2\delta_j)\ge\frac{63}{64}.
\]
The random variables $z-V$ and $z+2V$ are uniform in
$B(z,\delta_j)$ and $B(z,2\delta_j)$, respectively.  Dropping the
density condition from the definition of $\lambda_b$, we obtain
\[
 \Prb\{z-V\notin K\}\le \frac{1}{64},
 \qquad
 \Prb\{z+2V\notin K\}\le \frac{1}{64}.
\]
By the union bound, with probability at least $31/32$, both points belong to $K$.

On this event, by the concavity and non-negativity of $\rho_b$, we have
\begin{align*}
    \rho_b(z+V)\geq \frac{\rho_b(z)+\rho_b(z+2V)}{2}\geq \frac{1}{2}\rho_b(z)\,\qquad \rho_b(z)\geq \frac{\rho_b(z-V)+\rho_b(z+V)}{2}\geq \frac{1}{2}\rho_b(z+V).
\end{align*}
That is,
\[
 \frac12\rho_b(z)
 \le
 \rho_b(z+V)
 \le
 2\rho_b(z).
\]
The choice of $\delta_j$ gives, for either $z=x$ or $z=y$,
\[
 4\delta_j
 \le\frac{1}{2}\rho_b(z)\leq \rho_b(z+V)\leq 2\rho_b(z)
 \le64\delta_j.
\]
Consequently, on the preceding event,
\[
 z+V\in D_j.
\]
Since $z+V$ is uniform in $B(z,\delta_j)$, this proves
\begin{equation}\label{eq:band-retention}
 \vol(B(z,\delta_j)\setminus D_j)
 \le
 \frac{1}{32} v_n\delta_j^n,\qquad z\in \{x,y\}.
\end{equation}

Because $\delta_j<\rho_b(z)$,
\[
 \lambda_b(z,\delta_j)
 =
 \frac{\vol(B(z,\delta_j)\cap L_z)}{v_n \delta_j^n}
 \ge\frac{63}{64}.
\]
Hence
\begin{equation}\label{eq:density-retention}
 \vol(B(z,\delta_j)\setminus L_z)
 \le \frac{1}{64}v_n \delta_j^n.
\end{equation}

Combining them together, we have 
\begin{align*}
 \vol(G)
 &\ge
 \vol(B(x,\delta_j)\cap B(y,\delta_j))
 -\vol(B(x,\delta_j)\setminus D_j)
 -\vol(B(x,\delta_j)\setminus L_x)
 -\vol(B(y,\delta_j)\setminus L_y)\\
 &\ge
 \left(
  \frac9{10}-\frac{1}{32} - \frac{2}{64}
 \right)v_n\delta_j^n = \frac{67}{80}v_n\delta_j^n.
\end{align*}

For every $w\in G$, we have
\[
 x,y,w\in D_j,
 \qquad
 \|w-x\|_2,\|w-y\|_2\le\delta_j,\qquad
 f_b(w)\ge\beta f_b(x),
 \qquad
 f_b(w)\ge\beta f_b(y).
\]
Therefore, the $j$th term of the transition density satisfies
\[
 k(x,w)\ge\frac{\alpha\beta}{v_n\delta_j^n},
 \qquad
 k(y,w)\ge\frac{\alpha\beta}{v_n\delta_j^n}.
\]
Thus, the measure
\[
 d\xi(w)
 :=
 \frac{\alpha\beta}{v_n\delta_j^n}\one_G(w)\,dw
\]
is dominated by both
${\sf P}_{\rm loc}(x,\cdot)$ and
${\sf P}_{\rm loc}(y,\cdot)$, and
\[
 \xi(K)
 =
 \frac{\alpha\beta}{v_n\delta_j^n}\vol(G)
 \ge
  \frac{67}{80}\alpha\beta.
\]
Using the fact that two probability measures that dominate a common subprobability
measure of mass $a$ have total-variation distance at most $1-a$, we conclude that
\[
 \dtv\bigl(
 {\sf P}_{\rm loc}(x,\cdot),
 {\sf P}_{\rm loc}(y,\cdot)
 \bigr)
 \le
 1-\frac{67}{80}\alpha\beta.
\]
The result follows with
$\alpha_0:=\frac{67}{80}\alpha\beta$.
\end{proof}

\subsection{Conductance and warm-start mixing}
\begin{proposition}[Conductance of log-affine hit-and-run]
\label{prop:log-affine-conductance}
For every $0<s<1/4$,
\begin{equation}\label{eq:log-affine-conductance-proposition}
 \Phi_s({\sf P}_b)
 \gtrsim \frac{\sigma_{\min}}{\sigma_{\max}}
 \frac{\psi_n}{ n\bigl(1+\log(2/s)\bigr)}.
\end{equation}
\end{proposition}

\begin{proof}
For sets whose smaller side has mass at least $s/2$, Corollary~\ref{cor:log-affine-radius-profile} gives weighted expansion with
\begin{align*}
    \cP_{\rho_b,\nu_b}(A)
 \gtrsim \frac{\psi_n}{\sqrt n}
  \frac{\sigma_{\min}}
       {\sigma_{\max}}
 \frac{p}{1+\log(1/p)}\geq \gamma_s p,\qquad \gamma_s:=\Theta\left(\frac{\psi_n}{\sqrt n} \frac{\sigma_{\min}}
       {\sigma_{\max}} \frac{1}{1+\log(2/s)}\right)
 \le1.
\end{align*}
Apply Proposition~\ref{prop:conductance-transfer} to $\lambda=\nu_b$, $w=\rho_b$, and $\mathsf P={\sf P}_{\rm loc}$, using Lemma~\ref{lem:log-affine-local-overlap} with $r_0=c/\sqrt n$ and universal $\alpha_0$.  This gives that 
\begin{align*}
    \Phi_s({\sf P}_{\rm loc})\gtrsim \frac{\gamma_s}{\sqrt{n}}.
\end{align*}
For every measurable $S$, by Lemma~\ref{lem:log-affine-local-kernel},
\[
 \cQ_{{\sf P}_b,\nu_b}(S,S^c)
 =\mathcal E_{{\sf P}_b}(\one_S,\one_S)
 \ge\mathcal E_{{\sf P}_{\rm loc}}(\one_S,\one_S)
 =\cQ_{{\sf P}_{\rm loc},\nu_b}(S,S^c)
\]
which transfers the conductance bound to ${\sf P}_b$.
\end{proof}

\begin{proof}[Proof of Theorem~\ref{thm:log-affine-main}]
The conductance statement is Proposition~\ref{prop:log-affine-conductance}.  Positivity of ${\sf P}_b$ was proved at the beginning of this appendix. Then, applying Lemma~\ref{lem:LS} with the stationary measure  $\nu_b$ and $s=\varepsilon/(2M)$ completes the proof of the theorem.
\end{proof}

\bibliographystyle{alpha}
\bibliography{refs}

@article{kook2026spectral,
  title={Spectral Gaps of Hit-and-Run and Coordinate Hit-and-Run},
  author={Kook, Yunbum and Vempala, Santosh S},
  journal={arXiv preprint arXiv:2608.16878},
  year={2026}
}

@book{AFP,
  author={Ambrosio, Luigi and Fusco, Nicola and Pallara, Diego},
  title={Functions of Bounded Variation and Free Discontinuity Problems},
  series={Oxford Mathematical Monographs},
  publisher={The Clarendon Press, Oxford University Press},
  address={New York},
  year={2000}
}

@article{Berwald,
  title={Verallgemeinerung eines Mittelwertsatzes von J. Favard f{\"u}r positive konkave Funktionen},
  author={Berwald, Ludwig},
  journal={Acta Mathematica},
  volume={79},
  number={1},
  pages={17--37},
  year={1947},
  publisher={Springer}
}

@article{CE,
  title={Hit-and-run mixing via localization schemes},
  author={Chen, Yuansi and Eldan, Ronen},
  journal={Discrete \& Computational Geometry},
  volume={75},
  number={3},
  pages={747--794},
  year={2026},
  publisher={Springer},
  doi={10.1007/s00454-025-00808-4}
}

@article{KLS95,
  title={Isoperimetric problems for convex bodies and a localization lemma},
  author={Kannan, Ravi and Lov{\'a}sz, L{\'a}szl{\'o} and Simonovits, Mikl{\'o}s},
  journal={Discrete \& Computational Geometry},
  volume={13},
  number={3},
  pages={541--559},
  year={1995},
  publisher={Springer}
}

@article{KlartagNeedles, 
 title={Needle Decompositions in Riemannian Geometry}, 
 volume={249}, 
 ISSN={0065-9266}, 
 url={http://dx.doi.org/10.1090/memo/1180}, 
 DOI={10.1090/memo/1180}, 
 number={1180}, 
 journal={Memoirs of the American Mathematical Society}, 
 publisher={American Mathematical Society (AMS)}, 
 author={Klartag, Bo'az}, 
 year={2017}
}

@article{MaiRigidity,
  author={Mai, Cong Hung},
  title={Rigidity for the Isoperimetric Inequality of Negative Effective Dimension on Weighted Riemannian Manifolds},
  journal={Geometriae Dedicata},
  volume={202},
  pages={213--232},
  year={2019},
  doi={10.1007/s10711-018-0410-x}
}

@article{MaiOhta,
  author={Mai, Cong Hung and Ohta, Shin-ichi},
  title={Quantitative Estimates for the {Bakry--Ledoux} Isoperimetric Inequality},
  journal={Commentarii Mathematici Helvetici},
  volume={96},
  number={4},
  pages={693--739},
  year={2021},
  doi={10.4171/CMH/523}
}

@article{BertrandFathi,
  author={Bertrand, J{\'e}r{\^o}me and Fathi, Max},
  title={Stability of Eigenvalues and Observable Diameter in {RCD}(1,$\infty$) Spaces},
  journal={The Journal of Geometric Analysis},
  volume={32},
  number={11},
  pages={Paper No.~270},
  year={2022},
  doi={10.1007/s12220-022-00999-9}
}

@article{LangharstPutterman,
    author = "Dylan Langharst and Eli Putterman",
     title = "Weighted Berwald's inequality",
   journal = "Indiana Univ. Math. J.",
  fjournal = "Indiana University Mathematics Journal",
    volume = 74,
      year = 2025,
     issue = 1,
     pages = "47--90",
      issn = "0022-2518",
     coden = "IUMJAB",
   mrclass = "52A39, 52A41, 28A75",
}

@article{LovaszHR,
  title={Hit-and-run mixes fast},
  author={Lov{\'a}sz, L{\'a}szl{\'o}},
  journal={Mathematical programming},
  volume={86},
  number={3},
  pages={443--461},
  year={1999},
  publisher={Springer}
}

@article{LS,
  title={Random walks in a convex body and an improved volume algorithm},
  author={Lov{\'a}sz, L{\'a}szl{\'o} and Simonovits, Mikl{\'o}s},
  journal={Random structures \& algorithms},
  volume={4},
  number={4},
  pages={359--412},
  year={1993},
  publisher={Wiley Online Library}
}

@inproceedings{LS90,
  author       = {L{\'{a}}szl{\'{o}} Lov{\'{a}}sz and
                  Mikl{\'{o}}s Simonovits},
  title        = {The Mixing Rate of Markov Chains, an Isoperimetric Inequality, and
                  Computing the Volume},
  booktitle    = {31st Annual Symposium on Foundations of Computer Science, St. Louis,
                  Missouri, USA, October 22-24, 1990, Volume {I}},
  pages        = {346--354},
  publisher    = {{IEEE} Computer Society},
  year         = {1990},
  url          = {https://doi.org/10.1109/FSCS.1990.89553},
  doi          = {10.1109/FSCS.1990.89553}
}

@article{LV,
author = {Lov{\'a}sz, L{\'a}szl{\'o} and Vempala, Santosh},
title = {Hit-and-Run from a Corner},
journal = {SIAM Journal on Computing},
volume = {35},
number = {4},
pages = {985-1005},
year = {2006},
doi = {10.1137/S009753970544727X}
}

@article{LVlogconcave,
  title={The geometry of logconcave functions and sampling algorithms},
  author={Lov{\'a}sz, L{\'a}szl{\'o} and Vempala, Santosh},
  journal={Random Structures \& Algorithms},
  volume={30},
  number={3},
  pages={307--358},
  year={2007},
  publisher={Wiley Online Library}
}

@article{Smith,
  title={Efficient Monte Carlo procedures for generating points uniformly distributed over bounded regions},
  author={Smith, Robert L},
  journal={Operations Research},
  volume={32},
  number={6},
  pages={1296--1308},
  year={1984},
  publisher={INFORMS}
}

@incollection{VempalaSurvey,
  author={Vempala, Santosh S.},
  title={Geometric random walks: a survey},
  booktitle={Combinatorial and Computational Geometry},
  editor={Goodman, Jacob E. and Pach, J{\'a}nos and Welzl, Emo},
  series={Mathematical Sciences Research Institute Publications},
  volume={52},
  pages={573--612},
  publisher={Cambridge University Press},
  address={Cambridge},
  year={2005}
}

@article{KlartagKLS,
  author={Klartag, Bo'az},
  title={Logarithmic bounds for isoperimetry and slices of convex sets},
  journal={Ars Inveniendi Analytica},
  year={2023},
  pages={Paper No.~4, 17 pp.},
  doi={10.15781/jsjy-0b06}
}

@article{LetwinKLS,
  title={The {KLS} constant is {$O(\log^{1/4} n)$}},
  author={Letwin, Brayden},
  journal={arXiv preprint arXiv:2607.24164},
  year={2026}
}

@article{DFK91,
  author={Dyer, Martin E. and Frieze, Alan M. and Kannan, Ravi},
  title={A Random Polynomial-Time Algorithm for Approximating the Volume of Convex Bodies},
  journal={Journal of the ACM},
  volume={38},
  number={1},
  pages={1--17},
  year={1991},
  doi={10.1145/102782.102783}
}

@article{KLS97,
  author={Kannan, Ravi and Lov{\'a}sz, L{\'a}szl{\'o} and Simonovits, Mikl{\'o}s},
  title={Random Walks and an {$O^*(n^5)$} Volume Algorithm for Convex Bodies},
  journal={Random Structures \& Algorithms},
  volume={11},
  number={1},
  pages={1--50},
  year={1997}
}

@article{LVVolume,
  author={Lov{\'a}sz, L{\'a}szl{\'o} and Vempala, Santosh},
  title={Simulated Annealing in Convex Bodies and an {$O^*(n^4)$} Volume Algorithm},
  journal={Journal of Computer and System Sciences},
  volume={72},
  number={2},
  pages={392--417},
  year={2006},
  doi={10.1016/j.jcss.2005.08.004}
}

@article{CV18,
  author={Cousins, Ben and Vempala, Santosh S.},
  title={Gaussian Cooling and {$O^*(n^3)$} Algorithms for Volume and Gaussian Volume},
  journal={SIAM Journal on Computing},
  volume={47},
  number={3},
  pages={1237--1273},
  year={2018},
  doi={10.1137/15M1054250}
}

@article{EldanSL,
  author={Eldan, Ronen},
  title={Thin Shell Implies Spectral Gap Up to Polylog via a Stochastic Localization Scheme},
  journal={Geometric and Functional Analysis},
  volume={23},
  number={2},
  pages={532--569},
  year={2013},
  doi={10.1007/s00039-013-0214-y}
}

@article{LeeVempalaKLS,
  author={Lee, Yin Tat and Vempala, Santosh S.},
  title={Eldan's Stochastic Localization and the {KLS} Conjecture: Isoperimetry, Concentration and Mixing},
  journal={Annals of Mathematics},
  volume={199},
  number={3},
  pages={1043--1092},
  year={2024},
  doi={10.4007/annals.2024.199.3.2}
}

@article{ChenKLS,
  author={Chen, Yuansi},
  title={An Almost Constant Lower Bound of the Isoperimetric Coefficient in the {KLS} Conjecture},
  journal={Geometric and Functional Analysis},
  volume={31},
  number={1},
  pages={34--61},
  year={2021},
  doi={10.1007/s00039-021-00558-4}
}

@article{KlartagLehecKLS,
  author={Klartag, Bo'az and Lehec, Joseph},
  title={Bourgain's Slicing Problem and {KLS} Isoperimetry Up to Polylog},
  journal={Geometric and Functional Analysis},
  volume={32},
  number={5},
  pages={1134--1159},
  year={2022},
  doi={10.1007/s00039-022-00612-9}
}

@misc{GuanSlicing,
  author={Guan, Qingyang},
  title={A Note on {Bourgain}'s Slicing Problem},
  year={2024},
  eprint={2412.09075},
  archivePrefix={arXiv},
  primaryClass={math.MG},
  note={arXiv:2412.09075}
}

@article{KlartagLehecSlicing,
  author={Klartag, Bo'az and Lehec, Joseph},
  title={Affirmative Resolution of {Bourgain}'s Slicing Problem Using {Guan}'s Bound},
  journal={Geometric and Functional Analysis},
  volume={35},
  number={4},
  pages={1147--1168},
  year={2025},
  doi={10.1007/s00039-025-00718-w}
}

@incollection{EldanKlartagThinSlice,
  author={Eldan, Ronen and Klartag, Bo'az},
  title={Approximately Gaussian Marginals and the Hyperplane Conjecture},
  booktitle={Concentration, Functional Inequalities and Isoperimetry},
  series={Contemporary Mathematics},
  volume={545},
  pages={55--68},
  publisher={American Mathematical Society},
  year={2011},
  doi={10.1090/conm/545/10764}
}

@misc{BizeulSlicing,
  author={Bizeul, Pierre},
  title={The Slicing Conjecture via Small Ball Estimates},
  year={2025},
  eprint={2501.06854},
  archivePrefix={arXiv},
  primaryClass={math.FA},
  note={arXiv:2501.06854}
}

@misc{KlartagLehecThinShell,
  author={Klartag, Bo'az and Lehec, Joseph},
  title={Thin-Shell Bounds via Parallel Coupling},
  year={2025},
  eprint={2507.15495},
  archivePrefix={arXiv},
  primaryClass={math.PR},
  note={arXiv:2507.15495}
}

@misc{ChenKlartagThinShell,
  author={Chen, Yuansi and Klartag, Bo'az},
  title={Digesting the Proof of the Sharp Thin-Shell Inequality},
  year={2026},
  eprint={2607.23307},
  archivePrefix={arXiv},
  primaryClass={math.MG},
  note={arXiv:2607.23307}
}

@article{JLLV26,
  author={Jia, He and Laddha, Aditi and Lee, Yin Tat and Vempala, Santosh S.},
  title={Reducing Isotropy and Volume to {KLS}: Faster Rounding and Volume Algorithms},
  journal={Journal of the ACM},
  volume={73},
  number={2},
  pages={8:1--8:21},
  year={2026},
  doi={10.1145/3795687}
}

@article{RudolfUllrichPositivity,
  author={Rudolf, Daniel and Ullrich, Mario},
  title={Positivity of {Hit-and-Run} and Related Algorithms},
  journal={Electronic Communications in Probability},
  volume={18},
  number={49},
  pages={1--8},
  year={2013},
  doi={10.1214/ECP.v18-2507}
}

@article{RudolfUllrichComparison,
  author={Rudolf, Daniel and Ullrich, Mario},
  title={Comparison of {Hit-and-Run}, Slice Sampler and Random Walk Metropolis},
  journal={Journal of Applied Probability},
  volume={55},
  number={4},
  pages={1186--1202},
  year={2018},
  doi={10.1017/jpr.2018.78}
}

@article{ShiTianZhangSL,
  title={Perspectives on Stochastic Localization},
  author={Shi, Bobby and Tian, Kevin and Zhang, Matthew S},
  journal={arXiv preprint arXiv:2510.04460},
  year={2025}
}

@incollection{Cheeger,
  author={Cheeger, Jeff},
  title={A Lower Bound for the Smallest Eigenvalue of the Laplacian},
  booktitle={Problems in Analysis},
  editor={Gunning, Robert C.},
  pages={195--199},
  publisher={Princeton University Press},
  address={Princeton, NJ},
  year={1970}
}

@article{anttila2003central,
  title={The central limit problem for convex bodies},
  author={Anttila, Milla and Ball, Keith and Perissinaki, Irini},
  journal={Transactions of the American Mathematical Society},
  volume={355},
  number={12},
  pages={4723--4735},
  year={2003}
}

@inproceedings{bobkov2004central,
  title={On the central limit property of convex bodies},
  author={Bobkov, Sergey G and Koldobsky, Alexander},
  booktitle={Geometric Aspects of Functional Analysis: Israel Seminar 2001-2002},
  pages={44--52},
  year={2004},
  organization={Springer}
}

@inproceedings{KookZhangRenyi,
  author={Kook, Yunbum and Zhang, Matthew S.},
  title={R\'enyi-Infinity Constrained Sampling with {$d^3$} Membership Queries},
  booktitle={Proceedings of the 2025 Annual ACM-SIAM Symposium on Discrete Algorithms},
  pages={5278--5306},
  year={2025},
  publisher={SIAM},
  doi={10.1137/1.9781611978322.181}
}

@inproceedings{LeeShenTian21,
  author={Lee, Yin Tat and Shen, Ruoqi and Tian, Kevin},
  title={Structured Logconcave Sampling with a {Restricted Gaussian Oracle}},
  booktitle={Proceedings of the Thirty-Fourth Conference on Learning Theory},
  series={Proceedings of Machine Learning Research},
  volume={134},
  pages={2993--3050},
  year={2021},
  publisher={PMLR},
  url={https://proceedings.mlr.press/v134/lee21a.html}
}

@article{KannanNarayananDikin,
  author={Kannan, Ravindran and Narayanan, Hariharan},
  title={Random Walks on Polytopes and an Affine Interior Point Method for Linear Programming},
  journal={Mathematics of Operations Research},
  volume={37},
  number={1},
  pages={1--20},
  year={2012},
  doi={10.1287/moor.1110.0519}
}

@article{ChenDwivediWainwrightYu,
  author={Chen, Yuansi and Dwivedi, Raaz and Wainwright, Martin J. and Yu, Bin},
  title={Fast {MCMC} Sampling Algorithms on Polytopes},
  journal={Journal of Machine Learning Research},
  volume={19},
  number={55},
  pages={1--86},
  year={2018},
  url={https://jmlr.org/papers/v19/18-158.html}
}

@InProceedings{RHMCLewis,
  title = 	 {Sampling Polytopes with Riemannian HMC: Faster Mixing via the Lewis Weights Barrier},
  author =       {Gatmiry, Khashayar and Kelner, Jonathan and Vempala, Santosh S.},
  booktitle = 	 {Proceedings of Thirty Seventh Conference on Learning Theory},
  pages = 	 {1796--1881},
  year = 	 {2024},
  editor = 	 {Agrawal, Shipra and Roth, Aaron},
  volume = 	 {247},
  series = 	 {Proceedings of Machine Learning Research},
  month = 	 {30 Jun--03 Jul},
  publisher =    {PMLR}
}

@inproceedings{kook2024gaussian,
  title={Gaussian cooling and Dikin walks: The interior-point method for logconcave sampling},
  author={Kook, Yunbum and Vempala, Santosh S},
  booktitle={The Thirty Seventh Annual Conference on Learning Theory},
  pages={3137--3240},
  year={2024},
  organization={PMLR}
}

@misc{kook2026d25mixingbounddikinwalks,
      title={Beyond the $d^{2.5}$-mixing bound for Dikin walks on polytopes}, 
      author={Yunbum Kook},
      year={2026},
      eprint={2607.13943},
      archivePrefix={arXiv},
      primaryClass={cs.DS},
      url={https://arxiv.org/abs/2607.13943}, 
}

@inproceedings{LeeVempalaRHMC,
  author={Lee, Yin Tat and Vempala, Santosh S.},
  title={Convergence Rate of {Riemannian Hamiltonian Monte Carlo} and Faster Polytope Volume Computation},
  booktitle={Proceedings of the 50th Annual ACM SIGACT Symposium on Theory of Computing},
  pages={1115--1121},
  year={2018},
  publisher={ACM},
  doi={10.1145/3188745.3188774}
}

@inproceedings{LaddhaLeeVempala,
  author={Laddha, Aditi and Lee, Yin Tat and Vempala, Santosh S.},
  title={Strong Self-Concordance and Sampling},
  booktitle={Proceedings of the 52nd Annual ACM SIGACT Symposium on Theory of Computing},
  pages={1212--1222},
  year={2020},
  publisher={ACM},
  doi={10.1145/3357713.3384272}
}

@inproceedings{ChenChewiSalimWibisono22,
  author={Chen, Yongxin and Chewi, Sinho and Salim, Adil and Wibisono, Andre},
  title={Improved Analysis for a Proximal Algorithm for Sampling},
  booktitle={Proceedings of the Thirty-Fifth Conference on Learning Theory},
  series={Proceedings of Machine Learning Research},
  volume={178},
  pages={2984--3014},
  year={2022},
  publisher={PMLR},
  url={https://proceedings.mlr.press/v178/chen22c.html}
}

@article{bertsimas2004solving,
  title={Solving convex programs by random walks},
  author={Bertsimas, Dimitris and Vempala, Santosh},
  journal={Journal of the ACM (JACM)},
  volume={51},
  number={4},
  pages={540--556},
  year={2004},
  publisher={ACM New York, NY, USA}
}

@article{kalai2006simulated,
  title={Simulated annealing for convex optimization},
  author={Kalai, Adam Tauman and Vempala, Santosh},
  journal={Mathematics of Operations Research},
  volume={31},
  number={2},
  pages={253--266},
  year={2006},
  publisher={INFORMS}
}

@article{Holmes2006BayesianAV,
  title={Bayesian auxiliary variable models for binary and multinomial regression},
  author={Chris C. Holmes and Leonhard Held},
  journal={Bayesian Analysis},
  year={2006},
  volume={1},
  pages={145-168},
  url={https://api.semanticscholar.org/CorpusID:8209006}
}

@article{albert1993bayesian,
  title={Bayesian analysis of binary and polychotomous response data},
  author={Albert, James H and Chib, Siddhartha},
  journal={Journal of the American statistical Association},
  volume={88},
  number={422},
  pages={669--679},
  year={1993},
  publisher={Taylor \& Francis}
}

@inproceedings{KookVempalaCold,
  author={Kook, Yunbum and Vempala, Santosh S.},
  title={Faster Logconcave Sampling from a Cold Start in High Dimension},
  booktitle={2025 IEEE 66th Annual Symposium on Foundations of Computer Science},
  pages={997--1006},
  year={2025},
  publisher={IEEE},
  doi={10.1109/FOCS63196.2025.00052}
}

@article{KookVempalaZerothOrder,
  title={Zeroth-order logconcave sampling},
  author={Kook, Yunbum and Vempala, Santosh S},
  journal={arXiv preprint arXiv:2507.18021},
  year={2025}
}

@article{KookVempalaUnified,
  title={A unified complexity bound for logconcave sampling},
  author={Kook, Yunbum and Vempala, Santosh S},
  journal={arXiv preprint arXiv:2606.12694},
  year={2026}
}

@inproceedings{KookVempalaZhangInOut,
  author={Kook, Yunbum and Vempala, Santosh S. and Zhang, Matthew S.},
  title={In-and-Out: Algorithmic Diffusion for Sampling Convex Bodies},
  booktitle={Advances in Neural Information Processing Systems},
  volume={37},
  pages={108354--108388},
  year={2024},
  doi={10.52202/079017-3440}
}

@inproceedings{KookVempalaDiffusion,
  author={Kook, Yunbum and Vempala, Santosh S.},
  title={Sampling and Integration of Logconcave Functions by Algorithmic Diffusion},
  booktitle={Proceedings of the 57th Annual ACM Symposium on Theory of Computing},
  pages={924--932},
  year={2025},
  publisher={ACM},
  doi={10.1145/3717823.3718202}
}

@article{bourgain1986high,
  title={On high dimensional maximal functions associated to convex bodies},
  author={Bourgain, Jean},
  journal={American Journal of Mathematics},
  volume={108},
  number={6},
  pages={1467--1476},
  year={1986},
  publisher={JSTOR}
}

@article{buser1982note,
  author={Buser, Peter},
  title={A Note on the Isoperimetric Constant},
  journal={Annales Scientifiques de l'\'Ecole Normale Sup\'erieure},
  volume={15},
  number={2},
  pages={213--230},
  year={1982},
  doi={10.24033/asens.1426}
}

@article{ledoux2004spectral,
  author={Ledoux, Michel},
  title={Spectral Gap, Logarithmic {Sobolev} Constant, and Geometric Bounds},
  journal={Surveys in Differential Geometry},
  volume={9},
  pages={219--240},
  year={2004},
  doi={10.4310/SDG.2004.v9.n1.a6}
}
\end{document}